%% file: main.tex
\RequirePackage{amsmath}
\documentclass[runningheads]{llncs}

\usepackage[utf8]{inputenc}
\usepackage[T1]{fontenc}

\usepackage{amssymb}

\usepackage{amsmath}
\usepackage{stmaryrd}
\SetSymbolFont{stmry}{bold}{U}{stmry}{m}{n} 
\usepackage[ruled,vlined]{algorithm2e}
\usepackage{url}
\usepackage{paralist}
\usepackage[hidelinks]{hyperref}
\hypersetup{bookmarksdepth=3,bookmarksopen,bookmarksnumbered} 
\usepackage[numbers,sort&compress]{natbib}
\usepackage{todonotes}
\usepackage{datetime2}
\usepackage{xstring}
\usepackage{subcaption}
\usepackage{catchfile}
\usepackage[capitalize]{cleveref}
\usepackage[frozencache]{minted}
\definecolor{codehighlight}{RGB}{220, 255, 220}
\setminted{autogobble, breaklines, frame=none, framesep=3mm,  
           mathescape=true, escapeinside=||, highlightcolor=codehighlight} 
\definecolor{operator}{RGB}{0, 102, 187}
\newcommand{\op}[1]{\textcolor{operator}{#1}}
\usepackage{xspace}
\usepackage{bm}

\usepackage{pgfplots}
\usepgfplotslibrary{groupplots}
\usepackage{pgfplotstable}
\pgfplotsset{compat=1.15}

\newcommand{\T}{\mathcal{T}}

\newcommand{\A}{\mathcal{A}}
\newcommand{\sem}[1]{\llbracket#1\rrbracket}
\newcommand{\ignore}[1]{}

\newcommand{\abs}[1]{\lvert#1\rvert}

\newcommand\restr[2]{{\left.\kern-\nulldelimiterspace#1\vphantom{|}\right|_{#2}}}

\input defs

\spnewtheorem*{lem}{Lemma}{\bfseries}{\itshape}
\spnewtheorem*{cor}{Corollary}{\bfseries}{\itshape}

\begin{document}

\title{Improving Thread-Modular\texorpdfstring{\\}{ }Abstract Interpretation}
%
%

\author{
    Michael Schwarz\inst{1}\and
    Simmo Saan\inst{2}\and
    Helmut Seidl\inst{1}\and
    Kalmer Apinis\inst{2}\and\texorpdfstring{\\}{}
    Julian Erhard\inst{1}\and
    Vesal Vojdani\inst{2}
}
\authorrunning{M. Schwarz et al.\ }
%
\institute{
    Technische Universit\"at M\"unchen, Garching, Germany\texorpdfstring{\\}{}
    \email{\{m.schwarz, helmut.seidl, julian.erhard\}@tum.de}\and
    University of Tartu, Tartu, Estonia\texorpdfstring{\\}{}
    \email{\{simmo.saan, kalmer.apinis, vesal.vojdani\}@ut.ee}
}

\maketitle              
\begin{abstract}
We give thread-modular non-relational value analyses as abstractions of a local trace semantics.
The semantics as well as the analyses are formulated by means of global invariants and
side-effecting constraint systems.
We show that a generalization of the analysis provided by the static analyzer \textsc{Goblint}
as well as a natural improvement of Antoine Min\'e's approach can be obtained as instances
of this general scheme. We show that these two analyses are incomparable w.r.t.\ precision
and provide a refinement which improves on both precision-wise.
%
We also report on a preliminary experimental comparison of the given analyses
on a meaningful suite of benchmarks.

\end{abstract}

\input{intro}
\input{side}
\input{traces/basics}
\input{traces/threads}
\input{traces/locks}
\input{traces/globals}
\input{traces/merge}
\input{traces/intuitive-example}
\input{traces/example-formalism}

\input{analyses/intro}
\input{analyses/protection-based}
\input{analyses/lock-centered}
\input{analyses/write-centered}
\input{analyses/combined}
\section{Proving the Analyses Sound}\label{s:correct}
\input{correctness/definitions}
\input{correctness/lock-centered}
\input{correctness/write-centered}
\input{correctness/protection-based}
\input{analyses/evaluation}
\input{conclusion}

\renewcommand{\doi}[1]{\textsc{doi}: \href{http://dx.doi.org/#1}{\nolinkurl{#1}}}
\bibliographystyle{splncs04nat}
\bibliography{lit}
\newpage
\appendix
\input{analyses/mine-appendix}
\end{document}

%% file: defs.tex
\newcommand{\init}{\textsf{init}}
\newcommand{\Let}{\textbf{let}}
\newcommand{\Iif}{\textbf{if}}
\newcommand{\Then}{\textbf{then}}
\newcommand{\Eelse}{\textbf{else}}
\newcommand{\In}{\textbf{in}}

\newcommand{\Get}{\eta}

\newcommand{\da}{\!\!\downarrow}
\newcommand{\create}{\textsf{create}}
\newcommand{\new}{\textsf{new}}

\newcommand{\lock}{\textsf{lock}}
\newcommand{\unlock}{\textsf{unlock}}

\newcommand{\self}{\textsf{self}}

\newcommand{\sink}{\textsf{sink}}
\newcommand{\last}{\textsf{last}}

\newcommand{\loc}{\textsf{loc}}
\newcommand{\id}{\textsf{id}}
\newcommand{\V}{{\cal V}}
\newcommand{\N}{{\cal N}}
\newcommand{\E}{{\cal E}}
\newcommand{\X}{{\cal X}}
\newcommand{\XX}{{\textsf X}}
\newcommand{\R}{{\cal R}}
\newcommand{\G}{{\cal G}}
\newcommand{\D}{{\cal D}}
\newcommand{\C}{{\cal C}}
\newcommand{\I}{{\cal I}}

\renewcommand{\C}{{\cal C}}
\newcommand{\M}{\textsf{M}}
\newcommand{\MM}{\mathcal{M}}
\newcommand{\UM}{\mathcal{U_M}}
\newcommand{\To}{{\Rightarrow}}
\newcommand{\LockCentered}{{2}}
\newcommand{\ProtectionBased}{{1}}
\newcommand{\ModPB}{{1}'}

\newcommand{\WriteCentered}{{3}}
\newcommand{\ModWC}{{3}'}
\newcommand{\ModWCL}{{3'}}

\newcommand{\lTlW}{\textsf{last\_tl\_write}}
\newcommand{\lTlL}{\textsf{last\_tl\_lock}}

\newcommand{\lW}{\textsf{last\_write}}

\newcommand{\minLSince}{\textsf{min\_lockset\_since}}

\newcommand{\eval}{\textsf{eval}}
\newcommand{\CWCU}{{'}} 
\newcommand{\GetPrimeP}{\textsf{split}[\eta]} 
\newcommand{\GetPrime}{\textsf{split}[\eta]'}

%% file: intro.tex
\section{Introduction}\label{s:intro}

In a series of papers starting around 2012, Antoine Min\'e and his co-authors developed methods
for abstract interpretation of concurrent systems \cite{Mine2012,Mine2014,Mine2016,Mine2017,Mine2018},
which can be considered the \emph{gold standard} for thread-modular static analysis of these
systems.
The core analysis from \cite{Mine2012} consists of a refinement of data flow which takes schedulability into
account by propagating values written before unlock operations to corresponding lock operations ---
provided that appropriate side-conditions are met.
Due to these side-conditions, more flows are generally excluded than in approaches
as, e.g., \cite{De2011,Mukherjee2017}.
An alternative approach, suggested by Vojdani~\cite{Vojdani2010,Vojdani2016},
is realized in the static analyzer \textsc{Goblint}.
This analysis is not based on data flows.
Instead, for each global $g$, a set of mutexes that definitely protect accesses to $g$ is determined.
Then \emph{side-effects} during the analysis of the threads' local states are used
to accumulate an abstraction of the set of all possibly written values.
This base approach then is enhanced by means of \emph{privatization}
to account for exclusive manipulations by individual threads.
This approach is similar to the thread-local shape analysis of \citet{Gotsman07},
which infers lock-invariants~\cite{OHearn07} by privatizing carved-out sections
of the heap owned by a thread.
Despite its conceptual simplicity and perhaps to our surprise, it turns out the Vojdani
style analysis is \emph{not} subsumed by Min\'e's approach but is incomparable.
Since Min\'e's analysis is more precise on many examples, we highlight only non-subsumption
in the other direction here.

\begin{example}\label{e:incomparable}
We use sets of integers for abstracting int values.
Consider the following concurrent program with global variable \texttt{g} and local variables \texttt{x} and \texttt{y}, and
assume here that $g$ is intialized to $0$:
\begin{center}
\begin{minipage}[t]{5cm}
\begin{minted}{c}
main:
   y = |\op{create}|(t1);
   |\op{lock}|(a);
   |\op{lock}|(b);
   x = g;
   ...

\end{minted}
\end{minipage}
\begin{minipage}[t]{4cm}
\begin{minted}{c}
t1:
   |\op{lock}|(a);
   |\op{lock}|(b);
   g = 42;
   |\op{unlock}|(a);
   g = 17;
   |\op{unlock}|(b);
\end{minted}
\end{minipage}
\end{center}
%
Program execution starts at program point \texttt{main} where, after creation
of another thread \texttt{t1} and locking of the mutexes \texttt{a} and \texttt{b}, the value of
the global \texttt{g} is read.
The created thread, on the other hand, also locks the mutexes \texttt{a} and \texttt{b}.
Then, it writes to \texttt{g} the two values 42 and 17 where mutex \texttt{a} is unlocked in-between the two writes,
and mutex \texttt{b} is unlocked only in the very end.

According to Min\'e's analysis, the value $\{42\}$ is merged into the local state at the operation
\textsf{lock}(\texttt{a}), while the value $\{17\}$ is merged at the operation \texttt{lock(b)}.
Thus, the local \texttt{x} receives the value $\{0,17,42\}$.

Vojdani's analysis, on the other hand, finds out that all accesses to \texttt{g} are protected by
the mutex \texttt{b}. Unlocking of \texttt{a}, therefore, does not publish the intermediately written
value $\{42\}$, but only the final value $\{17\}$ at \texttt{unlock(b)} is published.
Therefore, the local \texttt{x} only receives the value $\{0,17\}$.
\qed
\end{example}
The goal of this paper is to better understand
this intriguing incomparability and develop precision improvements to refine
these analyses. We concentrate only on the basic setting of \emph{non-relational}
analysis and a concurrent setting without precise thread \emph{id}s. We also ignore add-ons such as
thread priorities or
effects of weak memory, which are of major concern in \cite{Ferrara08,Alglave2011,Mine2016,Mine2018}.
As a common
framework for the comparison, we use \emph{side-effecting} constraint systems~\cite{apinis2012side}.
Constraint systems with side-effects extend ordinary constraint systems in that during
the evaluation of the right-hand side of one unknown, contributions to other unknowns
may also be triggered. This kind of formalism allows combining flow- and context-sensitive analysis
of the local state with flow- and context-insensitive analysis of globals.
Within the analyzer \textsc{Goblint}, this has been applied to the analysis
of multi-threaded systems~\cite{Vojdani2010,Vojdani2016}.
While in \textsc{Goblint} a single unknown is introduced per global,
we show how to express Min\'e's analysis in this formalism using
multiple unknowns per global.

To prove the given thread-modular analyses correct, we rely on a
trace semantics of the concurrent system.
Here, we insist on maintaining the \emph{local views}
of executing threads (\emph{ego threads}) only. The idea of tracking the events possibly affecting a particular
local thread configuration goes back to \cite{Lamport1978} (see also \cite{DistrAlgos}),
and is also used extensively for the verification of concurrent systems via separation logic
~\cite{Brookes2007,Nanevski2014,Nanevski2015, Nanevski2019}.
%
%
Accordingly, we collect all attained local configurations of threads affecting a
thread-local configuration $\bar u$ of an ego thread into the \emph{local trace}
reaching $\bar u$.
A \emph{thread-local} concrete semantics was also used in \citet{Mukherjee2017} for proving
the correctness of thread-modular analyses. The semantics there, however, is based on
\emph{interleaving} and permits \emph{stale} values for unread globals.
In contrast,
we consider a \emph{partial order} of past events and explicitly exclude the \emph{values}
of globals from local traces. These are instead
recovered from the local trace by searching for the \emph{last preceding write} at the point
when the value of the global is accessed.
%
%
We show that the set of all local traces
can conveniently be described by the least solution of a side-effecting constraint system
which is of a form quite similar to the ones used by the analyses and thus well-suited
for proving their correctness.

Having formulated both the analyses of Min\'e \cite{Mine2012} and
	Vojdani \cite{Vojdani2010,Vojdani2016}  by means of side-effecting constraint systems,
our contributions thus can be summarized as follows:
\begin{itemize}
\item 	we provide enhancements of each of these analyses which significantly
	increase their precision --- but still are incomparable;
\item 	since both analyses are expressed within the same framework, these improved
	versions can be integrated into one combined analysis;
\item	we prove the new analyses correct relative to a concrete local trace semantics
	of concurrent programs;
\item	we provide implementations of the new analyses to
	compare their precision and efficiency.
\end{itemize}

\noindent
The paper is organized as follows.
After a brief introduction into side-effecting constraint systems in \cref{s:side},
we introduce our toy language for which the concrete local trace semantics as well as the
analyses are formalized and indicate its operational semantics (\cref{s:lang}).
Our analyses then are provided in \cref{s:analysis}, while their correctness proofs
are deferred to \cref{s:correct}.
The experimental evaluation is provided in \cref{s:experimental}.
\cref{s:conclusion} finally concludes.

%% file: side.tex
\section{Side-effecting Systems of Constraints}\label{s:side}

In \cite{apinis2012side}, side-effecting systems of constraints are advocated as a convenient
framework for formalizing the accumulation of flow- (and possibly also context-) sensitive information
together with flow- (as well as context-) insensitive information.
Assume that
$\XX$ is a set of unknowns where for each $x\in\XX$,
$\D_x$ is a complete lattice of possible (abstract or concrete) values of $x$.
Let $\D$ denote the disjoint union of all sets $\D_x$.
Let $\XX\To\D$ denote the set of all mappings
$\eta:\XX\to\D$ where $\eta\,x\in\D_x$.
Technically, a (side-effecting) constraint takes the form
$x\sqsupseteq f_x$
where $x\in\XX$ is the left-hand side and the right-hand side
$f_x:(\XX\To\D)\to((\XX\To\D)\times\D_x)$ takes
a mapping $\eta:\XX\To\D$, while returning
a collection of side-effects to other unknowns in $\XX$ together with the
contribution to the left-hand side.

Let $\C$ denote a set of such constraints.
A mapping $\eta:\XX\To\D$ is called \emph{solution} of $\C$ if for all constraints
$x\sqsupseteq f_x$ of $\C$, it holds for $(\eta',d) = f_x\,\eta$  that
$\eta\sqsupseteq\eta'$ and
$\eta\,x\sqsupseteq d$; that is,
all side-effects of the right-hand side
and its contribution to the left-hand side are accounted for by $\eta$.
Assuming that all right-hand sides
are monotonic,
the system $\C$ is known to have a \emph{least} solution.

%% file: traces/basics.tex
\section{A Local Trace Semantics}\label{s:local}\label{s:traces}\label{s:lang}

Let us assume that there are disjoint sets $\X,\G$ of local and global variables which
take values from some set $\V$ of values.
Values may be of built-in types to compute with, e.g., of type $\textbf{int}$,
or a thread \emph{id} from a subset $\I\subseteq\V$.
The latter values are assumed to be \emph{abstract}, i.e.,
can only be compared with other thread \emph{id}s for equality.
We implicitly assume that all programs are well-typed; i.e., a variable
either always holds thread \emph{id}s or \textbf{int} values.
Moreover, there is one particular local variable  $\self\in\X$ holding
the thread \emph{id} of the current thread which is only implicitly assigned
at program start or when creating the thread.
Before program execution, global variables
are assumed to be uninitialized and will receive initial values via assignments
from the main thread, while local variables (except for $\self$)
may initially have any value.
%
%
Finally, we assume that there is a set $\M$ of mutexes.
A \emph{local program state} thus is a mapping $\sigma:\X\to\V$
where $\sigma\,\self\in\I$.
Let $\Sigma$ denote the set of all local program states.

Let $\A$ denote the set of actions.
\ignore{
This set should provide
guards $\textsf{exp}?$ and assignments $x=\textsf{exp}$ for expressions $\textsf{exp}$ which
refer to local variables only.
Later, we will add further primitives, for thread creation, for
locking/unlocking as well as for writing to and reading from globals.
}
Each thread is assumed to be represented by some control-flow graph
where each edge $e$ is of the form $(u,A,u')$ for program points $u,u'$ and action $A$.
Let $\N$ and $\E$ denote the sets of all program points and control-flow edges.
%
%
Let $\cal T$ denote a set of \emph{local traces}.
A local trace should be understood as the \emph{view} of a particular thread, the \emph{ego} thread,
on the global execution of the system.
Each $t\in\T$ ends at some program point $u$
with local state $\sigma$ where the ego thread \emph{id} is given by $\sigma\,\self$.
This pair $(u,\sigma)$ can be extracted from $t$ via the mapping
$
\sink:\T\to\N\times\Sigma
$.
For a local trace $t$ and local variable $x$, we also write $t(x)$ for the
result of $\sigma(x)$ if $\sink\,t = (u,\sigma)$.
Likewise, the functions $\id:\T\to\I$ and $\loc:\T\to\N$ return the thread \emph{id} and the
program point of the unique sink node, respectively.

We assume that there is a set $\init$ of initial local traces ${\bf 0}_\sigma$
with $\sink\,{\bf 0}_\sigma =(u_0,\sigma)$
where $u_0$ and $\sigma$ are the
start point and initial assignment to the local variables of the initial thread, respectively.
In particular, $\sigma\,\self = \mathit{0}$ for the initial thread \emph{id} $\mathit{0}$.
For every local trace that is not in $\init$ and where the ego thread has not just been started,
there should be a last action in $\A$ executed by the ego thread.
It can be extracted by means
of the function $\last:\T\to\A\cup\{\bot\}$.
For local traces in $\init$ or local traces where the ego thread has just been started,
$\last$ returns $\bot$.
For realizing thread creation,
we make the assumption that starting from $(u,\sigma)$, there is at most one outgoing edge
at which a thread is created.
%
For convenience, we also assume that each thread execution provides a unique program point $u_1$
at which the new thread is meant to start where
the local state of the created thread  agrees with the local state before thread creation --
only that the variable \textsf{self} receives a fresh value.
Accordingly, we require a function $\new:\N\to\T\to 2^\T$
so that $\new\,u_1\,t$ either returns the empty set, namely, when creation of a thread starting at point $u_1$ 
is not possible for $t$, or a set $\{t_1\}$
for a single trace $t_1$ if such thread creation is possible. In the latter case,
\[
\last(t_1) = \bot, \sink(t_1) = (u_1,\sigma_1)
\]
where for $\sink(t) = (u,\sigma)$,
$\sigma_1 = \sigma\oplus\{\self\mapsto\nu(t)\}$ for some function $\nu:\T\to\I$
providing us with a fresh thread \emph{id}.
If thread \emph{id}s are unique for a given creation history in $\T$, we
may identify the set $\I$ with $\T$ and let $\nu$ be the identity function.

For each edge $e=(u,A,u')$,
we also require an operation $\sem{e}:\T^k\to 2^\T$
where the arity $k$ for different actions may vary between 1 and 2
according to the arity of operations $A$ at the edges and where the
returned set either is empty (i.e., the operation is undefined), singleton (the operation is defined
and deterministic), or a larger set (the operation is non-deterministic, e.g., when reading
unknown input).
This function takes a local trace and extends it by executing
the action corresponding to edge $e$, thereby incorporating the matching trace from
the second argument set (if necessary and possible). In particular for
$t\in\sem{e}(t_0,\ldots,t_{k-1})$, necessarily,
\begin{equation}
\loc(t_0) = u,\quad
\last(t) = A\quad\text{and}\quad
\loc(t) = u'
\label{e:last}
\end{equation}
The set $\T$ of all local traces is the least
solution of the constraints
\begin{equation}
\begin{array}{llll}
\T	&\supseteq	&\textbf{fun}\,\_\to(\emptyset,\init) \\
\T	&\supseteq	&\textbf{fun}\,\T\to(\emptyset,\new\,u_1\,\T),&\qquad (u_1\in\N)	\\
\T	&\supseteq	&\textbf{fun}\,\T\to(\emptyset,\sem{e}(\T,\ldots,\T)),&\qquad(e\in\E)
\end{array}\label{e:global}
\end{equation}
where sets of side-effects are empty.
Here (and subsequently), we abbreviate for functions $f:\T^k\to 2^\T$ and subsets $T_0,\ldots, T_{k-1}\subseteq\T$,
the longish formula $\bigcup\{f(t_0,\ldots,t_{k-1})\mid t_0\in T_0,\ldots,t_{k-1}\in T_{k-1}\}$ to
$f(T_0,\ldots,T_{k-1})$.

The constraint system \eqref{e:global} globally collects all local traces into
one set $\T$. It serves as the definition of all (valid) local traces (relative
to the definitions of the functions $\sem{e}$ and $\new$)
and thus, as our reference trace semantics.
Subsequently, we provide a \emph{local} constraint system for these traces.
Instead of collecting one big set,
the local constraint system introduces unknowns $[u],u\in\N$, together with
individual constraints for each control-flow edge $e\in\E$.
The value for unknown $[u]$ is meant to collect the set of
those local traces $t$ that reach program point $u$ (i.e., $\loc(t) = u$),
while the constraints for edges
describe the possible relationships between these sets --- quite as for
the trace semantics of a sequential programming language.
In order to deal with concurrency appropriately, we additionally introduce
unknowns
$[m],m\in\M$ for mutexes.
%
These unknowns will not have right-hand sides on their own but receive their values via side-effects.
%
In general, we will have the following constraints
\begin{equation}
\begin{array}{llll}
\relax[u_0]	&\supseteq	&\textbf{fun}\,\_\to(\emptyset, \init) \\
\relax[u']	&\supseteq	&\sem{u,A},&\qquad((u,A,u')\in\E)
\end{array}\label{e:local}
\end{equation}
where the concrete form of the right-hand side $\sem{u,A}$ depends on the action $A$
of the corresponding edge $e=(u,A,u')$.
%
%
%
In the following, we detail how the constraints corresponding to the various
actions are constructed.

%% file: traces/threads.tex
\subsection{Thread Creation}\label{ss:threads}
%
Recall that we assume that within the set $\X$ of local variables,
we have one dedicated variable \self\ holding the thread \emph{id} of the ego thread.
In order to deal with thread creation
the set $\A$ of actions provides the $x = \create(u_1);$ operation
where $u_1$ is a program point where thread execution should start, and $x$ is a local variable
which is meant to receive the thread \emph{id} of the created thread.
The effect of \create\ is modeled as a side-effect to the program point $u_1$.
%
%
This means for a program point $u$:
\[
\begin{array}{lll}
\sem{u,x = \textsf{create}(u_1)}\,\Get	&=&
		\Let\;T = \sem{e}(\Get\,[u])\;\In	\\
&&		(\{[u_1]\mapsto\new\,u_1\,(\Get\,[u])\},T)	\\
\end{array}
\]
%

%% file: traces/locks.tex
\subsection{Locking and Unlocking}\label{ss:locks}

For simplicity, we only consider a fixed finite set $\M$ of mutexes.
If instead a semantics with
dynamically created mutexes were to be formalized, we could identify mutexes, e.g.,
via the local trace of the creating thread (as we did for threads).
For a fixed set $\M$ of mutexes, the set
$\cal A$ of actions provides operations $\lock(a)$ and $\unlock(a)$,
$a\in \M$, where these operations are assumed to return no value, i.e.,
do always succeed.
Additionally, we assume that $\unlock(a)$ for $a\in \M$ is only called by a thread
currently holding the lock of $a$, and that mutexes are not re-entrant;
i.e., trying to lock a mutex already held is undefined.
%
For convenience, we initialize the unknowns $[a]$ for $a\in \M$ to $\init$.
Then we set
\[
\begin{array}{lll}
\sem{u,\lock(a)}\,\Get &=& (\emptyset,\sem{e}(\Get\,[u],\Get\,[a]))
\\[2ex]
\sem{u,\unlock(a)}\,\Get
	&=&	\Let\;T = \sem{e}(\Get\,[u])\;\In	\\
& &		(\{[a]\mapsto T\}, T)
\end{array}
\]

%% file: traces/globals.tex
\subsection{Local and Global Variables}\label{ss:globals}

Expressions $r$ occurring as guards as well as non-variable right-hand sides of assignments
may refer to local variables only. For these, we assume an evaluation function
$\sem{\,.\,}$ so that for each $\sigma:\X\to\V$, $\sem{r}\,\sigma$ returns a value in $\V$.
For convenience, we here encode boolean values as integers where 0 denotes \emph{false} and
every non-zero value \emph{true}.
This evaluation function $\sem{\,.\,}$ allows defining the semantics $\sem{e}$
of a control-flow edge $e$ whose action $A$ is either a guard or an assignment to a local variable.
Since no side-effect is triggered, we have
\[
\sem{u,A}\,\eta = (\emptyset,\sem{e}(\eta\,[u]))
\]
For reading from and writing to globals, we consider the actions
$g=x;$ (copy value of the local $x$ into the global $g$) and
$x=g;$ (copy value of the global $g$ into the local $x$) only.
Thus, $g = g+1;$ for global $g$ is not directly supported by our language
but must be simulated by reading from $g$ into a local, followed by incrementing the local whose
result is eventually written back into $g$.

We assume for the concrete semantics that program execution is always
\emph{sequentially consistent}, and that both reads and writes to globals are \emph{atomic}.
The latter is enforced by introducing a dedicated mutex $m_g\in\M$ for each global $g$
which is acquired before $g$ is accessed
and subsequently released.
This means that each access $A$ to $g$ occurs as
\mintinline{c}{|\op{lock}|(m|$_{\,\texttt{g}}$|); A; |\op{unlock}|(m|$_{\,\texttt{g}}$|);}.

Under this proviso, the current value of each global $g$ read by some thread
can be determined just by inspection of the current local trace.
We have
\[
\begin{array}{lll}
\sem{u,x=g}\,\Get	&=& (\emptyset,\sem{e}(\Get\,[u]))	\\
\sem{u,g=x}\,\Get	&=& (\emptyset,\sem{e}(\Get\,[u]))	\\
\end{array}
\]
i.e., both reading from and writing to global $g$ is a transformation of individual local traces only.

%% file: traces/merge.tex
\subsection{Completeness of the Local Constraint System}\label{ss:completeness}

\noindent With the following assumption on $\sem{e}$ in addition to \cref{e:last},
\begin{itemize}
	\item The binary operation $\sem{(u,\lock(a),u')}(t_0,t_1)$ only returns a non-empty set
		if $t_1 \in \init \lor \last(t_1) = (\unlock(a))$, i.e.,
		locking only incorporates local traces from the set $\init$ or local traces ending
		in a corresponding $\unlock(a)$
\end{itemize}

\noindent we obtain:

\begin{theorem}\label{t:main}
Let $\T$ denote the least solution of the global constraint system \eqref{e:global},
and $\Get$ denote the least solution of the local constraint system \eqref{e:local}.
Then
\begin{enumerate}
\item	$\Get\,[u] = \{t\in\T\mid\loc(t) = u\} \quad$ for all $u\in\N$;
\item	$\Get\,[a] = \init \cup \{t\in\T\mid\last(t)=(\unlock(a))\} \quad$
					for all $a\in\M$.
\end{enumerate}
\end{theorem}

\noindent
In fact, \cref{t:main} holds for any formalism for traces matching these assumptions.
Before detailing an example trace formalism in \cref{s:details}, we proceed with an intuitive example.

%% file: traces/intuitive-example.tex
\begin{example}\label{e:intuitive-traces}
Consider the following program and assume that execution starts at program point $u_0$.
\begin{center}
    \begin{minipage}[t]{5cm}
    \begin{minted}{c}
    u0:
       x = |\op{create}|(u6);
       |\op{lock}|(m|$_{\,\texttt{g}}$|);
       g = 1;
       |\op{unlock}|(m|$_{\,\texttt{g}}$|);
    \end{minted}
    \end{minipage}
    \begin{minipage}[t]{4cm}
    \begin{minted}{c}
    u6:
       y = 1;
       |\op{lock}|(m|$_{\,\texttt{g}}$|);
       g = 2;
       |\op{unlock}|(m|$_{\,\texttt{g}}$|);

    \end{minted}
    \end{minipage}
\end{center}
\noindent In this example, one of the traces in the set $\init$ of initial local traces is the trace $0_\sigma$
with $\sink\,0_\sigma = (u_0, \sigma_{u_0}) = (u_0, \{x \mapsto \mathit{0}, \self \mapsto \mathit{0} \})$;
i.e., local variable $x$ has value $\textit{0}$, $y$ has value $0$, and the initial thread has thread \emph{id} $\textit{0}$.
One of the traces reaching program point $u_1$ is $t$ which is obtained by prolonging
$0_\sigma$ where $\sink\,t = (u_1,\sigma_{u_1}) = (u_1, \{x \mapsto \mathit{1}, \self \mapsto \mathit{0} \})$.
We abbreviate $\bar{u_k}$ for $(u_k,\sigma_{u_k})$ and show traces as graphs (\cref{f:tracesmg}).
\ignore{For example, for $t$ we write:
\tikzset{
every node/.style={node distance=55pt},
programpoint/.style={circle,draw,font=\small},
edgelabel/.style={midway,font=\small,above=0.5mm}}

\begin{tikzpicture}
    \node[programpoint](pp0){$\bar{u_0}$};
    \node[programpoint,right=of pp0](pp1){};

    \draw[->](pp0)--(pp1)node[edgelabel]{$x=\create(u_6)$};
\end{tikzpicture}}
\noindent Since $\sem{\lock(m_g)}$ is a binary operation, to compute the set of local traces reaching $u_2$,
not only the local traces reaching its predecessor $u_1$
but also those traces stored at the constraint system unknown $[m_g]$
need to be considered.

\cref{f:tracesmg} shows all local traces starting with $\bar{u_0}$ stored at $[m_g]$, i.e.,
all local traces in which the last action of the ego thread is $\unlock(m_g)$ (that start with $\bar{u_0}$).
Out of these, traces (a) and (c) are compatible with $t$.
Prolonging the resulting traces for the following assignment and unlock operations leads to traces
(b) and (d) reaching the program point
after the $\unlock(m_g)$ in this thread. Therefore, (b) and (d) are among those traces that are side-effected to $[m_g]$.

\tikzset{
every node/.style={node distance=35pt and 55pt},
programpoint/.style={circle,draw,font=\small},
edgelabel/.style={midway,font=\small,above=0.5mm}}
\begin{figure*}[ht!]
    \begin{subfigure}{0.14\textwidth}
        \center
        \begin{tikzpicture}
            \node[programpoint](pp0){$\bar{u_0}$};
        \end{tikzpicture}
        \caption{}
    \end{subfigure}
    \begin{subfigure}{0.85\textwidth}
        \center
    \begin{tikzpicture}
        \node[programpoint](pp0){$\bar{u_0}$};
        \node[programpoint,right=of pp0](pp1){};
        \node[programpoint,right=of pp1](pp2){};
        \node[programpoint,right=of pp2](pp3){};
        \node[programpoint,right=of pp3](pp4){};

        \draw[->](pp0)--(pp1)node[edgelabel]{$x=\create(u_6)$};
        \draw[->](pp1)--(pp2)node[edgelabel]{$\lock(m_g)$};
        \draw[->](pp2)--(pp3)node[edgelabel]{$g=1$};
        \draw[->](pp3)--(pp4)node[edgelabel]{$\unlock(m_g)$};

        \draw[->,red, out=340,in=220](pp0) to node[edgelabel, below=0.5mm]{$\to_{m_g}$} (pp2);
    \end{tikzpicture}
    \caption{}
    \end{subfigure}
    \begin{subfigure}{1.0\textwidth}
    \center
    \begin{tikzpicture}
        \node[programpoint](pp0){$\bar{u_0}$};
        \node[programpoint,below right=of pp0](ppa){$\bar{u_6}$};
        \node[programpoint,right=of ppa](ppb){};
        \node[programpoint,right=of ppb](ppc){};
        \node[programpoint,right=of ppc](ppd){};
        \node[programpoint,right=of ppd](ppe){};

        \draw[->,blue](pp0)--(ppa)node[edgelabel,left=0.5mm]{$\to_c$};

        \draw[->](ppa)--(ppb)node[edgelabel]{$y=1$};
        \draw[->](ppb)--(ppc)node[edgelabel]{$\lock(m_g)$};
        \draw[->](ppc)--(ppd)node[edgelabel]{$g=2$};
        \draw[->](ppd)--(ppe)node[edgelabel]{$\unlock(m_g)$};

        \draw[->,red, out=330,in=125](pp0) to node[edgelabel, below=0.5mm]{$\to_{m_g}$} (ppc);
    \end{tikzpicture}
    \caption{}
    \label{sf:ax}
    \end{subfigure}
    \begin{subfigure}{1.0\textwidth}
        \center
    \begin{tikzpicture}
        \node[programpoint](pp0){$\bar{u_0}$};

        \node[programpoint,right=of pp0](pp1){};
        \node[programpoint,right=of pp1](pp2){};
        \node[programpoint,right=of pp2](pp3){};
        \node[programpoint,right=of pp3](pp4){};

        \node[programpoint,below right=of pp0](ppa){$\bar{u_6}$};
        \node[programpoint,right=of ppa](ppb){};
        \node[programpoint,right=of ppb](ppc){};
        \node[programpoint,right=of ppc](ppd){};
        \node[programpoint,right=of ppd](ppe){};

        \draw[->](pp0)--(pp1)node[edgelabel]{$x=\create(u_6)$};
        \draw[->](pp1)--(pp2)node[edgelabel]{$\lock(m_g)$};
        \draw[->](pp2)--(pp3)node[edgelabel]{$g=1$};
        \draw[->](pp3)--(pp4)node[edgelabel]{$\unlock(m_g)$};

        \draw[->,blue](pp0)--(ppa)node[edgelabel, left=0.5mm]{$\to_c$};

        \draw[->](ppa)--(ppb)node[edgelabel]{$y=1$};
        \draw[->](ppb)--(ppc)node[edgelabel]{$\lock(m_g)$};
        \draw[->](ppc)--(ppd)node[edgelabel]{$g=2$};
        \draw[->](ppd)--(ppe)node[edgelabel]{$\unlock(m_g)$};

        \draw[->,red, out=330,in=135](pp0) to node[edgelabel, below=0.5mm]{$\to_{m_g}$} (ppc);
        \draw[->,red, out=100,in=210](ppe) to node[edgelabel, below=0.5mm]{$\to_{m_g}$} (pp2);
    \end{tikzpicture}
    \caption{}
\end{subfigure}
\begin{subfigure}{1.0\textwidth}
    \center
    \begin{tikzpicture}
        \node[programpoint](pp0){$\bar{u_0}$};

        \node[programpoint,right=of pp0](pp1){};
        \node[programpoint,right=of pp1](pp2){};
        \node[programpoint,right=of pp2](pp3){};
        \node[programpoint,right=of pp3](pp4){};

        \node[programpoint,below right=of pp0](ppa){$\bar{u_6}$};
        \node[programpoint,right=of ppa](ppb){};
        \node[programpoint,right=of ppb](ppc){};
        \node[programpoint,right=of ppc](ppd){};
        \node[programpoint,right=of ppd](ppe){};

        \draw[->](pp0)--(pp1)node[edgelabel]{$x=\create(u_6)$};
        \draw[->](pp1)--(pp2)node[edgelabel]{$\lock(m_g)$};
        \draw[->](pp2)--(pp3)node[edgelabel]{$g=1$};
        \draw[->](pp3)--(pp4)node[edgelabel]{$\unlock(m_g)$};

        \draw[->,blue](pp0)--(ppa)node[edgelabel,left=0.5mm]{$\to_c$};

        \draw[->](ppa)--(ppb)node[edgelabel]{$y=1$};
        \draw[->](ppb)--(ppc)node[edgelabel]{$\lock(m_g)$};
        \draw[->](ppc)--(ppd)node[edgelabel]{$g=2$};
        \draw[->](ppd)--(ppe)node[edgelabel]{$\unlock(m_g)$};

        \draw[->,red, out=340,in=220](pp0) to node[edgelabel, below=0.5mm]{$\to_{m_g}$} (pp2);
        \draw[->,red, out=230,in=130](pp4) to node[edgelabel, below=0.5mm]{$\to_{m_g}$} (ppc);
    \end{tikzpicture}
    \caption{}
\end{subfigure}
    \caption{Local traces of \cref{e:intuitive-traces} starting with $\bar{u_0}$ stored at $[m_g]$.}
    \label{f:tracesmg}
\end{figure*}
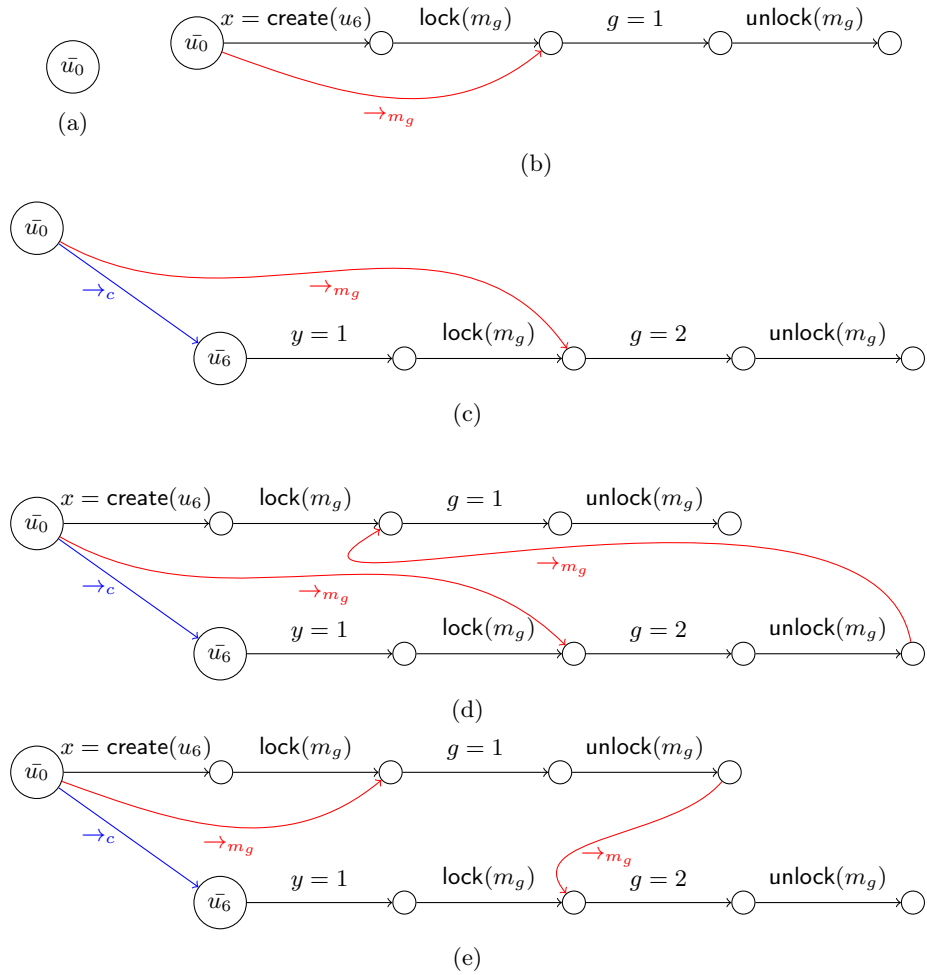
\end{example}

%% file: traces/example-formalism.tex
\subsection{Example Formalism for Local Traces}\label{s:details}
The concrete concurrency semantics imposes \emph{restrictions} onto when binary actions are defined.
In particular, binary operations $\sem{e}$ may only be defined for a pair $(t_0,t_1)$ if certain parts of $t_0$
and $t_1$ represent the same computation.
In order to make such restrictions explicit, we introduce a concrete representation of local traces.

A \emph{raw} (finite) trace of single thread $i\in\I$ is a sequence
$\lambda =\bar u_0a_1\ldots\bar u_{n-1}a_n\bar u_n$ for states
$\bar u_j = (u_j,\sigma_j)$ with
$\sigma_j\,\self = i$, and actions $a_j\in\A$ corresponding to
the local state transitions of the thread $i$ starting in configuration
$\bar u_0$ and executing actions $a_j$.
In that sequence, every action $\lock(m)$ is assumed to succeed,
and when accessing a global $g$, any value may be read.
We may view $\lambda$ as an acyclic graph whose nodes are the 3-tuples
$(j,u_j,\sigma_j), j= 0,\ldots,n,$ and
whose edges are $((j-1,u_{j-1},\sigma_{j-1}),a_j,(j,u_j,\sigma_j)), j= 1,\ldots,n$.
Let $V(\lambda)$ and $E(\lambda)$ denote the
set of nodes and edges of this graph, respectively.

Let $\Lambda(i)$ denote the set of all individual traces for thread $i$, and
$\Lambda$ the union of all these sets.

A \emph{raw global trace} of threads is an acyclic graph $\tau=(\V,\E)$ where
$\V = \bigcup\{V(\lambda_i)\mid i\in I\}$ and $\E = \bigcup\{E(\lambda_i)\mid i\in I\}$
for a set $I$ of thread \emph{id}s and raw local traces $\lambda_i\in\Lambda(i)$.
On the set $\V$, we define the \emph{(immediate) program order}
as the set of all pairs $\bar u\to_{p}\bar u'$ for which there is an edge $(\bar u,a,\bar u')$ in $\E$.
In order to formalize our notion of local traces, we extend the program order to a
\emph{causality order} which additionally takes the order into account in which threads are
created as well as the order in which mutex locks are acquired and released.

For $a\in\M$, let $a^+ \subseteq \V$ denote the set of nodes $\bar u$
where an incoming edge is labeled $\lock(a)$, i.e., $ \exists x\,(x,\lock(a),\bar u) \in \E$,
and $a^-$ analogously for $\unlock(a)$.
On the other hand, let $C$ denote the set of nodes with an \emph{outgoing} edge labeled $x'=\create(u_1)$ (for any local variable $x'$ and program point $u_1$).
Let $S$ denote the set of minimal nodes w.r.t.\ to $\to_p$, i.e., the points at which threads start and let $\bf 0$ the node $(0,u_0,\sigma_0)$ where $\sigma_0\,\self = 0$.

A \emph{global trace} $t$ then is represented by a tuple $(\tau,\to_c,(\to_a)_{a\in\M})$
where $\tau$ is a raw global trace and the relations $\to_c$ and $\to_a$ ($a\in\M$) are
the create and locking orders for the respective mutexes.
The \emph{causality order} $\leq$ of $t$ then is obtained as the reflexive and transitive closure
of the union $\to_p\cup\to_c\cup\bigcup_{a\in\M}\to_a$.
These orders should satisfy the following properties.
\begin{itemize}
\item	\textbf{Causality order} $\leq$ should be a partial order with unique least element
	$(0,u_0,\sigma_0)$ where $\sigma_0\,\self=0$;
\item \textbf{Create order}:
$\to_c \subseteq C \times (S\setminus \{ \textbf{0}\})$: 
$\forall s \in (S \setminus \{ \textbf{0}\}): \abs{\{ z \mid z \to_c s \}} = 1$, i.e.,
every thread except the initial thread is created by exactly one
$\create(...)$ action and $\forall x: \abs{\{ z \mid x \to_c z \}} \leq 1$, i.e.,
each $\create(...)$ action creates at most one thread.
Additionally, for $((j-1,u_{j-1},\sigma_{j-1}), x = \create(v),(j,u_{j},\sigma_{j})) \in \E$ and
    $(j-1,u_{j-1},\sigma_{j-1}) \to_{c} (0,v,\sigma'_0)$:
    $\sigma'_0 = \sigma_{j-1} \oplus \{ \self \mapsto i' \}$ for some thread \emph{id} $i'$ where
    $\sigma_j\,x = i'$, i.e., the creating and the created thread agree on the thread \emph{id}
    of the created thread and the values of locals.
    \item \textbf{Locking order}:
$\forall a \in \M: \to_a \subseteq (a^- \cup \textbf{0}) \times a^+$:
$\forall x: \abs{\{ z \mid x \to_a z \}} \leq 1$ and
$\forall y: \abs{\{ z \mid z \to_a y \}} = 1$, that is, for a mutex $a$ every lock is preceded
by exactly one unlock (or it is the first lock) of $a$, and each unlock is directly followed
by at most one lock.
\item	\textbf{Globals}:
Additionally, the following consistency condition on values read from globals needs to hold:
For $((j-1,u_{j-1},\sigma_{j-1}),x = g,(j,u_{j},\sigma_{j})) \in \E$,
there is a maximal node $(j',u_{j'},\sigma_{j'})$ w.r.t.\ to the causality order $\leq$ such that
$((j'-1,u_{j'-1},\sigma_{j'-1}),g = y,(j',u_{j'},\sigma_{j'})) \in \E$
and $(j',u_{j'},\sigma_{j'}) \leq (j-1,u_{j-1},\sigma_{j-1})$.
Then $\sigma_j\,x = \sigma_{j'-1}\,y$, i.e., the value read for a global is the last value written to it.
\end{itemize}

\noindent
A global trace $t$ is \emph{local} if it has a unique maximal element $\bar u = (j,u,\sigma)$
(w.r.t $\leq$). Then in particular, $\sink(t) = (u,\sigma)$.
The function $\last$ extracts the last action $A$ of the ego thread (if there is any) and returns $\bot$
otherwise. The partial functions $\new\,u$ for program points $u$ and
$\sem{e}$ for control-flow edges $e$ then are defined by extending
a given local trace appropriately.

%% file: analyses/intro.tex
\section{Static Analysis of Concurrent Programs}\label{s:analysis}

In the following, we present four analyses
which we will compare for precision 
and efficiency.
In the present paper, we are only interested in non-relational analyses.
An extension to relational analyses is deferred to a companion paper.
Also, we concentrate on mutexes only and hence do not track thread \emph{id}s.
In the same way as in Min\'e's paper, the precision of all presented analyses could be improved by tracking (abstract or concrete) thread \emph{id}s.
Also, weak memory effects at asynchronous memory accesses are ignored.

The first analysis (\emph{Protection-Based Reading}) is an improved version of Vojdani's analysis
\cite{Vojdani2010,Vojdani2016}, while
the second analysis (\emph{Lock-Centered Reading}) is an improved version of the analysis proposed by Min\'e \cite{Mine2012}.
%
The first analysis assumes that for each global $g$, some set $\MM[g]$ of mutexes exists which is held at each write operation to $g$
and maintains a private copy of the global as long as one of the mutexes from $\MM[g]$ is known to be held.
Since the assumption of non-empty program-wide protecting locksets is rather restrictive,
we present a third analysis (\emph{Write-Centered Reading}) which lifts this extra assumption
and thus strictly subsumes \emph{Protection-Based Reading}.
Interestingly,
\emph{Write-Centered Reading} and \emph{Lock-Centered Reading} are still incomparable.
We therefore sketch a fourth analysis which is more precise
than either of them.

Throughout this section, we assume that $\D$ is a complete lattice abstracting sets of values of
program variables.

%% file: analyses/protection-based.tex
\subsection{Protection-Based Reading}\label{s:protection-based}
The original analysis proposed by Vojdani \cite{Vojdani2016,Vojdani2010} and implemented in the
\textsc{Goblint} system assumes that for each global $g$,
there is a set of mutexes definitely held whenever $g$ is accessed.
The best information about the values of $g$ visible after acquiring a protecting lock is maintained
in a separate unknown $[g]$.
The value of the unknown $[g]$ for the global $g$ is eagerly privatized: It is incorporated into the local state for a program point
and currently held lockset whenever $g$ first becomes protected, i.e., a mutex protecting $g$ is acquired while none was held before.
As long as one of these protecting mutexes is held, all reads and writes refer to this local copy of the global and this copy can be
destructively updated. It is only when no mutex protecting $g$ is held anymore that the value of the local copy is published to the
unknown $[g]$.
This base analysis setting is enhanced in three ways:
\begin{itemize}
\item	Instead of assuming a set of mutexes protecting both reading and writing of $g$,
	we now just assume a set of mutexes definitely held at each write. While this does not necessarily lead to an improvement
	in precision, it allows for analyzing interesting patterns where, e.g., only a subset of
mutexes is acquired for reading from a global, while a
	superset is held when writing to it.
\item	Besides the unknown $[g]$ for describing the possible values of the global $g$ for protected accesses,
	another unknown $[g]'$ is introduced for the results of unprotected read accesses to $g$.
\item	Instead of incorporating the value of the global $g$ stored at $[g]$ into the local state at each
	$\lock$ operation of a mutex
	from the protecting set, the local state for a program point and currently held lockset only keeps track of the values written by the ego thread.
	At a read operation $x=g$, the value of global $g$ is assigned to the local variable $x$.
	For that the analysis relies on the value stored at unknown $[g]$ together with the value
	of $g$ stored in the local state,
	unless the ego thread has definitely written
	to $g$ since acquiring a protecting mutex and not yet released all protecting mutexes
	since then.
\end{itemize}
Recall that $\MM:\G\to 2^\M$ maps each global $g$ to the set of mutexes definitely held
when $g$ is \emph{written to}.
Due to our atomicity assumption, the set $\MM[g]$ is non-empty, since $m_g\in\MM[g]$ always holds.
For the moment, we assume this mapping to be given.
%
The unknown $[g]'$ stores an abstraction of all values ever written to $g$, while
the unknown $[g]$ stores an abstraction of all values that were written \emph{last} before releasing a protecting mutex of $g$ other than $m_g$.
For each pair $(u,S)$ of program point $u$ and currently held lockset $S$, on the other hand, the analysis maintains
(1) a set $P$
of definitely written globals $g$ since a protecting mutex of $g$
has been acquired and not all protecting mutexes have been released, together with
(2) a variable assignment $\sigma:\X\cup\G\to\D$ of potential descriptions of values for
local or global variables.

In case one of the mutexes in $\MM[g]$ is definitely held, after a write to variable $g$, all processing on $g$ is performed destructively
on the local copy.
Immediately after the write to $g$ (at the $\unlock(m_g)$) the value of the updated local copy is merged
into $[g]'$ via a side-effect.
On the other hand, the value of that copy must be merged into the value of $[g]$
only when it no longer can be guaranteed that all other protecting mutexes $(\MM[g] \setminus m_g)$ are held.

We start by giving the right-hand-side function for the start state at program point $u_0 \in \N$ with the empty lockset $\emptyset$, i.e.,
$[u_0,\emptyset] \sqsupseteq \init^\sharp$ where
\[
	\begin{array}{lll}
		\init^\sharp\,\_ &=& \Let\; \sigma = \{ x \mapsto \top \mid x \in \X \} \cup \{ g \mapsto \bot \mid g \in \G \} \;\In\\
		&& (\emptyset,(\emptyset,\sigma))
	\end{array}
\]
Now, consider the right-hand side $[v,S']\sqsupseteq\sem{[u,S],A}^\sharp$ for the edge $e=(u,A,v)$ of the control-flow graph
and appropriate locksets $S,S'$.
Consider the right-hand side for a thread creation edge.
For this, we require a function $\nu^\sharp\,u\,(P,\sigma)\,u_1$ that returns the (abstract) thread \emph{id} of a thread started
at an edge originating from $u$ in local state $(P,\sigma)$, where the new thread starts execution at program point $u_1$.
Since we do not track thread \emph{id}s, $\nu^\sharp$ may return $\top$ whereby all variables
holding thread \emph{id}s are also set to $\top$.
\[
\begin{array}{lll}
\sem{[u,S],x = \create(u_1)}^\sharp\Get	&=&	\Let\;(P,\sigma) = \Get\,[u,S]\;\In	\\
				& & \Let\; i = \nu^\sharp\,u\,(P,\sigma)\,u_1\;\In\\
				& & \Let\; \sigma' = \sigma \oplus (\{ \self \mapsto i\} \cup \{ g \mapsto \bot \mid g \in \G\})\;\In\\
				& & \Let\; \rho = \{ [u_1,\emptyset] \mapsto (\emptyset,\sigma') \}\;\In\\
				& &	(\rho,(P,\sigma \oplus \{ x \mapsto i \}))	\\
\end{array}
\]
This function has no effect on the local state apart from setting $x$ to the abstract thread \emph{id} of the newly created thread,
while providing an appropriate initial state to the startpoint of the newly created thread.
For guards and computations on locals, the right-hand-side functions are defined in the intuitive manner --- they operate on $\sigma$ only,
leaving $P$ unchanged.

Concerning locking and unlocking of mutexes $a$, the lock operation does not affect the local state,
while at each unlock, all local copies of globals $g$
for which not all protecting mutexes are held anymore,
are published via a side-effect to the respective unknowns $[g]$ or $[g]'$.
Moreover, globals for which none of the protecting mutexes are held anymore, are removed from $P$:
\[
\begin{array}{lll}
\sem{[u,S],\lock(a)}^\sharp\Get	&=&	(\emptyset,\Get\,[u,S])	\\[1ex]
\sem{[u,S],\lock(m_g)}^\sharp\Get	&=&	(\emptyset,\Get\,[u,S])	\\[1ex]
\sem{[u,S],\unlock(m_g)}^\sharp\Get	&=&	\Let\;(P,\sigma) = \Get\,[u,S]\;\In	\\
	& &	\Let\;P' = \{h\in P\mid ((S\setminus\{m_g\})\cap\MM[h])\neq\emptyset\}\;\In	\\
	& &	\Let\;\rho = \{[g]'\mapsto\sigma\,g \} \cup \{[g] \mapsto\sigma\,g \mid \MM[g] = \{m_g\}\}\;\In \\
	& &	(\rho,(P',\sigma))	\\[1ex]
\sem{[u,S],\unlock(a)}^\sharp\Get	&=&	\Let\;(P,\sigma) = \Get\,[u,S]\;\In	\\
	& &	\Let\;P' = \{g\in P\mid ((S\setminus\{a\})\cap\MM[g])\neq\emptyset\}\;\In	\\
	& &	\Let\;\rho = \{[g]\mapsto\sigma\,g\mid a\in \MM[g]\}\;\In \\
	& &	(\rho,(P',\sigma))
\end{array}
\]
for $a \not\in \{m_g \mid g \in \G \}$.
We remark that the locksets $S'$ of the corresponding left-hand unknowns now take the forms of
$S'=S\cup\{a\}$, $S'=S\cup\{m_g\}$, $S'= S\setminus\{m_g\}$ and $S' = S\setminus\{a\}$, respectively.
Recall that the dedicated mutex $m_g$ for each global $g$ has been
introduced for guaranteeing atomicity. It is always acquired immediately before and always released immediately after
each access to $g$. The special treatment of this dedicated mutex implies that all values written to $g$
are side-effected to the unknown $[g]'$, while
values written to $g$ are side-effected to the unknown $[g]$ only
when \unlock\ is called for a mutex \emph{different} from $m_g$.

For global $g$ and local $x$, we define for writing to and reading from $g$,
\[
\begin{array}{lll}
\sem{[u,S],g=x}^\sharp\Get	&=&	\Let\;(P,\sigma) = \Get\,[u,S]\;\In	\\
			& &	(\emptyset, (P\cup\{g\},
					\sigma\oplus\{g\mapsto(\sigma\,x)\}))	\\[1ex]
\sem{[u,S],x=g}^\sharp\Get	&=&	\Let\;(P,\sigma) = \Get\,[u,S]\;\In	\\
			& &	\Iif\;(g\in P)\;\Then	\\
			& &	\qquad(\emptyset,
				(P,\sigma\oplus\{x\mapsto(\sigma\,g)\}))	\\
			& &	\Eelse\;\Iif\;(S\cap\MM[g]=\{m_g\})\;\Then	\\
			& &	\qquad(\emptyset,
				(P,\sigma\oplus\{x\mapsto\sigma\,g\sqcup\Get\,[g]'\}) \\
			& &	\Eelse\;(\emptyset,
				(P,\sigma\oplus\{x\mapsto \sigma\,g\sqcup\Get\,[g]\}))
\end{array}
\]
%
Altogether, the resulting system of constraints $\C_\ProtectionBased$ is monotonic
(given that the right-hand-side functions for local computations as well as for guards
are monotonic) --- implying that the system has a unique least solution, which we denote by $\Get_\ProtectionBased$.
We remark that for this unique least solution $\Get_\ProtectionBased$, $\Get_\ProtectionBased\,[g] \sqsubseteq \Get_\ProtectionBased\,[g]'$ holds.

\begin{example}
Consider, e.g., the following program fragment and assume that $\MM[g] = \{a, m_g\}$ and that
that we use value sets for abstracting int values.
\begin{center}
	\begin{minipage}[t]{7cm}
	\begin{minted}{c}
	|\op{lock}|(a);
	|\op{lock}|(m|$_{\,\texttt{g}}$|);  g = 5;    |\op{unlock}|(m|$_{\,\texttt{g}}$|);
	|\op{lock}|(b);
	|\op{unlock}|(b);
	|\op{lock}|(m|$_{\,\texttt{g}}$|);  x = g;    |\op{unlock}|(m|$_{\,\texttt{g}}$|);
	|\op{lock}|(m|$_{\,\texttt{g}}$|);  g = x+1;  |\op{unlock}|(m|$_{\,\texttt{g}}$|);
	|\op{unlock}|(a);
	\end{minted}
	\end{minipage}
\end{center}
Then after $\unlock(b)$, the state attained by the program (where variable $\self$ is omitted for clarity of presentation) is
\[
s_1=(\{g\},\{g\mapsto \{5\}, x \mapsto \top \})
\]
where $[g]'$ has received the contribution $\{5\}$ but no side-effect to $[g]$ has been triggered.
The read of $g$ in the subsequent assignment refers to the local copy.
Accordingly, the second write to $g$ and the succeeding $unlock(m_g)$ result in the local state
\[
s_2=(\{g\},\{g\mapsto \{6\}, x \mapsto \{5\} \})
\]
with side-effect $\{6\}$ to $[g]'$ and no side-effect to $[g]$. 
Accordingly, after $\unlock(a)$, the attained state is
\[
s_3=(\emptyset,\{g\mapsto \{6\}, x \mapsto \{5\} \})
\]
and the value of $[g]$ is just $\{6\}$ --
even though $g$ has been written to twice.
We remark that without separate treatment of $m_g$, the value of $\{5\}$ would immediately be
side-effected to $[g]$.
\qed
\end{example}

\begin{theorem}
\emph{Protection-Based Reading} is sound w.r.t.\ the trace semantics.
\end{theorem}
\begin{proof}
In \cref{s:soundness-protection-based} we show that this analysis computes an abstraction of the result of the analysis presented in
\cref{s:write-centered}, which we then prove to be sound with respect to the trace semantics in \cref{s:soundness-write-centered}.\qed
\end{proof}
Thus, we never remove any values from the variable assignment for a local state.
An implementation may, however, in order to keep the representation of local states small, additionally track for each program point
and currently held lockset, a set $W$ of all globals
which possibly have been written (and not yet published) while holding protecting mutexes.
A local copy of a global $g$ may then safely be removed from $\sigma$ if $g \notin P \cup W$.
This is possible because for each $g \notin P \cup W$,
$\sigma\,g$ has already been side-effected and hence already is included in
$\Get\,[g]$ and $\Get\,[g]'$, and thus $\sigma\,g$ need not be
consulted on the next read of $g$. 

\medskip
As presented thus far, this analysis requires the map $\MM:\G\to 2^\M$
to be given beforehand.
This map can, e.g., be provided by some pre-analysis onto which the given analysis builds.
Alternatively, our analysis can be modified to infer $\MM$ on the fly.
For that, we consider the $\MM[g]$ to be separate unknowns of the
constraint system. They take values in the complete lattice $2^\M$ (ordered
by superset) and are initialized to the full set of all mutexes $\M$.
The right-hand-side function for writes to global $g$ then is extended
to provide the current lockset as a contribution to $\MM[g]$. This means that we now have:
\[
\begin{array}{lll}
\sem{[u,S],g=x}^\sharp\Get	&=&	\Let\;(P,\sigma) = \Get\,[u,S]\;\In	\\
			& &	(\{\MM[g]\mapsto S\},
				(P\cup\{g\},
					\sigma\oplus\{g\mapsto(\sigma\,x)\}))	\\[1ex]
\end{array}
\]
There is one (minor) obstacle, though: the right-hand-side function for control-flow edges with
$\unlock(a)$ is no longer monotonic in the unknowns $\MM[g],g\in\G$: If $\MM[g]$ shrinks to no
longer contain $a$, $\unlock(a)$ will no longer produce a side-effect to the unknown $[g]$,
whereas it previously did.

Another practical consideration is that, in order to further improve efficiency,
it is also possible to abandon
state-splitting according to held locksets --- at the cost of losing some precision.
To this end, it suffices to additionally track for each program point a set $\bar S$ of must-held mutexes
as part of the local state from the lattice $2^\M$ (ordered by superset), and replace $S$ with $\bar S$
in all right-hand sides.

%% file: analyses/lock-centered.tex
\subsection{Lock-Centered Reading}\label{s:lock-centered}
The analysis by Min\'e from \cite{Mine2012}, when stripped of thread \emph{id}s and other
features specific to real-time systems such as \textsc{Arinc653} and reformulated by means
of side-effecting constraint systems, works as follows:
It maintains for each pair $(u,S)$ of program point $u$
and currently held lockset $S$, copies of globals $g$ whose values are weakly updated whenever the lock for some mutex $a$
is acquired. In order to restrict the set of possibly read values, the global $g$ is
split into unknowns $[g,a,S']$ where $S'$ is a \emph{background} lockset held
by another thread immediately after executing $\unlock(a)$.
Then only the values of those unknowns $[g,a,S']$ are taken into account where
$S\cap S' = \emptyset$.

For a detailed account of Miné's analysis see \cref{s:mine}.
%
%
We identify two sources of imprecision in this analysis.
One source is \emph{eager reading}, i.e., reading in values of $g$ at
every $\lock(a)$ operation. This may import the values of \emph{too many} unknowns
$[g,a,S']$ into the local state.
Instead,
it suffices for each mutex $a$, to read values at the \emph{last} $\lock(a)$ before actually accessing the global.

Let $\UM$ denote the set of all upward-closed subsets of $\M$, ordered by subset inclusion.
For convenience, we represent each non-empty value in $\UM$ by the set of its minimal elements.
Thus, the \emph{least} element of $\UM$ is $\emptyset$, while the \emph{greatest} element is given by
the \emph{full} power set of mutexes (represented by $\{\emptyset\}$).

We now maintain a map $L: \M \to \UM$ in the local state that tracks for each mutex $a$ all
minimal
background locksets that were held when $a$ was acquired last.
This abstraction of acquisition histories~\cite{Kahlon2005,Kahlon2007} 
allows us to delay the reading of globals until the point where the program
actually accesses their values. We call this behavior \emph{lazy reading}.

The other source of imprecision is that each thread may publish values it has not written itself.
In order to address this issue, we let
 $\sigma\,g$ only maintain values the ego thread itself has written. 

A consequence of \emph{lazy reading} is that values for globals are now read from the
global invariant at each read.
In case the ego thread has definitely written to a variable and no additional locks have occurred since,
only the local copy needs to be read. To achieve that, we introduce an additional map $V:\M\to 2^\G$.
For mutex $a$, $V\,a$ is the set of global variables that were definitely written since $a$ was
last acquired. In case that $a$
has never been acquired by the ego thread,
we set $V\,a$ to the set of all global variables that have definitely been written
since the start of the thread.

We start by giving the right-hand-side function for the start state at program point $u_0 \in \N$ with the empty lockset $\emptyset$, i.e.,
$[u_0,\emptyset] \sqsupseteq \init^\sharp$ where
\[
	\begin{array}{lll}
		\init^\sharp\,\_ &=& \Let\; V = \{ a \mapsto \emptyset \mid a \in \M \} \;\In\\
		&& \Let\; L = \{ a \mapsto \emptyset \mid a \in \M \} \;\In\\
		&& \Let\; \sigma = \{ x \mapsto \top \mid x \in \X \} \cup \{ g \mapsto \bot \mid g \in \G \} \;\In\\
		&& (\emptyset,(V,L,\sigma))
	\end{array}
\]
Next, we sketch the right-hand-side function for a thread creation edge.
\[
\begin{array}{lll}
\sem{[u,S],x = \create(u_1)}^\sharp\Get	&=&	\Let\;(V,L,\sigma) = \Get\,[u,S]\;\In	\\
				& & \Let\; V' = \{ a \mapsto \emptyset \mid a \in \M \}\;\In\\
				& & \Let\; L' = \{ a \mapsto \emptyset \mid a \in \M \} \;\In\\
				& & \Let\; i = \nu^\sharp\,u\,(V,L,\sigma)\,u_1\;\In\\
				& & \Let\; \sigma' = \sigma \oplus (\{ \self \mapsto i\} \cup \{ g \mapsto \bot \mid g \in \G\})\;\In\\
				& & \Let\; \rho = \{ [u_1,\emptyset] \mapsto (V',L',\sigma') \}\;\In\\
				& &	(\rho,(V,L,\sigma \oplus \{ x \mapsto i \}))
\end{array}
\]
This function has no effect on the local state apart from setting $x$ to the abstract thread \emph{id} of the newly created thread,
while providing an appropriate initial state to the startpoint of the newly created thread.
For guards and computations on locals, the right-hand-side functions are once more defined in the obvious way. 

Locking a mutex $a$ resets $V\,a$ to $\emptyset$ and updates $L$, whereas unlock side-effects the value of globals to the appropriate unknowns.
\[
\begin{array}{lll}
\sem{[u,S],\lock(a)}^\sharp\Get
	&=&	\Let\;(V,L,\sigma) = \Get\,[u,S]\;\In	\\
	& & \Let\;V' = V \oplus \{ a \mapsto \emptyset\}\;\In \\
	& & \Let\;L' = L \oplus \{ a \mapsto \{S\} \}\;\In\\
	& &	(\emptyset, (V',L', \sigma))	\\[1ex]
\sem{[u,S],\unlock(a)}^\sharp\Get
	&=&	\Let\;(V,L,\sigma) = \Get\,[u,S]\;\In	\\
	& & \Let\;\rho = \{ [g,a, S \setminus \{a\}] \mapsto \sigma\,g \mid g \in \G \}\;\In\\
	& &	(\rho, (V,L, \sigma))
\end{array}
\]
The right-hand-side function for an edge writing to a global $g$ then consists of a strong update to the local copy and addition of $g$ to $V\,a$ for all
mutexes $a$.
For reading from a global $g$, those values $[g,a,S']$ need to be taken into account where $a$ is one of the mutexes acquired in the past and
the intersection of some set in $L\,a$ with the set of mutexes $S'$ held while publishing is empty.
\[
\begin{array}{lll}

\sem{[u,S],g=x}^\sharp\Get	&=&	\Let\;(V,L,\sigma) = \Get\,[u,S]\;\In	\\
			& & \Let\;V' = \{a \mapsto (V\,a \cup \{ g \}) \mid a \in \M \}\;\In\\
			& &	(\emptyset, (V',L,
					\sigma\oplus\{g\mapsto(\sigma\,x)\}))	\\[1ex]
\sem{[u,S],x=g}^\sharp\Get	&=&	\Let\;(V,L,\sigma) = \Get\,[u,S]\;\In	\\
			& &\Let\;d =\bigsqcup\{\eta[g,a,S']\mid a \in \M, g \not\in V\,a, B\in L\,a, B\cap S' =\emptyset\}\;\In\\
			& &	(\emptyset,
				(V,L,\sigma\oplus\{x\mapsto \sigma\,g \sqcup d\}))
\end{array}
\]
In case that
$L\,a=\emptyset$, i.e., if according to the analysis
no thread reaching $u$ with lockset $S$ has ever locked mutex $a$,
then no values from $[g,a,S']$ will be read.
\begin{theorem}
	\emph{Lock-Centered Reading} is sound w.r.t.\ to the trace semantics.
\end{theorem}
\begin{proof}
The proof is deferred to \cref{s:soundness-lock-centered}.
The central issue is to prove that when reading a global $g$,
the restriction to the values of unknowns $[g,a,S']$ as indicated by the
right-hand-side function is sound (see \cref{prop:read'}).\qed
\end{proof}

%% file: analyses/write-centered.tex
\subsection{Write-Centered Reading}\label{s:write-centered}
In this section, we provide a refinement of \emph{Protection-Based Reading} which
abandons the assumption that each global $g$ is write-protected by some fixed set of mutexes $\MM[g]$.
%
In order to lift the assumption, we introduce the additional data-structures $W,P:\G\to\UM$ to be maintained by the
analysis for each unknown $[u,S]$ for program point $u$ and currently held lockset $S$.
The map $W$ tracks for each global $g$ the set of minimal locksets held when $g$ was last written by the ego thread.
At the start of a thread, no global has been written by it yet; hence, we set $W\,g=\emptyset$ for all globals $g$.
The map $P$ on the other hand, tracks for each global $g$ all minimal locksets the ego thread has held since its last write to $g$.
A global $g$ not yet written to by the ego thread is mapped to the \emph{full} power set of mutexes (represented by $\{\emptyset\}$).
The unknowns for a global $g$ now are of the form $[g,a,S,w]$ for mutexes $a$,
background locksets $S$ at $\unlock(a)$ and minimal lockset $w$ when $g$ was last written.

We start by giving the right-hand-side function for the start state at program point $u_0 \in \N$ with the empty lockset $\emptyset$, i.e.,
$[u_0,\emptyset] \sqsupseteq \init^\sharp$ where
\[
	\begin{array}{lll}
		\init^\sharp\,\_ &=& \Let\; W = \{ g \mapsto \emptyset \mid g \in \G \} \;\In\\
		&& \Let\; P = \{ g \mapsto \{\emptyset\} \mid g \in \G \} \;\In\\
		&& \Let\; \sigma = \{ x \mapsto \top \mid x \in \X \} \cup \{ g \mapsto \bot \mid g \in \G \} \;\In\\
		&& (\emptyset,(W,P,\sigma))
	\end{array}
\]
Next comes the right-hand-side function for a thread creating edge.
\[
\begin{array}{lll}
\sem{[u,S],x = \create(u_1)}^\sharp\Get	&=&	\Let\;(W,P,\sigma) = \Get\,[u,S]\;\In	\\
				& & \Let\; W' = \{ g \mapsto \emptyset \mid g \in \G \}\;\In\\
				& & \Let\; P' = \{ g \mapsto \{ \emptyset \} \mid g \in \G \} \;\In\\
				& & \Let\; i = \nu^\sharp\,u\,(W,P,\sigma)\,u_1\;\In\\
				& & \Let\; \sigma' = \sigma \oplus (\{ \self \mapsto i\} \cup \{ g \mapsto \bot \mid g \in \G\})\;\In\\
				& & \Let\; \rho = \{ [u_1,\emptyset] \mapsto (W',P',\sigma') \}\;\In\\
				& &	(\rho,(W,P,\sigma \oplus \{ x \mapsto i \}))
\end{array}
\]
This function has no effect on the local state apart from setting $x$ to the abstract thread \emph{id} of the newly created thread
while providing an appropriate initial state to the startpoint of the newly created thread.
For guards and computations on locals, the right-hand-side functions are once more defined intuitively --- they operate on $\sigma$ only,
leaving $W$ and $P$ unchanged.
While nothing happens at locking, unlocking now updates the data-structure $P$ and additionally
side-effects the current local values for each global $g$ to the corresponding unknowns.
\[
	\begin{array}{lll}
		\sem{[u,S],\lock(a)}^\sharp\Get
		&=&	(\emptyset, \Get\,[u,S]) \\[1ex] 
	\sem{[u,S],\unlock(a)}^\sharp\Get
		&=&	\Let\;(W,P,\sigma) = \Get\,[u,S]\;\In	\\
		& & \Let\;P' = \{ g \mapsto P\,g \sqcup \{S \setminus \{a\}\} \mid g \in \G \}\;\In\\
		& & \Let\;\rho = \{ [g,a,S \setminus \{a\},w] \mapsto \sigma\,g \mid g \in \G, w\in W\,g \}\;\In\\
		& &	(\rho, (W,P',\sigma))
	\end{array}
\]
%
When writing to a global $g$, on top of recording the written value in $\sigma$, $W\,g$ and $P\,g$ are set to the set $\{S\}$
for the current lockset $S$.
When reading from a global $g$,
now only values stored at $\Get\,[g,a,S',w]$ are taken into account, provided
\begin{itemize}
\item	$a\in S$, i.e., $a$ is one of the currently held locks;
\item	$S\cap S'=\emptyset$; i.e., the intersection of the current lockset $S$ with the background lockset at the corresponding operation
	$\unlock(a)$ after the write producing the value stored at this unknown is empty;
\item	$w\cap S''=\emptyset$ for some $S''\in P\,g$, i.e., the background lockset at the write producing the value stored at this unknown
	is disjoint with one of the locksets in $P\,g$. This excludes writes where the ego thread has since its last thread-local write always
	held at least one of the locks in $w$. In this case,
	that write can not have happened between the
	last thread-local write of the reading ego thread and its read;
\item	$a \notin S'''$ for some $S''' \in P\,g$, i.e., $a$ has not been continuously held by the thread since its last write to $g$.
\end{itemize}
Accordingly, we define
\[
\begin{array}{lll}

\sem{[u,S],g=x}^\sharp\Get	&=&	\Let\;(W,P,\sigma) = \Get\,[u,S]\;\In	\\
			& & \Let\;W' = W \oplus \{g \mapsto \{S\} \}\;\In\\
			& & \Let\;P' = P \oplus \{g \mapsto \{S\} \}\;\In\\
			& &	(\emptyset, (W',P',
					\sigma\oplus\{g\mapsto(\sigma\,x)\}))	\\[1ex]
\sem{[u,S],x=g}^\sharp\Get	&=&	\Let\;(W,P,\sigma) = \Get\,[u,S]\;\In	\\
			& &\Let\;d = \sigma\, g \sqcup\bigsqcup\{\eta[g,a,S',w] \mid  a\in S, S\cap S' =\emptyset, \\
			& & \qquad \exists S'' \in P\,g: S'' \cap w = \emptyset, \\
			& & \qquad \exists S''' \in P\,g: a \notin S''' \}\;\In\\
			& &	(\emptyset,
				(W,P,\sigma\oplus\{x\mapsto d\}))\\[1ex]
\end{array}
\]

\begin{example}\label{e:write-centered}
We use integer sets for abstracting int values.
Consider the following concurrent program with global variable \texttt{g} and local variables \texttt{x}, \texttt{y}, and \texttt{z}:
	\begin{center}
	\begin{minipage}[t]{6cm}
	\begin{minted}{c}
main:
  y = |\op{create}|(t1);
  z = |\op{create}|(t2);
  |\op{lock}|(c);
  |\op{lock}|(m|$_{\,\texttt{g}}$|); g = 31; |\op{unlock}|(m|$_{\,\texttt{g}}$|);
  |\op{lock}|(a);
  |\op{lock}|(b);
  |\op{lock}|(m|$_{\,\texttt{g}}$|); x = g; |\op{unlock}|(m|$_{\,\texttt{g}}$|);
  ...

	\end{minted}
	\end{minipage}
	\begin{minipage}[t]{6cm}
	\begin{minted}{c}
t1:
  |\op{lock}|(a);
  |\op{lock}|(b);
  |\op{lock}|(m|$_{\,\texttt{g}}$|); g = 42; |\op{unlock}|(m|$_{\,\texttt{g}}$|);
  |\op{unlock}|(a);
  |\op{lock}|(m|$_{\,\texttt{g}}$|); g = 17; |\op{unlock}|(m|$_{\,\texttt{g}}$|);
  |\op{unlock}|(b);

t2:
  |\op{lock}|(c);
  |\op{lock}|(m|$_{\,\texttt{g}}$|); g = 59; |\op{unlock}|(m|$_{\,\texttt{g}}$|);
  |\op{unlock}|(c);
	\end{minted}
	\end{minipage}
\end{center}
At the read $x=g$, the current lockset is $\{a,b,c,m_g\}$ and in the local state $P\,g = \{ \{ c\}\}$.
The only unknown where all conditions above are fulfilled is the unknown $[g,b,\emptyset,\{b,m_g\}]$ which has value $\{17\}$.
Hence this is the only value read from the unknowns for g and together with the value $\{31\}$ from $\sigma\,g$ the final value for
$x$ is $\{17,31\}$.
This is more precise than either of the analyses presented thus far:
\emph{Protection-Based Reading} cannot exclude any values of $x$ as $\MM[g] = \{m_g\}$, and thus has
$\{17,31,42,59\}$ for $x$.
\emph{Lock-Centered Reading} has $V\,c = \{g\}$ at the read. This excludes the write by $t2$ and
thus results in $\{17,31,42\}$ for $x$.
\qed
\end{example}
%
\begin{theorem}
	\emph{Write-Centered Reading} is sound w.r.t.\ the local trace semantics.
\end{theorem}
\begin{proof}
  The proof is deferred to \cref{s:soundness-write-centered}.
  The central issue is to prove that when reading a global $g$,
  the restriction to the values of unknowns $[g,a,S',w]$ as indicated by the
  right-hand-side function is sound (see \cref{prop:read}).\qed
\end{proof}

\noindent \emph{Protection-Based Reading} from \cref{s:protection-based} is shown
to be an abstraction of this analysis in \cref{s:soundness-protection-based}.

%% file: analyses/combined.tex
\subsection{Combining Write-Centered with Lock-Centered Reading}\label{s:combined}
The analyses described in \cref{s:lock-centered,s:write-centered} are sound, yet
incomparable. This is evidenced by \cref{e:write-centered}, in which \emph{Write-Centered} is more precise
than \emph{Lock-Centered Reading}, and the following example, where the opposite is the case.
\begin{example}\label{e:lock-centered-beats-write-centered}
Assume that we use value sets for abstracting int values.
Consider the following concurrent program with global variable \texttt{g} and local variables \texttt{x} and \texttt{y}:
	\begin{center}
	\begin{minipage}[t]{6cm}
	\begin{minted}{c}
main:
  y = |\op{create}|(t1);
  |\op{lock}|(d);
  |\op{lock}|(a);
  |\op{unlock}|(d);
  |\op{lock}|(m|$_{\,\texttt{g}}$|); x = g; |\op{unlock}|(m|$_{\,\texttt{g}}$|);
  ...

	\end{minted}
	\end{minipage}
	\begin{minipage}[t]{6cm}
	\begin{minted}{c}
t1:
  |\op{lock}|(d);
  |\op{lock}|(a);
  |\op{lock}|(m|$_{\,\texttt{g}}$|); g = 42; |\op{unlock}|(m|$_{\,\texttt{g}}$|);
  |\op{unlock}|(a);
  |\op{lock}|(m|$_{\,\texttt{g}}$|); g = 17; |\op{unlock}|(m|$_{\,\texttt{g}}$|);
  |\op{unlock}|(d);
	\end{minted}
	\end{minipage}
	\end{center}
For \emph{Write-Centered Reading}, both the value at unknowns $[g,a,\{d\},\{d,a,m_g\}]$ with value $\{42\}$ and $[g,d,\emptyset,\{d,m_g\}]$
with value $\{17\}$ are read, resulting in a value of $\{17,42\}$ for $x$.
For \emph{Lock-Centered Reading}, at the read in \texttt{main}, $L\,a = \{\{d\}\}$, and hence $[g,a,\{d\}]$ with value $\{42\}$ does not
fulfill the conditions under which its value is taken into account, resulting in a value of $\{17\}$ for $x$.\qed
\end{example}
To obtain an analysis that is sound and more precise than \emph{Write-Centered} and \emph{Lock-Centered Reading},
both can be combined.
For the combination, we do not rely on a reduced product construction, and instead
exploit the information of all simultaneously tracked data-structures $V,W,P,L$
together for improving the set of writes read at a particular read operation.
For completeness, we list all right-hand-side functions, starting with the one for the initial program point $u_0$ and the empty lockset:
\[
	\begin{array}{lll}
		\init^\sharp\,\_ &=& \Let\; W = \{ g \mapsto \emptyset \mid g \in \G \} \;\In\\
		&& \Let\; P = \{ g \mapsto \{\emptyset\} \mid g \in \G \} \;\In\\
		&& \Let\; V = \{ a \mapsto \emptyset \mid a \in \M \} \;\In\\
		&& \Let\; L = \{ a \mapsto \emptyset \mid a \in \M \} \;\In\\
		&& \Let\; \sigma = \{ x \mapsto \top \mid x \in \X \} \cup \{ g \mapsto \bot \mid g \in \G \} \;\In\\
		&& (\emptyset,(W,P,V,L,\sigma))
	\end{array}
\]
Next, for thread creation, locking and unlocking, and writing to a global:
\[
	\begin{array}{lll}

	\sem{[u,S],x = \create(u_1)}^\sharp\Get	&=&	\Let\;(W,P,V,L,\sigma) = \Get\,[u,S]\;\In	\\
		& & \Let\; W' = \{ g \mapsto \emptyset \mid g \in \G \}\;\In\\
		& & \Let\; P' = \{ g \mapsto \{ \emptyset \} \mid g \in \G \} \;\In\\
		& & \Let\; V' = \{ a \mapsto \emptyset \mid a \in \M \}\;\In\\
		& & \Let\; L' = \{ a \mapsto \emptyset \mid a \in \M \} \;\In\\
		& & \Let\; i = \nu^\sharp\,u\,(P,\sigma)\,u_1\;\In\\
		& & \Let\; \sigma' = \sigma \oplus (\{ \self \mapsto i\} \cup \{ g \mapsto \bot \mid g \in \G\})\;\In\\
		& & \Let\; \rho = \{ [u_1,\emptyset] \mapsto (W',P',V',L',\sigma') \}\;\In\\
		& &	(\rho,(W,P,V,L,\sigma \oplus \{ x \mapsto i \}))\\[1ex]

	\sem{[u,S],\lock(a)}^\sharp\Get
		&=&	\Let\;(W,P,V,L,\sigma) = \Get\,[u,S]\;\In	\\
		& & \Let\;V' = V \oplus \{ a \mapsto \emptyset \} \;\In\\
		& & \Let\;L' = L \oplus \{ a \mapsto \{S\} \}\;\In\\
		& &	(\emptyset, (W,P,V',L',\sigma))	\\[1ex] 

	\sem{[u,S],\unlock(a)}^\sharp\Get
		&=&	\Let\;(W,P,V,L,\sigma) = \Get\,[u,S]\;\In	\\
		& & \Let\;P' = \{ g \mapsto P\,g \sqcup \{S \setminus \{m\}\} \mid g \in \G \}\;\In\\
		& & \Let\;\rho = \{ [g,a, S \setminus \{a\},w] \mapsto \sigma\,g \mid g \in \G,w\in W\,g \}\;\In\\
		& &	(\rho, (W,P',V,L,\sigma))\\[1ex]

	\sem{[u,S],g=x}^\sharp\Get	&=&	\Let\;(W,P,V,L,\sigma) = \Get\,[u,S]\;\In	\\
		& & \Let\;W' = W \oplus \{g \mapsto \{S\} \}\;\In\\
		& & \Let\;P' = P \oplus \{g \mapsto \{S\} \}\;\In\\
		& & \Let\;V' = \{a \mapsto (V\,a \cup \{ g\})  \mid a \in \M\}\;\In\\
		& &	(\emptyset, (W',P',V',L,
			\sigma\oplus\{g\mapsto(\sigma\,x)\}))
	\end{array}
\]
The key point for reading is that the data-structure $P$ is not only used to restrict the set of reads for \emph{Write-Centered Reading},
but can also be used for restricting the set for \emph{Lock-Centered Reading}.
\[
\begin{array}{lll}
\sem{[u,S],x=g}^\sharp\Get	&=&	\Let\;(W,P,V,L,\sigma) = \Get\,[u,S]\;\In	\\
			& &\Let\;d_m = \bigsqcup\{\eta[g,a,S',w]\mid g\not\in V\,a, B\in L\,a, B\cap S' =\emptyset,\\
			& & \qquad \exists S'' \in P\,g: S'' \cap w = \emptyset \}\;\In\\
			& &\Let\;d_g = \bigsqcup\{\eta[g,a,S',w] \mid a\in S, S\cap S' =\emptyset, \\
			& & \qquad \exists S'' \in P\,g: S'' \cap w = \emptyset, \\
			& & \qquad \exists S''' \in P\,g: a \notin S''' \}\;\In\\
			& &\Let\;d = \sigma\, g \sqcup (d_m \sqcap d_g)\;\In\\
			& &	(\emptyset,(W,P,V,L,\sigma\oplus\{x\mapsto d\}))
\end{array}
\]


%% file: correctness/definitions.tex
We begin by making some definitions for local traces that are shared among the subsequent proofs.

For every node $\bar u = (j,u,\sigma)\in\V$ of a local trace $t$, we define the lockset $L_t[\bar u]$
as the set of mutexes whose locks have been acquired by the thread $\sigma\,\self$ and not yet
released. This set can be defined by induction on $j$ by keeping track of the lock and unlock
operations of the thread $\sigma\,\self$.
Furthermore, for any node $\bar u\in\V$, we define $\bar u\da_{t}$ as the local sub-trace
$(\tau',\to'_c,(\to'_a)_{a\in\M})$ of $t$
where $\tau'=(V',E')$ is the subgraph of $\tau$ on the set $V'$ of all nodes $\bar u'$ with
$\bar u'\leq\bar u$, and the relations $\to'_c,(\to'_a)_{a\in\M}$ are the restrictions
of the corresponding relations $\to_c,(\to_a)_{a\in\M}$ of $t$ to $V'$.
In particular $t = (\bar u_t)\da_{t}$ for some node $\bar u_t$ of $t$.
Let $\T_S$ denote the set of local traces $t$ so that $L_t[\bar u_t] = S$ for this node
$\bar u_t$ of $t$.

Moreover, it is convenient for a local trace $t$, to consider the raw trace of the
ego thread, i.e., the thread with thread \emph{id} $\id(t)$.
Let us call this the subgraph of $t$ \emph{raw ego trace}.

Recall that in each local trace $t$, there is for each global $g$ that is ever written in $t$,
a unique \emph{last} write operation $((j-1,u_{j-1},\sigma_{j-1}),g=x,(j,u_j,\sigma_j))$
(at timepoint $j-1$ and program point $u_{j-1}$ with local state $\sigma_{j-1}$).
This means that the endpoint $(j',u',\sigma')$ of any other write operation to $g$ in $t$ precedes
$(j-1,u_{j-1},\sigma_{j-1})$ w.r.t.\ the causality ordering.

Let $\lW_g: \T \to ((\mathbb{N}_0\times\N\times\Sigma) \times \A \times (\mathbb{N}_0\times\N\times\Sigma)) \cup \{\bot \}$
be a function to extract such a last write to $g$ from a local trace,
where $\bot$ indicates no write to $g$ has happened so far in the given local trace.
We call a write occurring at an edge $(\bar u,g=x,\bar u')$
in the raw ego trace \emph{thread-local}.
If there is a thread-local write to a global $g$, there also is a \emph{last}
thread-local write to $g$.
Let $\lTlW_g: \T \to ((\mathbb{N}_0\times\N\times\Sigma) \times \A \times (\mathbb{N}_0\times\N\times\Sigma)) \cup \{\bot \}$ be a function
to extract the last thread-local write to $g$ if it exists, and return $\bot$ otherwise.

Similarly, we call a lock at an edge $(\bar u,\lock(a),\bar u')$
in the raw ego trace \emph{thread-local}.
If there is a thread-local lock of a global $a$, there also is a \emph{last}
thread-local lock of $a$.
Let $\lTlL_a: \T \to ((\mathbb{N}_0\times\N\times\Sigma) \times \A \times (\mathbb{N}_0\times\N\times\Sigma)) \cup \{\bot \}$ be a function
to extract the last thread-local lock of $a$ if it exists, and return $\bot$ otherwise.

Last, we define a function
$\minLSince: \T \to (\mathbb{N}_0\times\N\times\Sigma) \to \UM$ that extracts the
upwards-closed set of minimal locksets the ego thread has held since a given node
of the raw ego trace.
Again, $\minLSince(t,\bar u)$ can computed inductively by
considering the raw ego trace only.

For a set $T$ of local traces, we define the set of values that are written at last thread-local writes to $g$ by
\[
\eval_g(T)= \{ \sigma\,x \mid t \in T,
\lTlW_g\,t = ((j-1,u,\sigma), g = x,\bar u')\}
\]

%% file: correctness/lock-centered.tex
\subsection{Lock-Centered Reading}\label{s:soundness-lock-centered}
\newcommand{\CMU}{{''}} 
Let the constraint system for \emph{Lock-Centered Reading} from \cref{s:lock-centered} be called
$\C_\LockCentered$.
We construct from the constraint system $\C$ for the concrete collecting semantics a system $\C\CMU$
so that the set of unknowns of $\C\CMU$ matches the set of unknowns of $\C_\LockCentered$.
This means that
each unknown $[u]$ for program point $u$ is replaced with the set of unknowns $[u,S]$, $S\subseteq\M$,
while the unknown $[a]$ for a mutex $a$ is replaced with the set of unknowns
$[g,a,S]$, $g\in\G,S\subseteq\M$.
Accordingly, the constraint system $\C\CMU$ consists of these constraints:
\[
  \begin{array}{llll}
    \relax [u_0,\emptyset] & \supseteq &\textbf{fun}\,\_\to (\emptyset,\init) \\
    \relax [u',S\cup\{a\}] & \supseteq & \sem{[u,S],\lock(a)}\CMU &\quad (u,\lock(a),u')\in\E, a\in\M \\
    \relax [u',S\setminus\{a\}] & \supseteq & \sem{[u,S],\unlock(a)}\CMU & \quad (u,\unlock(a),u') \in \E, a \in \M\\
    \relax [u',S] & \supseteq & \sem{[u,S],A}\CMU & \quad (u,A,u') \in \E,\forall a \in\M:\\
    & & & \qquad A\neq \lock(a), A \neq \unlock(a)  \\
  \end{array}
\]
where new right-hand-side functions (relative to the semantics $\sem{e}$ of control-flow edges $e$)
are given by
\[
\begin{array}{lll}
\sem{[u,S],x = \create(u_1)}\CMU\,\Get\CMU	&=&
		\Let\;T = \sem{e}(\Get\CMU\,[u,S])\;\In	\\
&&		(\{[u_1,\emptyset]\mapsto\new\,u_1\,(\Get\CMU\,[u,S])\},T)	\\[1ex]

\sem{[u,S],\lock(a)}\CMU\,\Get\CMU &=& \Let\;T' = \bigcup \{ \Get\CMU\,[g,a,S'] \mid g\in\G, S' \subseteq \M \}\;\In \\
  & & (\emptyset,\sem{e}(\Get\CMU\,[u,S],T'))\\[1ex]

\sem{[u,S],\unlock(a)}\CMU\,\Get\CMU
    &=&	\Let\;T = \sem{e}(\Get\CMU\,[u,S])\;\In	\\
    & & \Let\;\rho = \{[g,a,S\setminus \{ a \}]\mapsto T \mid g \in \G\}\;\In \\
    & &	(\rho, T)\\[1ex]
  \end{array}
\] and
\[
\begin{array}{lll}
\sem{[u,S],x=g}\CMU\,\Get\CMU	&=& (\emptyset,\sem{e}(\Get\CMU\,[u,S]))	\\[1ex]

\sem{[u,S],g=x}\CMU\,\Get\CMU	&=& (\emptyset,\sem{e}(\Get\CMU\,[u,S]))
\end{array}
\]
In contrast to the right-hand-side functions of $\C$, the new right-hand sides now also re-direct
side-effects not to unknowns $[a], a\in\M$, but to appropriate more specific unknowns
$[g,a,S'], g\in\G, a\in\M,S'\subseteq\M$.
For a mapping $\Get$ from the unknowns of $\C$ to $2^\T$, we construct a mapping $\GetPrime$
from the unknowns of $\C\CMU$ to $2^\T$ by
\[
\begin{array}{llll}
  \GetPrime\,[u,S] &=& \Get[u] \cap \T_S \quad & \text{ for } u \in \N, S \subseteq \M \\
  \GetPrime\,[g,a,S] &=& \Get[a] \cap \T_S  & \text { for } g \in \G, a\in\M, S \subseteq \M \\
\end{array}
\]
Thus,
\[
\begin{array}{lll}
\Get[u]	&=& \bigcup\{\GetPrime\,[u,S]\mid S\subseteq\M\}	\\
\Get[a] &=& \bigcup\{\GetPrime\,[g,a,S]\mid g\in\G, S\subseteq\M \}
\end{array}
\]
for all program points $u$ and mutexes $a$.
Moreover, we have,
\begin{proposition}\label{p:GetGetCMU}
The following two statements are equivalent:
\begin{itemize}
  \item $\Get$ is the least solution of $\C$;
  \item $\GetPrime$ is the least solution of $\C\CMU$.
\end{itemize}
\end{proposition}
\begin{proof}
 The proof of \cref{p:GetGetCMU} is by fixpoint induction.  \qed
\end{proof}
The next proposition indicates that the new unknown $[g,a,S]$ collects a superset of local
traces whose last write to the global $g$ can be read by a thread satisfying the specific assumptions
(L0) through (L2) below.
\begin{proposition}\label{prop:read'}
Consider the $i$-th approximation $\Get^i$ to the least solution $\GetPrime$ of constraint system $\C\CMU$,
a control-flow edge $(u,x=g,u')$ of the program, and a local trace $t \in\Get^i\,[u',S]$
in which the last action is $x=g$,
that ends in $\bar u' = (j,u',\sigma)$, i.e., $t = (\bar u')\downarrow_t$.

\noindent For every mutex $a \in M$,
let $L\,a$ denote the singleton set containing the background lockset of the ego thread at the last
thread-local lock of $a$, given that the ego thread has ever acquired $a$ in $t$,
and set $L\,a= \emptyset$ otherwise. 
Also, for every mutex $a$, let $V\,a$ the set of globals written by the ego thread since $a$ was last acquired by it,
or all globals written since the start of the ego thread in case it has never acquired $a$.

\noindent Then, the value $d = \sigma\,x$ that is read for $g$ is produced by a write to $g$
which
\begin{itemize}
  \item either is the last thread-local write to $g$ in $t$; or
  \item is the last \emph{thread-local} write to $g$ in some local trace stored at $\Get^{i'}\,[g,a,S']$ for some $i' < i$, i.e.,
	\[
	d\in\eval_g (\Get^{i'}\,[g,a,S'])
	\]
  where
	\begin{enumerate}
	\item[(L0)]	$a$ has been acquired by the ego thread, i.e., $L\, a \neq \emptyset$,
	\item[(L1)]	$L\,a = \{B\}$ such that $B\cap S' = \emptyset$,
	\item[(L2)]	$g \not\in V\,a$,
	\end{enumerate}
\end{itemize}
\end{proposition}
\begin{proof}
The proof is by fixpoint induction.
We prove that the values read non-thread-locally for a global $g$ at some $(u,x=g,u')$ when
constructing the local traces of $\Get^i$, are the last \emph{thread-local} writes of a local trace
$t'$ ending in an unlock operation that is added to
$\Get^{i'}\,[g,a,S']$ in some prior iteration $i' < i$ for certain $a$ and $S'$ satisfying (L1) and (L2).

This property holds for $i=0$, as in $\Get^0$, all unknowns for program points and currently held locksets
(except for the initial program point and the empty lockset) are $\emptyset$, and therefore no reads from globals or unlocks can happen.

For the induction step $i > 0$, there are two proof obligations: First that the property holds
for all reads from a global, and additionally that all traces ending in an unlock operation are
once more side-effected to appropriate unknowns in this iteration.

For the first obligation, consider a local trace $t \in\Get^i\,[u',S]$ where the last action is
$x=g$.
There is a last write to $g$ in $t$:
\[
\lW_g\,t = ((j'-1,u_{j'-1},\sigma_{j'-1}), g = x', \bar u'') = l
\]

Let $i_0 = \id\,t$ and $i_1 = \sigma_{j'-1}\,\self$ the thread \emph{id}s of the reading ego thread
and the thread performing the last write, respectively.
We distinguish two cases:\\
\emph{Case 1: $i_0 = i_1$.} The last write is thread-local to $t$ (and $l$ is therefore also the last thread-local write to $g$ in $t$).\\
\emph{Case 2: $i_0 \neq i_1$.} The last write is not thread-local.
Consider the maximal sub-trace $t''$ of $t$ with $\id(t'') = i_1$.
%
Let $a$ denote the last (w.r.t.\ to the program order) mutex unlocked by $i_1$ in $t''$
for which the following additional conditions hold:
\begin{itemize}
  \item $a$ is unlocked in $t''$ by $i_1$ after the last write to $g$ ($l$)
  \item $a$ has also been locked by $i_0$ in $t$
  \item the last lock of $a$ by $i_0$ succeeds the unlock of $a$ by $i_1$ w.r.t.\ the causality order
$\leq$ of $t$.
\end{itemize}
We observe that there is at least one mutex, namely $m_g$, which is unlocked by $i_1$ after its last write to $g$
and subsequently locked by $i_0$ before $g$ is read.
Let $S'$ denote the background lockset held at the last action $\unlock(a)$ by $i_1$.
Let $B$ denote the background lockset at the last $\lock(a)$ of the ego thread $i_0$, i.e., $L\,a=\{B\}$.

First, assume that property (L1) is violated for $a$.
Then there is some $c\in B\cap S'$, implying that $t''$ cannot be a local subtrace of $t$.
To see this, assume for a contradiction that $c$ later is unlocked by $i_1$ so that $i_0$ can
acquire $c$.
Then, however, the conditions are also fulfilled for $c$, meaning that $a$ is not the \emph{last} such mutex.
If on the other hand, $c$ is never unlocked by $i_1$ in $t''$, thread $i_0$ will not be able to acquire
$c$ before its last operation $\lock(a)$, yielding a contradiction.

Accordingly, now assume that $B\cap S'=\emptyset$.
We claim that then also $g\not\in V\,a$ must hold.
If this were not the case, some thread-local write to $g$  by $i_0$ has happened after
the last operation $\lock(a)$. Then, however, the write in $t''$ happens before this write to $g$ by $i_0$,
and is thus not the last write, yielding once again a contradiction.

The local trace $t'$ which is the sub-trace of $t''$ ending in this $\unlock(a)$ of $i_1$ thus contains the last write to $g$ in $t$.
It was constructed during some earlier iteration $i' < i$ and, by induction hypothesis, added to $\Get^{i'}[g,a,S']$
during the $i'$-th iteration.
We conclude that the value $d$ read from $g$ by $i_0$ is given by
$d = \sigma_{j'-1}\,x' \in \eval_g (\Get^{i'}\,[g,a,S']) \subseteq  \eval_g (\Get^{i}\,[g,a,S'])$.

It remains to show that any trace $t$ with $\last(t)=\unlock(a)$, $a\in\M$ ending in
$\bar u''$, i.e., $t=(\bar u'')\downarrow_t$,
produced in this iteration $i$ is side-effected to $\Get^i\,[g,a,S]$
where $S = L_t[\bar u'']$. This, however, follows from the construction of $\C\CMU$.
\qed
\end{proof}
Let us now relate the post-solutions of $\C\CMU$ and $C_\LockCentered$ to each other.
As a first step, we define a function
$\beta$ that extracts from a local trace $t$ for each mutex $a$
\begin{itemize}
  \item the set $V\,a$ of global variables that were written by the ego thread since $a$ was last acquired by it,
  or all global variables written since the start of the ego thread in case it has never acquired $a$; and
  \item the set $L\,a$ containing the background lockset when $a$ was acquired by the ego thread last.
\end{itemize}
Additionally, $\beta$ extracts a map $\sigma$ that contains
the values of the locals at the sink of $t$ as well as the last-written thread-local values of globals. Thus, we define
\[
  \begin{array}{lll}
  \beta\,t &=& (V,L,\sigma) \qquad \text{where} \\[1ex]

  V &=& \{ a \mapsto \{ g \mid g\in\G, (\_,g=x, \bar u ') = \lTlW_g\,t, \bar u \leq \bar u'  \} \mid  a \in \M,\\
  && \qquad (\_,\lock(a), \bar u) = \lTlL_a\,t\}\\
  && \cup\; \{ a \mapsto \{ g \mid g\in\G, (\_,g=x,\_) = \lTlW_g\,t \} \mid a \in \M,\\
  && \qquad \bot = \lTlL_a\,t \} \\[1ex]

  L &=& \{ a \mapsto \{ L_t[\bar u] \} \mid a \in \M, (\bar u,\lock(a),\_) = \lTlL_a\,t \} \\
  && \cup \; \{ a \mapsto \emptyset \mid a \in \M, \bot = \lTlL_a\,t   \}\\[1ex]

  \sigma &=&  \{ x \mapsto \{t(x)\} \mid x \in \X \} \cup \{ g \mapsto \emptyset \mid g \in \G, \bot = \lTlW_{g}\,t\}\\
  && \cup \; \{ g \mapsto \{\sigma_{j-1}\,x\} \mid g\in\G, ((j-1,u_{j-1},\sigma_{j-1}),g=x,\_) = \lTlW_{g}\,t\}
  \end{array}
\]
\noindent
This abstraction function $\beta$ is used to specify concretization functions for the values of unknowns $[u,S]$ for program points
and currently held locksets as well as for unknowns $[g,a,S]$.
\[
\begin{array}{lll}
  \gamma_{u,S}(V^\sharp,L^\sharp,\sigma^\sharp) &=& \{ t \in \T_S \mid \loc\,t=u, \beta\,t = (V,L,\sigma),\\
  && \quad  \sigma \subseteq \gamma_\D\circ\sigma^\sharp, V \sqsubseteq V^\sharp, L \sqsubseteq L^\sharp \}\\[1ex]
\end{array}
\]
where $\subseteq$ and $\sqsubseteq$ are extended point-wise from domains to maps into domains.
Moreover,
\[
\begin{array}{lll}
  \gamma_{g,a,S}(v) &=& \{ t \in \T_S \mid \last\,t=\unlock(a),\\
  && \quad ((j-1,u_{j-1},\sigma_{j-1}),g=x, \_) = \lTlW_g\,t, \sigma_{j-1}\,x \in \gamma_\D(v) \}\\
  && \cup\;\{ t \in \T_S \mid \last\,t=\unlock(a), \lTlW_g\,t = \bot \}
\end{array}
\]
where $\gamma_\D:\D\to 2^\V$ is the concretization function for abstract values in $\D$.

Let $\Get_\LockCentered$ be a post-solution of $\C_\LockCentered$. We then construct from it a mapping $\Get'$ by:
\[
  \begin{array}{llll}
    \Get' [u,S] &=& \gamma_{u,S}(\Get_\LockCentered\,[u,S]) \qquad & u\in\N, S\subseteq\M\\
    \Get' [g,a,S] &=& \gamma_{g,a,S}(\Get_\LockCentered\,[g,a,S]) \qquad & g\in\G, a\in\M, S\subseteq\M
  \end{array}
\]
Altogether, the correctness of the constraint system $\C_\LockCentered$ follows from the
following theorem.

\begin{theorem}\label{t:lock-centered}
    Every post-solution of $\C_\LockCentered$ is sound w.r.t.\ the local trace semantics.
\end{theorem}
\begin{proof}
Recall from \cref{p:GetGetCMU}, that the least solution of $\C\CMU$ is sound
w.r.t.\ the local trace semantics as specified by the constraint system $\C$.

It thus suffices to prove that the mapping $\Get'$
as constructed above, is a post-solution of the constraint system $\C\CMU$.
For that, we verify by fixpoint induction that for the $i$-th approximation $\Get^i$
to the least solution $\GetPrime$ of $\C\CMU$, $\Get^i \subseteq \Get'$ holds.
To this end, we verify for the start point $u_0$ and the empty lockset, that
\[
  (\emptyset,\init) \subseteq (\Get',\Get'\,[u_0,\emptyset])
\] holds and for each edge $(u,A,v)$ of the control-flow graph
and each possible lockset $S$, that
\[
  \sem{[u,S],A}\CMU\, \Get^{i-1} \subseteq(\Get',\Get'\,[v,S'])
\] holds.

\noindent First, for the start point $u_0$ and the empty lockset:
\[
  (\emptyset,\init) \subseteq (\Get',\Get'\,[u_0,\emptyset])
\]
As there are no side-effects triggered, it suffices to check that $\init \subseteq \Get'\,[u_0,\emptyset]$.
\[
	\begin{array}{lll}
		\init^\sharp_\LockCentered\,\_ &=& \Let\; V^\sharp = \{ a \mapsto \emptyset \mid a \in \M \} \;\In\\
		&& \Let\; L^\sharp = \{ a \mapsto \emptyset \mid a \in \M \} \;\In\\
		&& \Let\; \sigma^\sharp = \{ x \mapsto \top \mid x \in \X \} \cup \{ g \mapsto \bot \mid g \in \G \} \;\In\\
		&& (\emptyset,(V^\sharp,L^\sharp,\sigma^\sharp))
	\end{array}
\]
Let $\Get_\LockCentered\,[u_0,\emptyset] = (V^{\sharp'},L^{\sharp'},\sigma^{\sharp'})$ the
value provided by $\Get_\LockCentered$ for the start point and the empty lockset.
Since $\Get_\LockCentered$ is a post-solution of $\C_\LockCentered$,
$V^\sharp \sqsubseteq V^{\sharp'}$, $L^\sharp \sqsubseteq L^{\sharp'}$, and $\sigma^\sharp \sqsubseteq \sigma^{\sharp'}$ all hold.
Then, by definition:
\[
\begin{array}{lll}
\Get'[u_0,\emptyset] = \gamma_{u_0,\emptyset}(V^{\sharp'},L^{\sharp'},\sigma^{\sharp'}) &=&
    \{ t \in \T_\emptyset \mid \loc\,t=u_0, \beta\,t = (V,L,\sigma),\\
    && \quad  \sigma \subseteq \gamma_\D\circ\sigma^{\sharp'}, V \sqsubseteq V^{\sharp'}, L \sqsubseteq L^{\sharp'} \}
\end{array}
\]
For every trace $t \in \init$, let
\[
\begin{array}{lll}
  \beta\,t &=& (V,L,\sigma) \text{ where: }\\[1ex]

  V &=& \{ a \mapsto \{ g \mid g\in\G, (\_,g=x, \bar u ') = \lTlW_{g}\,t, \bar u \leq \bar u'  \} \mid a\in\M, \\
  && \qquad (\_,\lock(a), \bar u) = \lTlL_a\,t\}\\
  && \cup \; \{ a \mapsto \{ g \mid g\in\G, (\_,g=x, \_) = \lTlW_{g}\,t \} \mid a\in\M,\\
  && \qquad \bot = \lTlL_a\,t \} \\
  &=& \{ a \mapsto \emptyset \mid a \in \M \}\\[1ex]

  L &=& \{ a \mapsto \{ L_{t}[\bar u] \} \mid a \in \M, (\bar u,\lock(a), \_) = \lTlL_a\,t \} \\
  && \cup \; \{ a \mapsto \emptyset \mid a \in \M, \bot = \lTlL_a\,t   \}\\
  &=& \{ a \mapsto \emptyset \mid a \in \M \}\\[1ex]

  \sigma &=&  \{ x \mapsto \{t(x)\} \mid x \in\X \} \cup \{ g \mapsto \emptyset \mid g \in \G, \bot = \lTlW_{g}\,t\} \\
  && \cup \; \{ g \mapsto \{\sigma_{j-1}\,x\}  \mid g \in \G, ((j-1,u_{j-1},\sigma_{j-1}),g=x, \_) = \lTlW_{g}\,t\} \\
  &=& \{ x \mapsto \{t(x)\} \mid x \in\X \} \cup \{ g \mapsto \emptyset \mid g \in \G\}
\end{array}
\]
Thus,
\[
  \begin{array}{lll}
    L &=& L^{\sharp} \sqsubseteq L^{\sharp'}\\
    V &=& V^{\sharp} \sqsubseteq V^{\sharp'}\\
    \sigma &=& \{ x \mapsto \{t(x)\} \mid x \in\X \} \cup \{ g \mapsto \emptyset \mid g \in \G\} \\
    &\subseteq& \gamma_\D\circ(\{ x \mapsto \top \mid x \in\X \} \cup \{ g \mapsto \bot \mid g \in \G\}) = \gamma_\D\circ\sigma^\sharp \\
    &\subseteq& \gamma_\D\circ\sigma^{\sharp'}
  \end{array}
\]
Altogether, $t \in \Get'\,[u_0,\emptyset]$ for all $t \in \init$.

\vspace{1em}
\noindent Next, we verify for each edge $(u,A,v)$ of the control-flow graph
and each possible lockset $S$, that
\[
  \sem{[u,S],A}\CMU\, \Get^{i-1} \subseteq(\Get',\Get'\,[v,S'])
\] holds.

\vspace{1em}
\noindent We first consider a write to a global $g=x$.
\[
\begin{array}{lll}
    \sem{[u,S],g=x}\CMU\,\Get\CMU	&=& (\emptyset,\sem{e}(\Get\CMU\,[u,S]))\\[1ex]

    \sem{[u,S],g=x}^\sharp_\LockCentered\Get_\LockCentered	&=&	\Let\;(V^\sharp,L^\sharp,\sigma^\sharp) = \Get_\LockCentered\,[u,S]\;\In	\\
    & & \Let\;V^{\sharp''} = \{a \mapsto (V^{\sharp}\,a \cup \{ g \}) \mid a \in \M \}\;\In\\
    & & \Let\;\sigma^{\sharp''} = \sigma^\sharp\oplus\{g\mapsto(\sigma^\sharp\,x)\} \;\In\\
    & &	(\emptyset, (V^{\sharp''},L^\sharp,\sigma^{\sharp''}))
\end{array}
\]
Let $\Get_\LockCentered\,[u,S] = (V^\sharp,L^\sharp,\sigma^\sharp)$ and
$\Get_\LockCentered\,[v,S] = (V^{\sharp'},L^{\sharp'},\sigma^{\sharp'})$ the
value provided by $\Get_\LockCentered$ for the end point of the given control-flow edge and lockset.
Since $\Get_\LockCentered$ is a post-solution of $\C_\LockCentered$,
$V^{\sharp''} \sqsubseteq V^{\sharp'}$, $L^{\sharp} \sqsubseteq L^{\sharp'}$, and $\sigma^{\sharp''} \sqsubseteq \sigma^{\sharp'}$
hold.
Then, by definition:
\[
\begin{array}{lll}
\Get'[v,S] = \gamma_{v,S}(V^{\sharp'},L^{\sharp'},\sigma^{\sharp'}) &=&
    \{ t \in \T_S \mid \loc\,t=v, \beta\,t = (V,L,\sigma),\\
    && \quad  \sigma \subseteq \gamma_\D\circ\sigma^{\sharp'}, V \sqsubseteq V^{\sharp'}, L \sqsubseteq L^{\sharp'} \}
\end{array}
\]
For every trace $t \in \Get^{i-1}\,[u,S]$, let $\beta\,t = (V,L,\sigma)$. By induction hypothesis,
$V \sqsubseteq V^\sharp$, $L \sqsubseteq L^\sharp$, and
$\sigma \subseteq \gamma_\D\circ\sigma^\sharp$. Let $t'= \sem{e}\{ t\}$, then $\loc(t') = v$, $t' \in \T_S$, and
\[
\begin{array}{lll}
  \beta\,t' &=& (V',L',\sigma') \text{ where: }\\[1ex]

  V' &=& \{ a \mapsto \{ g' \mid g'\in\G, (\_,g'=x, \bar u ') = \lTlW_{g'}\,t', \bar u \leq \bar u'  \} \mid a\in\M, \\
  && \qquad (\_,\lock(a), \bar u) = \lTlL_a\,t'\}\\
  && \cup \; \{ a \mapsto \{ g' \mid g'\in\G, (\_,g'=x, \_) = \lTlW_{g'}\,t' \} \mid a\in\M,\\
  && \qquad \bot = \lTlL_a\,t' \} \\
  &=& \{ a \mapsto (V\,a \cup \{ g\}) \mid a \in \M \}\\[1ex]

  L' &=& \{ a \mapsto \{ L_{t'}[\bar u] \} \mid a \in \M, (\bar u,\lock(a), \_ ) = \lTlL_a\,t' \} \\
  && \cup \; \{ a \mapsto \emptyset \mid a \in \M, \bot = \lTlL_a\,t'   \}\\
  &=& L\\[1ex]

  \sigma' &=&  \{ x \mapsto \{t'(x)\} \mid x \in\X \} \cup \{ g' \mapsto \emptyset \mid g' \in \G, \bot = \lTlW_{g'}\,t'\} \\
  && \cup \; \{ g' \mapsto \{\sigma_{j-1}\,x\}  \mid g' \in \G, ((j-1,u_{j-1},\sigma_{j-1}),g'=x, \_ ) = \lTlW_{g'}\,t'\} \\
  &=& \sigma \oplus \{ g \mapsto \sigma\,x \}
\end{array}
\]
Thus,
\[
  \begin{array}{lll}
    L' &=& L \sqsubseteq L^{\sharp} \sqsubseteq L^{\sharp'}\\
    V' &=& \{ V\,a \cup \{g\} \mid a \in \M \} \sqsubseteq \{a \mapsto V^\sharp\,a \cup \{ g \} \mid a \in \M \} = V^{\sharp''} \sqsubseteq V^{\sharp'}\\
    \sigma' &=& \sigma \oplus \{ g \mapsto \sigma\,x \} \subseteq \gamma_\D\circ(\sigma^\sharp \oplus \{ g \mapsto \sigma^\sharp\,x\}) =
  \gamma_\D\circ\sigma^{\sharp''}\subseteq \gamma_\D\circ\sigma^{\sharp'}
  \end{array}
\]
Altogether, $t' \in \Get'\,[v,S]$ for all $t \in \Get^{i-1}\,[u,S]$.
We conclude that the return value of
$\sem{[u,S],g=x}\CMU\,\Get^{i-1}$ is subsumed by the value $\Get'\,[v,S]$
and since the constraint causes no side-effects, the claim holds.

\vspace{1em}
\noindent Next, for a read from a global $x = g$:
\[
\begin{array}{lll}
  \sem{[u,S],x=g}\CMU\,\Get\CMU	&=& (\emptyset,\sem{e}(\Get\CMU\,[u,S]))	\\[1ex]

  \sem{[u,S],x=g}^\sharp_\LockCentered\,\Get_\LockCentered	&=&
            \Let\;(V^\sharp,L^\sharp,\sigma^\sharp) = \Get_\LockCentered\,[u,S]\;\In	\\
			& &\Let\;d = \sigma^\sharp\, g \sqcup\bigsqcup\{\Get_\LockCentered\,[g,a,S'] \mid a \in \M,\\
            && \quad g \not\in V^\sharp\,a, B\in L^\sharp\,a, B\cap S' =\emptyset\}\;\In\\
      & &\Let\;\sigma^{\sharp''} = \sigma^\sharp\oplus\{x\mapsto d\}\;\In\\
			& &	(\emptyset,
				(V^\sharp,L^\sharp,\sigma^{\sharp''}))\\
\end{array}
\]
Let $\Get_\LockCentered\,[u,S] = (V^\sharp,L^\sharp,\sigma^\sharp)$ and
$\Get_\LockCentered\,[v,S] = (V^{\sharp'},L^{\sharp'},\sigma^{\sharp'})$ the
value provided by $\Get_\LockCentered$ for the end point of the given control-flow edge and lockset.
Since $\Get_\LockCentered$ is a post-solution of $\C_\LockCentered$,
$V^{\sharp} \sqsubseteq V^{\sharp'}$, $L^{\sharp} \sqsubseteq L^{\sharp'}$, and $\sigma^{\sharp''} \sqsubseteq \sigma^{\sharp'}$
hold.
Then, by definition:
\[
\begin{array}{lll}
\Get'[v,S] = \gamma_{v,S}(V^{\sharp'},L^{\sharp'},\sigma^{\sharp'}) &=&
    \{ t \in \T_S \mid \loc\,t=v, \beta\,t = (V,L,\sigma),\\
    && \quad  \sigma \subseteq \gamma_\D\circ\sigma^{\sharp'}, V \sqsubseteq V^{\sharp'}, L \sqsubseteq L^{\sharp'} \}
\end{array}
\]
For every trace $t \in \Get^{i-1}\,[u,S]$, let $\beta\,t = (V,L,\sigma)$. By induction hypothesis,
$V \sqsubseteq V^\sharp$, $L \sqsubseteq L^\sharp$, and
$\sigma \subseteq \gamma_\D\circ\sigma^\sharp$. Let $t'= \sem{e}\{ t\}$, then $\loc(t') = v$, $t' \in \T_S$, and
\[
\begin{array}{lll}
  \beta\,t' &=& (V',L',\sigma') \text{ where: }\\[1ex]

  V' &=& \{ a \mapsto \{ g' \mid g'\in\G, (\_,g'=x, \bar u ') = \lTlW_{g'}\,t', \bar u \leq \bar u'  \} \mid a \in \M,\\
  && \qquad (\_,\lock(a), \bar u) = \lTlL_a\,t'\}\\
  && \cup \; \{ a \mapsto \{ g' \mid g'\in\G, (\_,g'=x, \_) = \lTlW_{g'}\,t' \} \mid a\in\M,\\
  && \qquad \bot = \lTlL_a\,t' \}\\
  &=& V \\[1ex]

  L' &=& \{ a \mapsto \{ L_{t'}[\bar u] \} \mid a \in \M, (\bar u,\lock(a), \bar u ') = \lTlL_a\,t' \} \\
  && \cup \; \{ a \mapsto \emptyset \mid a \in \M, \bot = \lTlL_a\,t'   \}\\
  &=& L\\[1ex]

  \sigma' &=&  \{ x \mapsto \{t'(x)\} \mid x \in\X \} \cup \{ g' \mapsto \emptyset \mid g' \in \G, \bot = \lTlW_{g'}\,t'\} \\
  && \quad \cup \; \{ g' \mapsto \{\sigma_{j-1}\,x\}  \mid g' \in \G, ((j-1,u_{j-1},\sigma_{j-1}),g'=x, \_) = \lTlW_{g'}\,t'\} \\
  &=& \sigma \oplus \{ x \mapsto \{t'(x)\} \} \\
  &=& \sigma \oplus \{ x \mapsto \{\sigma_{j'-1}\,x'\} \mid \lW_g\,t' = ((j'-1,u_{j'-1},\sigma_{j'-1}), g = x', \_)  \}\\
  &=& \sigma \oplus \{ x \mapsto \{\sigma_{j'-1}\,x'\} \mid \lW_g\,t = ((j'-1,u_{j'-1},\sigma_{j'-1}), g = x', \_)  \}
\end{array}
\]
Thus, $V = V' \sqsubseteq V^{\sharp} \sqsubseteq V^{\sharp'}$ and  $L = L' \sqsubseteq L^{\sharp} \sqsubseteq L^{\sharp'}$.
Also $\sigma\,y = \sigma'\,y$ and therefore, $\sigma'\,y\subseteq(\gamma_\D\circ\sigma^{\sharp'})\,y$
for $y \not\equiv x$.
For $y \equiv x$, we consider two cases:
\begin{itemize}
  \item Last write to $g$ is thread-local ($\lTlW_g\,t = ((j'-1,u_{j'-1},\sigma_{j'-1}), g = x', \bar u'')$):
  Then $\sigma\,g= \{\sigma_{j'-1}\,x'\} \subseteq (\gamma\circ\sigma^\sharp)\,g$, thus
 $\sigma'\,x \subseteq (\gamma\circ\sigma^{\sharp''})\,x$
  and accordingly, $\sigma' \subseteq \gamma\circ\sigma^{\sharp'}$.
  \item Last write to $g$ is non-thread-local. Then
  \[
    \begin{array}{lll}
    \sigma'\,x &\subseteq& \bigcup \{ \eval_g(\Get^{i-1}\,[g,a,S']) \mid a \in \M, g \not\in V,a, B\in L\,a,\\
    && \qquad B\cap S' =\emptyset\ \} \qquad \text{(By \cref{prop:read'})} \\
    &\subseteq&   \bigcup \{ \eval_g(\Get'\,[g,a,S']) \mid a \in \M, g \not\in V^\sharp\,a, B\in L^\sharp\,a,\\
    && \qquad B\cap S' =\emptyset\ \}  \qquad \text{(By Induction Hypothesis)} \\
    &\subseteq&  \bigcup \{ \gamma_\D(\Get_\LockCentered\,[g,a,S']) \mid a \in \M, g \not\in V^\sharp\,a, B\in L^\sharp\,a,\\
    && \qquad B\cap S' =\emptyset\ \} \\
    &\subseteq& \gamma_\D (\bigsqcup \{ (\Get_\LockCentered\,[g,a,S']) \mid a \in \M, g \not\in V^\sharp\,a, B\in L^\sharp\,a,\\
    && \qquad B\cap S' =\emptyset\  \} \sqcup \sigma^\sharp\,g) \\
    &=& (\gamma_\D\circ\sigma^{\sharp''})\,x
       \subseteq (\gamma_\D\circ\sigma^{\sharp'})\,x  \\
    \end{array}
  \]
  and thus $\sigma' \subseteq \gamma_\D\circ\sigma^{\sharp'}$.
\end{itemize}
Altogether, $t' \in \Get'\,[v,S]$ for all $t \in \Get^{i-1}\,[u,S]$.
We conclude that the return value of
$\sem{[u,S],x=g}\CMU\,\Get^{i-1}$ is subsumed by the value $\Get'\,[v,S]$
and since the constraint causes no side-effects, the claim holds.

\vspace{1em}
\noindent Next, for $\lock(a)$, $a \in \M$:
\[
\begin{array}{lll}
  \sem{[u,S],\lock(a)}\CMU\,\Get\CMU &=& \Let\;T' = \bigcup \{ \Get\CMU\,[g,a,S] \mid g\in\G, S \subseteq \M \}\;\In \\
  & & (\emptyset,\sem{e}(\Get\CMU\,[u,S],T'))\\[1ex]

  \sem{[u,S],\lock(a)}^\sharp_\LockCentered\Get_\LockCentered
	&=&	\Let\;(V^\sharp,L^\sharp,\sigma^\sharp) = \Get_\LockCentered\,[u,S]\;\In	\\
	& & \Let\;V^{\sharp''} = V^\sharp \oplus \{ a \mapsto \emptyset\}\;\In \\
	& & \Let\;L^{\sharp''} = L^\sharp \oplus \{ a \mapsto \{S\} \}\;\In\\
	& &	(\emptyset, (V^{\sharp''},L^{\sharp''}, \sigma^\sharp))	\\[1ex]
\end{array}
\]
Let $\Get_\LockCentered\,[u,S] = (V^\sharp,L^\sharp,\sigma^\sharp)$ and
$\Get_\LockCentered\,[v,S\cup \{a\}] = (V^{\sharp'},L^{\sharp'},\sigma^{\sharp'})$ the
value provided by $\Get_\LockCentered$ for the end point of the given control-flow edge and lockset.
Since $\Get_\LockCentered$ is a post-solution of $\C_\LockCentered$,
$V^{\sharp''} \sqsubseteq V^{\sharp'}$, $L^{\sharp''} \sqsubseteq L^{\sharp'}$, and $\sigma^{\sharp} \sqsubseteq \sigma^{\sharp'}$
hold.
Then, by definition:
\[
\begin{array}{lll}
\Get'[v,S\cup \{a\}] = \gamma_{v,S \cup \{ a\}}(V^{\sharp'},L^{\sharp'},\sigma^{\sharp'}) &=&
    \{ t \in \T_{S \cup \{ a\}} \mid \loc\,t=v, \beta\,t = (V,L,\sigma),\\
    && \quad  \sigma \subseteq \gamma_\D\circ\sigma^{\sharp'}, V \sqsubseteq V^{\sharp'}, L \sqsubseteq L^{\sharp'} \}
\end{array}
\]
For every trace $t \in \Get^{i-1}\,[u,S]$, let $\beta\,t = (V,L,\sigma)$. By induction hypothesis,
$V \sqsubseteq V^\sharp$, $L \sqsubseteq L^\sharp$, and
$\sigma \subseteq \gamma_\D\circ\sigma^\sharp$. For any
$t' \in \sem{e}(\{t\}, \bigcup \{ \Get\CMU\,[g,a,S] \mid g\in\G, S \subseteq \M, w \subseteq \M \})$,
$\loc(t') = v$, $t' \in \T_{S\cup\{a\}}$, and
\[
\begin{array}{lll}
  \beta\,t' &=& (V',L',\sigma') \text{ where: }\\[1ex]

  V' &=& \{ a' \mapsto \{ g \mid g\in\G, (\_,g=x, \bar u ') = \lTlW_g\,t', \bar u \leq \bar u'  \} \mid a' \in \M, \\
  && \qquad(\_,\lock(a'), \bar u) = \lTlL_{a'}\,t'\}\\
  && \cup\; \{ a' \mapsto \{ g \mid g\in\G, (\_,g=x, \_) = \lTlW_g\,t' \} \mid a' \in \M,\\
  && \qquad \bot = \lTlL_a\,t' \}\\
  &=& V \oplus \{ a \mapsto \emptyset \} \\[1ex]

  L' &=& \{ a' \mapsto \{ L_{t'}[\bar u] \} \mid a' \in \M, (\bar u,\lock(a'), \_) = \lTlL_{a'}\,t' \} \\
  &&\cup \; \{ a' \mapsto \emptyset \mid a' \in \M, \bot = \lTlL_{a'}\,t' \}\\
  &=& L \oplus \{ a \mapsto \{S\} \}\\[1ex]

  \sigma' &=&  \{ x \mapsto \{t'(x)\} \mid x \in\X \} \cup \{ g \mapsto \emptyset \mid g \in \G, \bot = \lTlW_{g}\,t'\} \\
  &&\cup \; \{ g \mapsto \{\sigma_{j-1}\,x\}  \mid g \in \G, ((j-1,u_{j-1},\sigma_{j-1}),g=x, \_) = \lTlW_{g}\,t'\} \\
  &=& \sigma
\end{array}
\]
Therefore,
\[
\begin{array}{lll}
  V' &=& V \oplus \{ a \mapsto \emptyset \} \sqsubseteq V^{\sharp} \oplus \{ a \mapsto \emptyset \} = V^{\sharp''} \sqsubseteq V^{\sharp'}\\
  L' &=& L \oplus \{ a \mapsto \{ S \} \} \sqsubseteq L^{\sharp} \oplus \{ a \mapsto \{S\}\} = L^{\sharp''} \sqsubseteq L^{\sharp'}\\\
  \sigma' &=& \sigma \subseteq \gamma_\D\circ\sigma^{\sharp} = \gamma_\D\circ\sigma^{\sharp''} \subseteq \gamma_\D\circ\sigma^{\sharp'}
\end{array}
\]
Altogether, $t' \in \Get'\,[v,S\cup\{a\}]$ for all $t \in \Get^{i-1}\,[u,S]$.
We conclude that the return value of
$\sem{[u,S],\lock(a)}'\, \Get^{i-1}$ is subsumed by the value $\Get'\,[v,S\cup\{a\}]$
and since the constraint causes no side-effects, the claim holds.

\vspace{1em}
\noindent Next, for $\unlock(a)$ $a \in \M$:
\[
\begin{array}{lll}
  \sem{[u,S],\unlock(a)}\CMU\,\Get\CMU
  &=&	\Let\;T = \sem{e}(\Get\CMU\,[u,S])\;\In	\\
  & & \Let\;\rho = \{[g,a,S\setminus \{ a \}]\mapsto T \mid g \in \G\}\;\In \\
  & &	(\rho, T)\\[1ex]

  \sem{[u,S],\unlock(a)}^\sharp_\LockCentered\Get_\LockCentered
	&=&	\Let\;(V^\sharp,L^\sharp,\sigma^\sharp) = \Get_\LockCentered\,[u,S]\;\In	\\
	& & \Let\;\rho^\sharp = \{ [g,a, S \setminus \{a\}] \mapsto \sigma^\sharp\,g \mid g \in \G \}\;\In\\
	& &	(\rho, (V^\sharp,L^\sharp, \sigma^\sharp))
\end{array}
\]
Let $\Get_\LockCentered\,[u,S] = (V^\sharp,L^\sharp,\sigma^\sharp)$ and
$\Get_\LockCentered\,[v,S\setminus \{a\}] = (V^{\sharp'},L^{\sharp'},\sigma^{\sharp'})$ the
value provided by $\Get_\LockCentered$ for the end point of the given control-flow edge and lockset.
Since $\Get_\LockCentered$ is a post-solution of $\C_\LockCentered$,
$V^{\sharp} \sqsubseteq V^{\sharp'}$, $L^{\sharp} \sqsubseteq L^{\sharp'}$, and $\sigma^{\sharp} \sqsubseteq \sigma^{\sharp'}$
hold.
Then, by definition:
\[
\begin{array}{lll}
\Get'[v,S\setminus \{a\}] = \gamma_{v,S\setminus\{a \}}(V^{\sharp'},L^{\sharp'},\sigma^{\sharp'}) &=&
    \{ t \in \T_{S \setminus \{a\}} \mid \loc\,t=v, \beta\,t = (V,L,\sigma),\\
    && \quad  \sigma \subseteq \gamma_\D\circ\sigma^{\sharp'}, V \sqsubseteq V^{\sharp'}, L \sqsubseteq L^{\sharp'} \}
\end{array}
\]
For every trace $t \in \Get^{i-1}\,[u,S]$, let $\beta\,t = (V,L,\sigma)$. By induction hypothesis,
$V \sqsubseteq V^\sharp$, $L \sqsubseteq L^\sharp$, and
$\sigma \subseteq \gamma_\D\circ\sigma^\sharp$. Let $t'= \sem{e}\{ t\}$, then $\loc(t') = v$, $t' \in \T_{S \setminus \{a\}}$, and
\[
\begin{array}{lll}
  \beta\,t' &=& (V',L',\sigma') \text{ where: }\\[1ex]

  V' &=& \{ a' \mapsto \{ g \mid g\in\G, (\_,g=x, \bar u ') = \lTlW_g\,t', \bar u \leq \bar u'  \} \mid a' \in \M, \\
  && \qquad(\_,\lock(a'), \bar u) = \lTlL_{a'}\,t'\}\\
  && \cup \; \{ a' \mapsto \{ g \mid g\in\G, (\_,g=x, \_) = \lTlW_g\,t' \} \mid a' \in \M,\\
  && \qquad \bot = \lTlL_a\,t' \}\\
  &=& V \\[1ex]

  L' &=& \{ a' \mapsto \{ L_{t'}[\bar u] \} \mid a' \in \M, (\bar u,\lock(a'), \_) = \lTlL_{a'}\,t' \} \\
  && \cup \; \{ a' \mapsto \emptyset \mid a' \in \M, \bot = \lTlL_{a'}\,t'   \}\\
  &=& L\\[1ex]

  \sigma' &=&  \{ x \mapsto \{t'(x)\} \mid x \in\X \} \cup \{ g \mapsto \emptyset \mid g \in \G, \bot = \lTlW_{g}\,t'\} \\
  && \cup \; \{ g \mapsto \{\sigma_{j-1}\,x\}  \mid g \in \G, ((j-1,u_{j-1},\sigma_{j-1}),g=x, \_) = \lTlW_{g}\,t'\} \\
  &=& \sigma
\end{array}
\]
Thus, $V = V' \sqsubseteq V^{\sharp} \sqsubseteq V^{\sharp'}$, and $L = L' \sqsubseteq L^{\sharp} \sqsubseteq L^{\sharp'}$, and
$\sigma' \subseteq \gamma_\D\circ\sigma^{\sharp'}$.
Altogether, $t' \in \Get'\,[v,S\setminus\{a\}]$ for all $t \in \Get^{i-1}\,[u,S]$.
We conclude that the return value of
$\sem{[u,S],\unlock(a)}\CMU\,\Get^{i-1}$ is subsumed by the value $\Get'\,[v,S\setminus\{a\}]$.

Next, we consider the side-effects of the corresponding right-hand-side functions.
For each $g \in \G$, we distinguish two cases for $t'$:
\begin{itemize}
  \item $\lTlW _g\,t' = \bot$: Then side-effects $[g,a,S\setminus \{a\}] \mapsto \{ t'\}$
  are caused. They are accounted for by construction of $\Get'$:
  \[
    t' \in \{ t \in \T_{S\setminus\{a\}} \mid \last\,t=\unlock(a), \lTlW_g\,t = \bot \} \subseteq \Get'[g,a,S\setminus \{a\}]
  \]
  \item $\lTlW _g\,t' = ((j-1,u_{j-1},\sigma_{j-1}),g=x, \bar u ')$: Then the side-effects caused by $\C\CMU$ and $\C_\LockCentered$
  for $g$, respectively, are given by
  \[
  \begin{array}{lll}
    \rho' &=& [g,a,S\setminus \{a\}] \mapsto \{t' \}\\
    \rho^{\sharp'} &=& [g,a,S \setminus \{a\}] \mapsto \sigma^\sharp\,g\\
  \end{array}
  \]
  We remark that $\sigma\,g = \sigma_{j-1}\,x\subseteq
  (\gamma_\D\circ\sigma^{\sharp})\,g$ holds, and so does $\sigma^\sharp\,g \sqsubseteq \Get_\LockCentered\,[g,a,S \setminus \{ a\}]$
  because $\Get_\LockCentered$ is a post-solution of $C_\LockCentered$. Thus,
  \[
  t' \in   \Get'[g,a,S\setminus \{a\}] = \gamma_{g,a,S \setminus \{a\}}(\Get_\LockCentered\,[g,a,S \setminus \{ a\}])
  \]
\end{itemize}
Hence, all side-effects for $\unlock(a)$ of $\C\CMU$ are accounted for in $\Get'$,
and the claim holds.

\vspace{1em}
\noindent Next, for $x = \create(u_1)$:
\[
\begin{array}{lll}
  \sem{[u,S],x = \create(u_1)}\CMU\,\Get\CMU	&=&
		\Let\;T = \sem{e}(\Get\CMU\,[u,S])\;\In	\\
  &&		(\{[u_1,\emptyset]\mapsto\new\,u_1\,(\Get\CMU\,[u,S])\},T)\\[1ex]

  \sem{[u,S],x = \create(u_1)}^\sharp_\LockCentered\Get_\LockCentered	&=&	\Let\;(V^\sharp,L^\sharp,\sigma^\sharp) = \Get\,[u,S]\;\In	\\
  & & \Let\; V^{\sharp'''} = \{ a \mapsto \emptyset \mid a \in \M \}\;\In\\
  & & \Let\; L^{\sharp'''} = \{ a \mapsto \emptyset \mid a \in \M \} \;\In\\
  & & \Let\; i^\sharp = \nu^\sharp\,u\,(V^\sharp,L^\sharp,\sigma^\sharp)\,u_1\;\In\\
  & & \Let\; \sigma^{\sharp'''} = \sigma \oplus (\{ \self \mapsto i^\sharp \} \cup \{ g \mapsto \bot \mid g \in \G\})\;\In\\
  & & \Let\; \sigma^{\sharp''} = \sigma \oplus \{ x \mapsto i^\sharp \}\;\In\\
  & & \Let\; \rho^\sharp = \{ [u_1,\emptyset] \mapsto (V^{\sharp'''},L^{\sharp'''},\sigma^{\sharp'''}\}\;\In\\
  & &	(\rho^\sharp,(V^\sharp,L^\sharp,\sigma^{\sharp''}))
\end{array}
\]
Let $\Get_\LockCentered\,[u,S] = (V^\sharp,L^\sharp,\sigma^\sharp)$ and
$\Get_\LockCentered\,[v,S] = (V^{\sharp'},L^{\sharp'},\sigma^{\sharp'})$ the
value provided by $\Get_\LockCentered$ for the end point of the given control-flow edge and lockset.
Since $\Get_\LockCentered$ is a post-solution of $\C_\LockCentered$,
$V^{\sharp} \sqsubseteq V^{\sharp'}$, $L^{\sharp} \sqsubseteq L^{\sharp'}$, and $\sigma^{\sharp''} \sqsubseteq \sigma^{\sharp'}$
hold.
Then, by definition:
\[
\begin{array}{lll}
\Get'[v,S] = \gamma_{v,S}(V^{\sharp'},L^{\sharp'},\sigma^{\sharp'}) &=&
    \{ t \in \T_{S} \mid \loc\,t=v, \beta\,t = (V,L,\sigma),\\
    && \quad  \sigma \subseteq \gamma_\D\circ\sigma^{\sharp'}, V \sqsubseteq V^{\sharp'}, L \sqsubseteq L^{\sharp'} \}
\end{array}
\]
For every trace $t \in \Get^{i-1}\,[u,S]$, let $\beta\,t = (V,L,\sigma)$. By induction hypothesis,
$V \sqsubseteq V^\sharp$, $L \sqsubseteq L^\sharp$, and
$\sigma \subseteq \gamma_\D\circ\sigma^\sharp$. Let $t'= \sem{e}\{t\}$, then $\loc(t') = v$, $t' \in \T_S$, and
\[
\begin{array}{lll}
  \beta\,t' &=& (V',L',\sigma') \text{ where: }\\[1ex]

  V' &=& \{ a \mapsto \{ g \mid g\in\G, (\_,g=x, \bar u ') = \lTlW_g\,t', \bar u \leq \bar u'  \} \mid a \in \M, \\
  && \qquad(\_,\lock(a), \bar u) = \lTlL_{a}\,t'\}\\
  && \cup\; \{ a \mapsto \{ g \mid g\in\G, (\_,g=x,\_) = \lTlW_g\,t' \} \mid a \in \M,\\
  && \qquad \bot = \lTlL_a\,t' \}\\
  &=& V \\[1ex]

  L' &=& \{ a \mapsto \{ L_{t'}[\bar u] \} \mid a \in \M, (\bar u,\lock(a), \_) = \lTlL_{a}\,t' \} \\
  && \cup \; \{ a \mapsto \emptyset \mid a \in \M, \bot = \lTlL_{a}\,t'   \}\\
  &=& L\\[1ex]

  \sigma' &=&  \{ x \mapsto \{t'(x)\} \mid x \in\X \} \cup \{ g \mapsto \emptyset \mid g \in \G, \bot = \lTlW_{g}\,t'\} \\
  && \cup \; \{ g \mapsto \{\sigma_{j-1}\,x\}  \mid g \in \G, ((j-1,u_{j-1},\sigma_{j-1}),g=x, \_) = \lTlW_{g}\,t'\} \\
  &=& \sigma \oplus \{ x \mapsto \{\nu\,t\} \}
\end{array}
\]
Since by definition $\nu\,t \in\gamma_D (\nu^\sharp\,u\,(V^\sharp,L^\sharp,\sigma^\sharp)\,u_1) $, thus
\[
\begin{array}{lll}
  V' &=& V \sqsubseteq V^{\sharp} \sqsubseteq V^{\sharp'}\\
  L' &=& L \sqsubseteq L^{\sharp} \sqsubseteq L^{\sharp'}\\
  \sigma' &=& \sigma \oplus \{ x \mapsto \{\nu\,t\} \} \subseteq
    \gamma_\D\circ(\sigma^{\sharp} \oplus \{ x \mapsto \nu^\sharp\,u\,(V^\sharp,L^\sharp,\sigma^\sharp)\,u_1 \})
  = \gamma_\D\circ\sigma^{\sharp''} \subseteq \gamma_\D\circ\sigma^{\sharp'}
\end{array}
\]
Altogether, $t' \in \Get'\,[v,S]$ for all $t \in \Get^{i-1}\,[u,S]$.
We conclude that the return value of
$\sem{[u,S],x=\create(u_1)}\CMU\,\Get^{i-1}$ is subsumed by the value $\Get'\,[v,S]$.

Next, we consider the side-effects of the corresponding right-hand-side functions for $t \in \Get^{i-1}\,[u,S]$:
\[
  \begin{array}{lll}
    \rho' &=& [u_1,\emptyset] \mapsto\new\,u_1\,\{ t \} \\
    \rho^{\sharp'} &=& [u_1,\emptyset] \mapsto (V^{\sharp'''},L^{\sharp'''},\sigma^{\sharp'''})\\
  \end{array}
\]
Let $t'' = \new\,u_1\,\{ t \}$. Then,
\[
\begin{array}{lll}
  \beta\,t'' &=& (V''',L''',\sigma''') \text{ where: }\\[1ex]

  V''' &=& \{ a \mapsto \{ g \mid g\in\G, (\_,g=x, \bar u ') = \lTlW_g\,t'', \bar u \leq \bar u'  \} \mid a \in \M, \\
  && \qquad(\_,\lock(a), \bar u) = \lTlL_{a}\,t''\}\\
  && \cup\; \{ a \mapsto \{ g \mid g\in\G, (\_,g=x, \_) = \lTlW_g\,t'' \} \mid a \in \M,\\
  && \qquad \bot = \lTlL_a\,t'' \}\\
  &=& \{ a \mapsto \emptyset \mid a \in \M \}\\[1ex]

  L''' &=& \{ a \mapsto \{ L_{t''}[\bar u] \} \mid a \in \M, (\bar u,\lock(a), \_) = \lTlL_{a}\,t'' \} \\
  && \cup \; \{ a \mapsto \emptyset \mid a \in \M, \bot = \lTlL_{a}\,t''   \}\\
  &=& \{ a \mapsto \emptyset \mid a \in \M \}\\[1ex]

  \sigma''' &=&  \{ x \mapsto \{t''(x)\} \mid x \in\X \} \cup \{ g \mapsto \emptyset \mid g \in \G, \bot = \lTlW_{g}\,t''\} \\
  && \cup \; \{ g \mapsto \{\sigma_{j-1}\,x\}  \mid g \in \G, ((j-1,u_{j-1},\sigma_{j-1}),g=x, \_) = \lTlW_{g}\,t''\} \\
  &=& \sigma \oplus (\{ \self \mapsto \{\nu\,t\} \} \cup \{ g \mapsto \emptyset \mid g \in \G \} )
\end{array}
\]
Since by definition $\nu\,t \in\gamma_D (\nu^\sharp\,u\,(V^\sharp,L^\sharp,\sigma^\sharp)\,u_1) $, thus
\[
\begin{array}{lll}
  V''' &=& V^{\sharp'''}\\
  L''' &=& L^{\sharp'''}\\
  \sigma''' &=& \sigma \oplus (\{ \self \mapsto \{\nu\,t\} \} \cup \{ g \mapsto \emptyset \} ) \\
  &\subseteq& \gamma_\D\circ(\sigma^{\sharp} \oplus
    (\{ \self \mapsto \nu^\sharp\,u\,(V^\sharp,L^\sharp,\sigma^\sharp)\,u_1 \}  \cup \{ g \mapsto \bot \}))\\
  &=& \gamma_\D\circ\sigma^{\sharp'''}
\end{array}
\]
We remark that $(V^{\sharp'''},L^{\sharp'''},\sigma^{\sharp'''}) \sqsubseteq  \Get_\LockCentered\,[u_1,\emptyset]$ holds as
$\Get_\LockCentered$ is a post-solution of $C_\LockCentered$.
Thus,
\[
t'' \in \gamma_{u_1,\emptyset}(\Get_\LockCentered\,[u_1,\emptyset]) = \Get'[u_1,\emptyset]
\]
Hence, all side-effects for $x = \create(u_1)$ of $\C\CMU$ are accounted for in $\Get'$.
This concludes the proof.
\qed
\end{proof}

%% file: correctness/write-centered.tex
\subsection{Write-Centered Reading}\label{s:soundness-write-centered}
Let the constraint system for the Write-Centered Reading analysis from \cref{s:write-centered} be called $\C_\WriteCentered$.
We construct from the constraint system $\C$ for the concrete collecting semantics a system $\C\CWCU$
so that the set of unknowns of $\C\CWCU$ matches the set of unknowns of $\C_\WriteCentered$.
This means that
each unknown $[u]$ for program point $u$ is replaced with the set of unknowns $[u,S]$, $S\subseteq\M$,
while the unknown $[a]$ for a mutex $a$ is replaced with the set of unknowns
$[g,a,S,w]$, $g\in\G,S\subseteq\M, w\subseteq\M$.
Accordingly, the constraint system $\C\CWCU$ consists of these constraints:
\[
  \begin{array}{llll}
    \relax [u_0,\emptyset] & \supseteq &\textbf{fun}\,\_\to (\emptyset,\init) \\
    \relax [u',S\cup\{a\}] & \supseteq & \sem{[u,S],\lock(a)}\CWCU & \quad (u,\lock(a),u')\in\E, a\in\M \\
    \relax [u',S\setminus\{a\}] & \supseteq & \sem{[u,S],\unlock(a)}\CWCU & \quad (u,\unlock(a),u') \in \E, a \in \M\\
    \relax [u',S] & \supseteq & \sem{[u,S],A}\CWCU & \quad (u,A,u') \in \E, \forall a \in\M:\\
      &&& \qquad A\neq \lock(a), A \neq \unlock(a)  \\
  \end{array}
\]
where new right-hand-side functions (relative to the semantics $\sem{e}$ of control-flow edges $e$)
are given by:
\[
\begin{array}{lll}
\sem{[u,S],x = \textsf{create}(u_1)}\CWCU\,\Get\CWCU	&=&
		\Let\;T = \sem{e}(\Get\CWCU\,[u,S])\;\In	\\
&&		(\{[u_1,\emptyset]\mapsto\new\,u_1\,(\Get\CWCU\,[u,S])\},T)	\\[1ex]

\sem{[u,S],\lock(a)}\CWCU\,\Get\CWCU &=& \Let\;T' = \bigcup \{ \Get\CWCU\,[g,a,S',w] \mid g\in\G, S' \subseteq \M, w \subseteq \M \}\;\In \\
  & & (\emptyset,\sem{e}(\Get\CWCU\,[u,S],T'))\\[1ex]

\sem{[u,S],\unlock(a)}\CWCU\,\Get\CWCU
    &=&	\Let\;T = \sem{e}(\Get\CWCU\,[u,S])\;\In	\\
    & & \Let\;\rho = \{[g,a,S\setminus \{ a \}, w ]\mapsto \{t\} \mid t \in T, g \in \G, w \subseteq \M, \\
    & & \qquad ((\lTlW _g\,t = (\bar u, g =x, \bar u ') \land L_t[\bar u'] \subseteq w)\\
    & & \qquad \lor (\lTlW _g\,t = \bot)) \}\;\In \\
    & &	(\rho, T)\\[1ex]

\sem{[u,S],x=g}\CWCU\,\Get\CWCU	&=& (\emptyset,\sem{e}(\Get\CWCU\,[u,S]))	\\[1ex]

\sem{[u,S],g=x}\CWCU\,\Get\CWCU	&=& (\emptyset,\sem{e}(\Get\CWCU\,[u,S]))	\\
\end{array}
\]
In contrast to the right-hand-side functions of $\C$, the new right-hand sides now also re-direct
side-effects not to unknowns $[a], a\in\M$, but to appropriate more specific unknowns
$[g,a,S',w], g\in\G, a\in\M,S'\subseteq\M, w\subseteq\M$.
For a mapping $\Get$ from the unknowns of $\C$ to $2^\T$, we construct a mapping $\GetPrimeP$
from the unknowns of $\C\CWCU$ to $2^\T$ by
\[
\begin{array}{llll}
  \GetPrimeP [u,S] &=& \Get[u] \cap \T_S & \text{ for } u \in \N, S \subseteq \M \\
  \GetPrimeP [g,a,S,w] &=& \Get[a] \cap \{ t \in \T_{S} \mid & \text { for } g \in \G, a \in \M, S \subseteq \M,w \subseteq \M  	 	\\
    && \quad (\lTlW_g\,t = (\bar u, g=&x,\bar u')\land\,L_t[\bar u'] \subseteq w) 	\\
    && \quad \lor\,(\lTlW _g\,t = \bot) \}
\end{array}
\]
Thus,
\[
\begin{array}{lll}
\Get[u]	&=& \bigcup\{\GetPrimeP[u,S]\mid S\subseteq\M\}	\\
\Get[a] &=& \bigcup\{ \GetPrimeP[g,a,S,w]\mid g\in\G, S\subseteq\M,w\subseteq\M \}
\end{array}
\]
for all program points $u$ and mutexes $a$.
Moreover, we have:
\begin{proposition}\label{p:GetGet'}
The following two statements are equivalent:
\begin{itemize}
  \item $\Get$ is the least solution of $\C$;
  \item $\GetPrimeP$ is the least solution of $\C\CWCU$.
\end{itemize}
\end{proposition}
\begin{proof}
The proof of \cref{p:GetGet'} is by fixpoint induction.\qed
\end{proof}
The next proposition indicates that the new unknown $[g,a,S,w]$ collects a superset of local
traces whose last write to the global $g$ can be read by a thread satisfying the specific assumptions
(W0) through (W4) below.

\begin{proposition}\label{prop:read}
Consider the $i$-th approximation $\Get^i$ to the least solution $\GetPrimeP$ of constraint system $\C\CWCU$,
a control-flow edge $(u,x=g,u')$ of the program, and a local trace $t \in\Get^i\,[u',S]$
in which the last action is $x=g$,
that ends in $\bar u' = (j,u',\sigma)$, i.e., $t = (\bar u')\downarrow_t$.
Let $P=\minLSince(t,\bar v')$ denote the upwards-closed set of minimal locksets held by
the ego thread since the endpoint $\bar v'$ of its last thread-local write to $g$,
or $\{\emptyset\}$ if there is no thread-local write to $g$ in $t$.

Then, the value $d = \sigma\,x$ that is read for $g$, is produced by a write to $g$
which
\begin{itemize}
  \item either is the last thread-local write to $g$ in $t$; or
  \item is the last \emph{thread-local} write to $g$ in some local trace stored at $\Get^{i'}\,[g,a,S',w']$ for some $i' < i$
	 i.e.,
	\[
	d\in\eval_g (\Get^{i'}\,[g,a,S',w'])
	\]
  where
	\begin{enumerate}
	\item[(W0)]	$a\in S$,
	\item[(W1)]	$w' \subseteq \M$,
	\item[(W2)]	$S \cap S' = \emptyset$,
  \item[(W3)]	$\exists S'' \in P:\,S'' \cap w' = \emptyset$, and
	\item[(W4)]	$\exists S''' \in P:\,a \notin S'''$
	\end{enumerate}
\end{itemize}
\end{proposition}
\begin{proof}
The proof is by fixpoint induction:
We prove that the values read non-thread-locally for a global $g$ at some $(u,x=g,u')$ during
the computation of $\Get^i$ are the last \emph{thread-local} writes of a local trace
$t'$ ending in an unlock operation that is side-effected to an appropriate
$\Get^{i'}\,[g,a,S',w']$ in some prior iteration $i' < i$ for some $a$, $S'$, and $w'$ satisfying
(W0) and (W4).

This property holds for $i=0$, as in $\Get^0$, all unknowns for program points and currently held locksets
(except for the initial program point and the empty lockset) are $\emptyset$, and therefore no reads from globals or unlocks can happen.

For the induction step $i > 0$, there are two proof obligations: First that the property holds
for all reads from a global, and additionally that all traces ending in an unlock operation are
once more side-effected to appropriate unknowns in this iteration.

For the first obligation, consider a local trace $t \in\Get^i\,[u',S]$ where the last action is
$x=g$.
There is a last write to $g$ in $t$:
\[
\lW_g\,t = ((j'-1,u_{j'-1},\sigma_{j'-1}), g = x', \bar u'') = l.
\]
Let $i_0 = \id\,t$ and $i_1 = \sigma_{j-1}\,\self$ the thread \emph{id}s of the reading ego thread
and the thread performing the last write, respectively.
We distinguish two cases:\\
\emph{Case 1: $i_0 = i_1$.} The last write is thread-local to $t$ (and $l$ is therefore also the last thread-local write to $g$ in $t$).\\
\emph{Case 2: $i_0 \neq i_1$.} The last write is not thread-local.
Let $w_l$ the set of locks held on that last write to
$g$, i.e., $L_t[\bar u'']$.
Consider the maximal sub-trace $t'$ of $t$ with $\id(t') = i_1$ so that
$\last(t')$ unlocks some mutex in $S$. Let this mutex be $a$.
Such a sub-trace must exist since accessing $g$ is necessarily protected by $m_g$.
Let $S'$ denote the background lockset held at this last action in $t'$.
$t'$ was produced during some earlier iteration $i' < i$.
By induction hypothesis, we may assume this $t'$ was side-effected to $\Get^{i'}[g,a,S',w'']$
during the $i'$-th iteration, for all $w'' \supseteq w_l$.
Therefore, the read value $d$ is given by
\[
d = \sigma_{j-1}\,x' \in \eval_g (\Get^{i'}\,[g,a,S',w']) \subseteq  \eval_g (\Get^{i}\,[g,a,S',w'])
\]
It remains to prove that then the conditions hold for
$a$, $S'$, $w'$:
\begin{itemize}
\item[(W0)]	$a \in S$ (at least $m_g$, perhaps more)
\item[(W1)]	$w' \subseteq w_l \subseteq \M$ (by construction of the constraint for
	edges with unlock operations in $\C\CWCU$)
\item[(W2)]	$S \cap S' = \emptyset$\\
    Assume that this were not the case, i.e., $c\in S\cap S'$.
    Then thread $i_1$ holds the lock of mutex $c$ at the sink of
    every super-trace $t''$ of $t'$ in $t$ with $\id(t'') = i_1$.
    Since $c$ is never released, thread $i_0$ is unable to acquire $c$ --- which would
    be necessary to hold $S$ at the sink of $t$. Contradiction.
\item[(W3)]	$\exists S'' \in P\,g$, $S'' \cap w' = \emptyset$\\
    Let $l' = \lTlW_g t$ the last thread-local write to $g$ in $t$.
    If there is no last thread-local write, i.e., $l' = \bot$, then $P\,g = \{\emptyset\}$,
    and the condition holds. Otherwise, assume for a contradiction that $i_0$ has always maintained a
    non-empty lockset intersection with $w'$ since $l'$, i.e., since action $l'$ thread $i_0$ has at each point
    held one of the locks held when the write $l$ was performed. Then $l$ can not have
    happened after $l'$, and $l$ can not be the last write to $g$ in $t$.
\item[(W4)]	$\exists S''' \in P\,g$, $a \notin S'''$\\
    Let $l' = \lTlW_g t$ the last thread-local write to $g$ in $t$.
    If there is no last thread-local write, i.e., $l' = \bot$, then $P\,g = \{\emptyset\}$,
    and the condition holds.
    Otherwise, assume for a contradiction that $a \in S'''$ for all $S''' \in P\,g$.
    Since $a$ is unlocked by $i_1$ after $l$, $l$ can not have happened after $l'$ and $l$ can not
    be the last write to $g$ in $t$.
\end{itemize}
It now remains to show that any trace $t$ with $\last(t)=\unlock(a)$, $a\in\M$ ending in $(j,u,\sigma)$, i.e., $t=(j,u,\sigma)\downarrow_t$,
produced in this iteration $i$ is side-effected to $\Get^i\,[g,a,S,w]$
for $S = L_t[(j,u,\sigma)]$ and $w_l' \subseteq w\subseteq\M$, where $w_l'$ is the set of mutexes
held when writing to $g$ for the last time thread-locally in $t$,
if $g$ was written to at all. This, however, follows directly from the construction of
$\C\CWCU$.
\qed
\end{proof}

\noindent
Our goal is to relate post-solutions of the constraint systems
$\C\CWCU$ and $\C_\WriteCentered$ to each other. While the sets of unknowns of these two systems are the same,
the side-effects to unknowns are still not fully comparable.
Therefore, we modify the side-effects produced by
$\C_\WriteCentered$ for unlock operations to obtain yet another constraint system $\C_\ModWC$.
All right-hand-side functions remain the same except for $\unlock(a)$
which is now given by:
\[
	\begin{array}{lll}
	\sem{[u,S],\unlock(a)}^\sharp_\ModWCL\Get_\ModWC
		&=&	\Let\;(W,P,\sigma) = \Get\,[u,S]\;\In	\\
		& & \Let\;P' = \{ g \mapsto P\,g \sqcup \{S \setminus \{a\}\} \mid g \in \G \}\;\In\\
		& & \Let\;\rho = \{ [g,a,S \setminus \{a\},w] \mapsto \sigma\,g \mid\\
    && \qquad g \in \G, w'\in W\,g, w' \subseteq w \}\;\In\\
		& &	(\rho, (W,P',\sigma))
	\end{array}
\]
Instead of only side-effecting to \emph{minimal} sets $w'$ of locks held on a write to $g$,
the value now is side-effected to \emph{all}
supersets $w$ of such minimal elements.
This modification of the constraint system
only changes the values computed for globals, but not those for program points and currently held locksets: Upon reading, all
$[g,a,S,w]$ are consulted where there is an empty intersection of $w$ and some $P\,g$.
If this is the case for $w$,
it also holds for $w' \subseteq w$. Accordingly, the values additionally published to $[g,a,S,w]$,
are already read from $[g,a,S,w']$ directly in $C_\WriteCentered$.
More formally,
let $\Get_\WriteCentered$ be a post-solution of $\C_\WriteCentered$, define $\Get_\ModWC$ by
\[
\begin{array}{llll}
  \Get_\ModWC\,[u,S] &=& \Get_\WriteCentered\,[u,S] \qquad & u\in\N, S\subseteq\M\\
  \Get_\ModWC\,[g,a,S,w] &=& \bigsqcup \{\Get_\WriteCentered\,[g,a,S,w'] \mid w' \subseteq w\}
    \qquad & g\in\G, a\in\M, S\subseteq\M,w\subseteq\M
\end{array}
\]
Then, we have:
\begin{proposition}\label{prop:modBTPG}
$\Get_\ModWC$ as constructed above is a post-solution of $\C_\ModWC$.
\end{proposition}
\begin{proof}
The proof of \cref{prop:modBTPG} is by verifying for each edge $(u,A,v)$
of the control-flow graph, each possible lockset $S$,
and $\Get_\ModWC$ as constructed above, that
\[
\sem{[u,S],A}^\sharp_\ModWCL\,\Get_\ModWC \sqsubseteq(\Get_\ModWC,\Get_\ModWC\,[v,S'])
\]
holds.\qed
\end{proof}

\noindent It thus remains to relate post-solution of $\C\CWCU$ and $C_\ModWC$ to each other.
As a first step, we define a function
$\beta$ that extracts from a local trace $t$ for each global $g$ the minimal lockset $W\,g$
held at the last \emph{thread-local} write to $g$, as well as all minimal locksets $P\,g$
since the last \emph{thread-local} write to $g$.
Additionally, it extracts a map $\sigma$ that contains
the values of the locals at the sink of $t$ as well as the last-written thread-local values
of globals. Thus, we define
\[
  \begin{array}{lll}
  \beta\,t &=& (W,P,\sigma) \qquad \text{where} \\[1ex]
  W &=& \{ g \mapsto \{ L_t[\bar u'] \}   \mid g \in \G, (\_,g=x, \bar u ') = \lTlW_g\,t\}\\
  && \cup \; \{ g \mapsto \emptyset \mid g \in \G, \bot  = \lTlW_g\,t\ \}\\[1ex]
  P &=& \{ g \mapsto \minLSince\,t\,\bar u' \mid g \in \G, (\_,g=x, \bar u') = \lTlW_g\,t\}\\
  && \cup \; \{ g \mapsto \{\emptyset\} \mid g \in \G, \bot  = \lTlW_g\,t\ \}\\[1ex]
  \sigma &=& \{ x \mapsto \{t(x)\} \mid x \in \X \} \cup \{ g \mapsto \emptyset \mid g \in \G, \bot = \lTlW_g\,t\} \\[1ex]
  && \cup \; \{ g \mapsto \{\sigma_{j-1}\,x\}  \mid g \in \G, ((j-1,u_{j-1},\sigma_{j-1}),g=x, \_) = \lTlW_g\,t\} \\
  \end{array}
\]
\noindent
The abstraction function $\beta$ is used to specify concretization functions for
the values of unknowns $[u,S]$ for program points and currently held locksets as well as for unknowns $[g,a,S,w]$.
%
\[
\begin{array}{lll}
  \gamma_{u,S}(P^\sharp,W^\sharp,\sigma^\sharp) &=& \{ t \in \T_S \mid \loc\,t=u, \beta\,t = (W,P,\sigma),\\
  && \quad  \sigma \subseteq \gamma_\D\circ\sigma^\sharp, W \sqsubseteq W^\sharp, P \sqsubseteq P^\sharp\}
\end{array}
\]
where $\subseteq,\sqsubseteq$ are extended point-wise from domains to maps into domains.
Moreover,
%
\[
\begin{array}{lll}
  \gamma_{g,a,S,w}(v) &=& \{ t \in \T_S \mid \last\,t=\unlock(a),\\
  && (\_,\_,\sigma_{j-1}),g=x, \bar u ') = \lTlW_g\,t, \sigma_{j-1}\,x \in \gamma_\D(v),\\
  && w \subseteq  L_t[\bar u'] \}\\
  && \cup\; \{ t \in \T_S \mid \last\,t=\unlock(a), \lTlW_g\,t = \bot \}
\end{array}
\]
where $\gamma_\D:\D\to 2^\V$ is the concretization function for abstract values in $\D$.
Let $\Get_\ModWC$ be a post-solution of $\C_\ModWC$. We then construct from it a mapping $\Get\CWCU$ by:
\[
  \begin{array}{llll}
    \Get\CWCU [u,S] &=& \gamma_{u,S}(\Get_\ModWC\,[u,S]) \qquad & u\in\N, S\subseteq\M\\
    \Get\CWCU [g,a,S,w] &=& \gamma_{g,a,S,w}(\Get_\ModWC\,[g,a,S,w]) \qquad & g\in\G, a\in\M, S\subseteq\M,w\subseteq\M
  \end{array}
\]
Altogether, the correctness of the constraint system $\C_\WriteCentered$ follows from the
following theorem.

\begin{theorem}\label{t:write-centered}
Every post-solution of $\C_\WriteCentered$ is sound w.r.t.\ the local trace semantics.
\end{theorem}

\begin{proof}
Recall from \cref{p:GetGet'}, that the least solution of $\C\CWCU$ is sound
w.r.t.\ the local trace semantics as specified by the constraint system $\C$.
By \cref{prop:modBTPG}, it thus suffices to prove that the mapping $\Get\CWCU$
as constructed above, is a post-solution of the constraint system $\C\CWCU$.
For that, we verify by fixpoint induction that for the $i$-th approximation $\Get^i$
to the least solution $\GetPrimeP$ of $\C\CWCU$, $\Get^i \subseteq \Get\CWCU$ holds.
To this end, we verify for the start point $u_0$ and the empty lockset, that
\[
  (\emptyset,\init) \subseteq (\Get\CWCU,\Get\CWCU\,[u_0,\emptyset])
\] holds and for each edge $(u,A,v)$ of the control-flow graph
and each possible lockset $S$, that
\[
  \sem{[u,S],A}\CWCU\, \Get^{i-1} \subseteq(\Get\CWCU,\Get\CWCU\,[v,S'])
\] holds.

\noindent First, for the start point $u_0$ and the empty lockset:
\[
  (\emptyset,\init) \subseteq (\Get\CWCU,\Get\CWCU\,[u_0,\emptyset])
\]
As there are no side-effects triggered, it suffices to check that $\init \subseteq \Get\CWCU\,[u_0,\emptyset]$.
\[
	\begin{array}{lll}
		\init^\sharp_\ModWCL\,\_ &=& \Let\; W^\sharp = \{ g \mapsto \emptyset \mid g \in \G \} \;\In\\
		&& \Let\; P^\sharp = \{ g \mapsto \{\emptyset\} \mid g \in \G \} \;\In\\
		&& \Let\; \sigma^\sharp = \{ x \mapsto \top \mid x \in \X \} \cup \{ g \mapsto \bot \mid g \in \G \} \;\In\\
		&& (\emptyset,(W^\sharp,P^\sharp,\sigma^\sharp))
	\end{array}
\]
Let $\Get_\ModWC\,[u_0,\emptyset] = (W^{\sharp'},P^{\sharp'},\sigma^{\sharp'})$ the
value provided by $\Get_\ModWC$ for the start point and the empty lockset.
Since $\Get_\ModWC$ is a post-solution of $\C_\ModWC$,
$W^\sharp \sqsubseteq W^{\sharp'}$, $P^\sharp \sqsubseteq P^{\sharp'}$,
and $\sigma^\sharp \sqsubseteq \sigma^{\sharp'}$ all hold.
Then, by definition:
\[
\begin{array}{lll}
\Get\CWCU[u_0,\emptyset] = \gamma_{u_0,\emptyset}(W^{\sharp'},P^{\sharp'},\sigma^{\sharp'}) &=&
    \{ t \in \T_\emptyset \mid \loc\,t=u_0, \beta\,t = (W,P,\sigma),\\
    && \quad  \sigma \subseteq \gamma_\D\circ\sigma^{\sharp'}, W \sqsubseteq W^{\sharp'}, P \sqsubseteq P^{\sharp'} \}
\end{array}
\]
For every trace $t \in \init$, let
\[
\begin{array}{lll}
  \beta\,t &=& (W,P,\sigma) \text{ where }\\[1ex]

  W &=& \{ g \mapsto \{ L_{t}[\bar u'] \} \mid g \in \G, (\bar u,g=x, \bar u ') = \lTlW_{g}\,t\}\\
  && \quad \cup \; \{ g \mapsto \emptyset \mid g \in \G, \bot  = \lTlW_{g}\,t\ \}\\
  &=& \{ g \mapsto \emptyset \mid g \in \G \}\\[1ex]

  P &=&  \{ g \mapsto \minLSince\,t\,\bar u' \mid g \in \G, (\bar u,g=x, \bar u ') = \lTlW_{g}\,t\}\\
  && \quad \cup \; \{ g \mapsto \{\emptyset\} \mid g \in \G, \bot  = \lTlW_{g}\,t\ \}\\
  &=& \{g \mapsto \{\emptyset\} \mid g \in \G\}\\[1ex]

  \sigma &=&  \{ x \mapsto \{t(x)\} \mid x \in \X \} \cup \{ g \mapsto \emptyset \mid g \in \G, \bot = \lTlW_{g}\,t\} \\
  && \quad \cup \; \{ g \mapsto \{\sigma_{j-1}\,x\}  \mid g' \in \G, ((j-1,u_{j-1},\sigma_{j-1}),g=x, \bar u ') = \lTlW_{g}\,t\} \\
  &=& \{ x \mapsto \{t(x)\} \mid x \in \X \} \cup \{ g \mapsto \emptyset \mid g \in \G\}
\end{array}
\]
Thus,
\[
  \begin{array}{lll}
    W &=& W^{\sharp} \sqsubseteq W^{\sharp'}\\
    P &=& P^{\sharp} \sqsubseteq P^{\sharp'}\\
    \sigma &=& \{ x \mapsto \{t(x)\} \mid x \in \X \} \cup \{ g \mapsto \emptyset \mid g \in \G\} \\
    &\subseteq& \gamma_\D\circ(\{ x \mapsto \top \mid x \in \X \} \cup \{ g \mapsto \bot \mid g \in \G\}) = \gamma_\D\circ\sigma^\sharp \\
    &\subseteq& \gamma_\D\circ\sigma^{\sharp'}
  \end{array}
\]
Altogether, $t \in \Get\CWCU\,[u_0,\emptyset]$ for all $t \in \init$.

\vspace{1em}
\noindent Next, we verify for each edge $(u,A,v)$ of the control-flow graph
and each possible lockset $S$, that
\[
  \sem{[u,S],A}\CWCU\, \Get^{i-1} \subseteq(\Get\CWCU,\Get\CWCU\,[v,S'])
\] holds.

\vspace{1em}
\noindent We first consider a write to a global $g=x$.
\[
\begin{array}{lll}
  \sem{[u,S],g=x}\CWCU\,\Get\CWCU	&=& (\emptyset,\sem{e}(\Get\CWCU\,[u,S]))	\\[1ex]

    \sem{[u,S],g=x}^\sharp_\ModWCL\Get_\ModWC	&=&	\Let\;(W^\sharp,P^\sharp,\sigma^\sharp) = \Get_\ModWC\,[u,S]\;\In	\\
		& &	\Let\;W^{\sharp''} = W^\sharp \oplus \{g \mapsto \{S\} \}\;\In\\
		& & \Let\;P^{\sharp''} = P^\sharp \oplus \{g \mapsto \{S\} \}\;\In\\
    & & \Let\;\sigma^{\sharp''} =	\sigma^\sharp \oplus\{g\mapsto(\sigma^\sharp \,x)\} \;\In\\
		& &	(\emptyset, (W^{\sharp''},P^{\sharp''},\sigma^{\sharp''}))
\end{array}
\]
Let $\Get_\ModWC\,[u,S] = (W^\sharp,P^\sharp,\sigma^\sharp)$ and
$\Get_\ModWC\,[v,S] = (W^{\sharp'},P^{\sharp'},\sigma^{\sharp'})$ the
value provided by $\Get_\ModWC$ for the end point of the given control-flow edge and lockset.
Since $\Get_\ModWC$ is a post-solution of $\C_\ModWC$,
$W^{\sharp''} \sqsubseteq W^{\sharp'}$, $P^{\sharp''} \sqsubseteq P^{\sharp'}$, and $\sigma^{\sharp''} \sqsubseteq \sigma^{\sharp'}$
hold.
Then, by definition:
\[
\begin{array}{lll}
\Get\CWCU[v,S] = \gamma_{v,S}(W^{\sharp'},P^{\sharp'},\sigma^{\sharp'}) &=&
    \{ t \in \T_S \mid \loc\,t=v, \beta\,t = (W,P,\sigma),\\
    && \quad  \sigma \subseteq \gamma_\D\circ\sigma^{\sharp'}, W \sqsubseteq W^{\sharp'}, P \sqsubseteq P^{\sharp'} \}
\end{array}
\]
For every trace $t \in \Get^{i-1}\,[u,S]$, let $\beta\,t = (W,P,\sigma)$. By induction hypothesis,
$W \sqsubseteq W^\sharp$, $P \sqsubseteq P^\sharp$, and
$\sigma \subseteq\gamma_\D\circ\sigma^\sharp$. Let $t'= \sem{e}\{ t\}$, then $\loc(t') = v$, $t' \in \T_S$, and
\[
\begin{array}{lll}
  \beta\,t' &=& (W',P',\sigma') \text{ where }\\[1ex]

  W' &=& \{ g' \mapsto \{ L_{t'}[\bar u'] \} \mid g' \in \G, (\_,g'=x', \bar u ') = \lTlW_{g'}\,t'\}\\
  && \cup \; \{ g' \mapsto \emptyset \mid g' \in \G, \bot  = \lTlW_{g'}\,t'\ \}\\
  &=& W \oplus \{ g \mapsto \{ L_t'[\bar u'] \} \mid (\_,g=x, \bar u ') = \lTlW_g\,t'\}\\
  &=& W \oplus \{ g \mapsto \{ S \} \}\\[1ex]

  P' &=&  \{ g' \mapsto \minLSince\,t'\,\bar u' \mid g' \in \G, (\_,g'=x', \bar u ') = \lTlW_{g'}\,t'\}\\
  && \cup \; \{ g' \mapsto \{\emptyset\} \mid g' \in \G, \bot  = \lTlW_{g'}\,t'\ \}\\
  &=& P \oplus \{ g \mapsto \minLSince\,t'\,\bar u' \mid (\_,g=x, \bar u ') = \lTlW_g\,t'\}\\
  &=& P \oplus \{ g \mapsto \{ S \} \}\\[1ex]

  \sigma' &=&  \{ x \mapsto \{t'(x')\} \mid x' \in \X \} \cup \{ g' \mapsto \emptyset \mid g' \in \G, \bot = \lTlW_{g'}\,t'\} \\
  && \cup \; \{ g' \mapsto \{\sigma_{j-1}\,x'\}  \mid g' \in \G, ((j-1,u_{j-1},\sigma_{j-1}),g'=x', \_) = \lTlW_{g'}\,t'\} \\
  &=& \sigma \oplus \{ g \mapsto \{\sigma_{j-1}\,x\}  \mid ((j-1,u_{j-1},\sigma_{j-1}),g=x, \_) = \lTlW_{g}\,t'\} \\
  &=& \sigma \oplus \{ g \mapsto \{t(x)\}\} = \sigma \oplus \{ g \mapsto \sigma\,x \}
\end{array}
\]
Thus,
\[
  \begin{array}{lll}
  W' &=& W \oplus \{g \mapsto \{S\}\} \sqsubseteq W^\sharp \oplus \{g \mapsto \{S\}\} = W^{\sharp''}  \sqsubseteq W^{\sharp'}\\
  P' &=& P \oplus \{g \mapsto \{S\}\} \sqsubseteq P^\sharp \oplus \{g \mapsto \{S\}\} = P^{\sharp''}  \sqsubseteq P^{\sharp'}\\
  \sigma' &=& \sigma \oplus \{ g \mapsto \sigma\,x \} \subseteq \gamma_\D\circ(\sigma^\sharp \oplus \{ g \mapsto \sigma^\sharp\,x\}) =
  \gamma_\D\circ\sigma^{\sharp''}\subseteq \gamma_\D\circ\sigma^{\sharp'}
  \end{array}
\]
Altogether, $t' \in \Get\CWCU\,[v,S]$ for all $t \in \Get^{i-1}\,[u,S]$.
We conclude that the return value of
$\sem{[u,S],g=x}\CWCU\Get^{i-1}$ is subsumed by the value $\Get\CWCU\,[v,S]$
and since the constraint causes no side-effects, the claim holds.

\vspace{1em}
\noindent Next, for a read from a global $x = g$:
\[
\begin{array}{lll}
  \sem{[u,S],x=g}\CWCU\,\Get\CWCU	&=& (\emptyset,\sem{e}(\Get\CWCU\,[u,S]))	\\[1ex]

  \sem{[u,S],x=g}^\sharp_\ModWCL\Get_\ModWC	&=&
  \Let\;(W^\sharp,P^\sharp,\sigma^\sharp) = \Get_\ModWC\,[u,S]\;\In	\\
  & & \Let\;d = \sigma^\sharp\, g \sqcup\bigsqcup\{\Get_\ModWC\,[g,a,S',w] \mid  a\in S, S\cap S' =\emptyset, \\
  & & \quad \exists S'' \in P^\sharp\,g: S'' \cap w = \emptyset, \\
  & & \quad \exists S''' \in P^\sharp\,g: a \notin S''' \}\;\In\\
  & & \Let\;\sigma^{\sharp''} = \sigma^\sharp \oplus\{x\mapsto d\}\;\In\\
  & &	(\emptyset, (W^{\sharp},P^{\sharp},\sigma^{\sharp''}))
\end{array}
\]
Let $\Get_\ModWC\,[u,S] = (W^\sharp,P^\sharp,\sigma^\sharp)$ and
$\Get_\ModWC\,[v,S] = (W^{\sharp'},P^{\sharp'},\sigma^{\sharp'})$ the
value provided by $\Get_\ModWC$ for the end point of the given control-flow edge and lockset.
Since $\Get_\ModWC$ is a post-solution of $\C_\ModWC$,
$W^{\sharp} \sqsubseteq W^{\sharp'}$, $P^{\sharp} \sqsubseteq P^{\sharp'}$, and $\sigma^{\sharp''} \sqsubseteq \sigma^{\sharp'}$
hold.
Then, by definition:
\[
\begin{array}{lll}
\Get\CWCU[v,S] = \gamma_{v,S}(W^{\sharp'},P^{\sharp'},\sigma^{\sharp'}) &=&
    \{ t \in \T_S \mid \loc\,t=v, \beta\,t = (W,P,\sigma),\\
    && \quad  \sigma \subseteq \gamma_\D\circ\sigma^{\sharp'}, W \sqsubseteq W^{\sharp'}, P \sqsubseteq P^{\sharp'} \}
\end{array}
\]
For every trace $t \in \Get^{i-1}\,[u,S]$, let $\beta\,t = (W,P,\sigma)$. By induction hypothesis,
$W \sqsubseteq W^\sharp$, $P \sqsubseteq P^\sharp$, and
$\sigma \subseteq \gamma_\D\circ\sigma^\sharp$. Let $t'= \sem{e}\{ t\}$, then $\loc(t') = v$, $t' \in \T_S$, and
\[
\begin{array}{lll}
  \beta\,t' &=& (W',P',\sigma') \text{ where: }\\[1ex]

  W' &=& \{ g' \mapsto \{ L_{t'}[\bar u'] \} \mid g' \in \G, (\_,g'=x', \bar u ') = \lTlW_{g'}\,t'\}\\
  && \cup \; \{ g' \mapsto \emptyset \mid g' \in \G, \bot  = \lTlW_{g'}\,t'\ \}\\
  &=& W\\[1ex]

  P' &=&  \{ g' \mapsto \minLSince\,t'\,\bar u' \mid g' \in \G, (\_,g'=x', \bar u ') = \lTlW_{g'}\,t'\}\\
  && \cup \; \{ g' \mapsto \{\emptyset\} \mid g' \in \G, \bot  = \lTlW_{g'}\,t'\ \}\\
  &=& P\\[1ex]

  \sigma' &=&  \{ x' \mapsto \{t'(x)\} \mid x' \in \X \} \cup \{ g' \mapsto \emptyset \mid g' \in \G, \bot = \lTlW_{g'}\,t'\} \\
  && \cup \; \{ g' \mapsto \{\sigma_{j-1}\,x\}  \mid g' \in \G, ((j-1,u_{j-1},\sigma_{j-1}),g'=x', \_) = \lTlW_{g'}\,t'\} \\
  &=& \sigma \oplus \{ x \mapsto \{t'(x)\} \} \\
  &=& \sigma \oplus \{ x \mapsto \{\sigma_{j'-1}\,x'\} \mid \lW_g\,t'= ((j'-1,u_{j'-1},\sigma_{j'-1}), g = x', \_)  \}\\
  &=& \sigma \oplus \{ x \mapsto \{\sigma_{j'-1}\,x'\} \mid \lW_g\,t = ((j'-1,u_{j'-1},\sigma_{j'-1}), g = x', \_)  \}
\end{array}
\]
Thus, $W = W' \sqsubseteq W^{\sharp} \sqsubseteq W^{\sharp'}$ and $P = P' \sqsubseteq P^{\sharp} \sqsubseteq P^{\sharp'}$.
Also $\sigma\,y = \sigma'\,y$ and therefore, $\sigma'\,y\subseteq(\gamma_\D\circ\sigma^{\sharp'})\,y$
for $y \not\equiv x$.
For $y \equiv x$, we consider two cases:
\begin{itemize}
  \item Last write to $g$ is thread-local ($\lTlW_g\,t = ((j'-1,u_{j'-1},\sigma_{j'-1}), g = x', \bar u'')$):
  Then $\sigma\,g= \{\sigma_{j'-1}\,x'\} \subseteq (\gamma\circ\sigma^\sharp)\,g$, thus
 $\sigma'\,x \subseteq (\gamma\circ\sigma^{\sharp''})\,x$
  and accordingly, $\sigma' \subseteq \gamma\circ\sigma^{\sharp'}$.
  \item Last write to $g$ is non-thread-local. Then
  \[
    \begin{array}{lll}
    \sigma'\,x &\subseteq& \bigcup \{ \eval_g(\Get^{i-1}\,[g,a,S',w]) \mid a \in S, S \cap S' = \emptyset, w \subseteq \M \\
    && \qquad \exists S'' \in P:\,S'' \cap w = \emptyset\\
    && \qquad \exists S''' \in P:\,a \notin S''' \} \qquad \text{(By \cref{prop:read})} \\
    &\subseteq&  \bigcup \{ \eval_g(\Get\CWCU[g,a,S',w]) \mid a \in S, S \cap S' = \emptyset, w \subseteq \M \\
    && \qquad \exists S'' \in P^\sharp:\,S'' \cap w = \emptyset\\
    && \qquad \exists S''' \in P^\sharp:\,a \notin S''' \} \qquad \text{(By Induction Hypothesis)} \\
    &\subseteq&  \bigcup \{ \gamma_\D(\Get_\ModWC[g,a,S',w]) \mid a \in S, S \cap S' = \emptyset, w \subseteq \M \\
    && \qquad \exists S'' \in P^\sharp:\,S'' \cap w = \emptyset\\
    && \qquad \exists S''' \in P^\sharp:\,a \notin S''' \} \\
    &\subseteq& \gamma_\D (\bigsqcup \{ (\Get_\ModWC[g,a,S',w]) \mid a \in S, S \cap S' = \emptyset, w \subseteq \M \\
    && \qquad \exists S'' \in P^\sharp:\,S'' \cap w = \emptyset\\
    && \qquad \exists S''' \in P^\sharp:\,a \notin S''' \} \sqcup \sigma^\sharp\,g) \\
    &=& (\gamma_\D\circ\sigma^{\sharp''})\,x
       \subseteq (\gamma_\D\circ\sigma^{\sharp'})\,x  \\
    \end{array}
  \]
  and thus $\sigma' \subseteq \gamma_\D\circ\sigma^{\sharp''} \subseteq \gamma_\D\circ\sigma^{\sharp'}$.
\end{itemize}
Altogether, $t' \in \Get\CWCU\,[v,S]$ for all $t \in \Get^{i-1}\,[u,S]$.
We conclude that the return value of
$\sem{[u,S],x=g}\CWCU\,\Get^{i-1}$ is subsumed by the value $\Get\CWCU\,[v,S]$
and since the constraint causes no side-effects, the claim holds.

\vspace{1em}
\noindent Next, for $\lock(a)$, $a \in \M$:
\[
\begin{array}{lll}
  \sem{[u,S],\lock(a)}\CWCU\,\Get\CWCU &=& \Let\;T' = \bigcup \{ \Get\CWCU\,[g,a,S,w] \mid g\in\G, S \subseteq \M, w \subseteq \M \}\;\In \\
  & & (\emptyset,\sem{e}(\Get\CWCU\,[u,S],T')	\\[1ex]

  \sem{[u,S],\lock(a)}^\sharp_\ModWCL\Get_\ModWC	&=&
  \Let\;(W^\sharp,P^\sharp,\sigma^\sharp) = \Get_\ModWC\,[u,S]\;\In	\\
  & &	(\emptyset, (W^\sharp,P^\sharp,\sigma^\sharp))
\end{array}
\]
Let $\Get_\ModWC\,[u,S] = (W^\sharp,P^\sharp,\sigma^\sharp)$ and
$\Get_\ModWC\,[v,S\cup\{a\}] = (W^{\sharp'},P^{\sharp'},\sigma^{\sharp'})$ the
value provided by $\Get_\ModWC$ for the end point of the given control-flow edge and lockset.
Since $\Get_\ModWC$ is a post-solution of $\C_\ModWC$,
$W^{\sharp} \sqsubseteq W^{\sharp'}$, $P^{\sharp} \sqsubseteq P^{\sharp'}$, and $\sigma^{\sharp''} \sqsubseteq \sigma^{\sharp'}$
hold.
Then, by definition:
\[
\begin{array}{lll}
\Get\CWCU[v,S\cup\{a\}] = \gamma_{v,S\cup\{a\}}(W^{\sharp'},P^{\sharp'},\sigma^{\sharp'}) &=&
    \{ t \in \T_{S\cup\{a\}} \mid \loc\,t=v, \beta\,t = (W,P,\sigma),\\
    && \quad W \sqsubseteq W^{\sharp'}, P \sqsubseteq P^{\sharp'}, \sigma \in \gamma\,\sigma^{\sharp'} \}
\end{array}
\]
For every trace $t \in \Get^{-1}\,[u,S]$, let $\beta\,t = (W,P,\sigma)$. By induction hypothesis,
$W \sqsubseteq W^\sharp$, $P \sqsubseteq P^\sharp$, and
$\sigma \subseteq \gamma_\D\circ\sigma^\sharp$. Let $t' \in \sem{e}(\{t\}, \bigcup \{ \Get\CWCU\,[g,a,S,w] \mid g\in\G, S \subseteq \M, w \subseteq \M \})$,
then $\loc(t') = v$, $t' \in \T_{S\cup\{a\}}$, and
\[
\begin{array}{lll}
  \beta\,t' &=& (W',P',\sigma') \text{ where: }\\[1ex]

  W' &=& \{ g \mapsto \{ L_{t'}[\bar u'] \} \mid g \in \G, (\_,g=x, \bar u ') = \lTlW_{g}\,t'\}\\
  && \cup \; \{ g \mapsto \emptyset \mid g \in \G, \bot  = \lTlW_{g}\,t'\ \}\\
  &=& W\\[1ex]

  P' &=&  \{ g \mapsto \minLSince\,t'\,\bar u' \mid g \in \G, (\_,g=x, \bar u ') = \lTlW_{g}\,t'\}\\
  && \cup \; \{ g \mapsto \{\emptyset\} \mid g \in \G, \bot  = \lTlW_{g}\,t'\ \}\\
  &=& P\\[1ex]

  \sigma' &=&  \{ x \mapsto t'(x) \mid x \in \X \} \cup \{ g \mapsto \bot \mid g \in \G, \bot = \lTlW_{g}\,t'\} \\
  && \cup \; \{ g \mapsto \sigma_{j-1}\,x  \mid g \in \G, ((j-1,u_{j-1},\sigma_{j-1}),g=x, \_) = \lTlW_{g}\,t'\} \\
  &=& \sigma
\end{array}
\]
Thus, $W' = W \sqsubseteq W^\sharp \sqsubseteq W^{\sharp'}$, $P' = P \sqsubseteq P^\sharp \sqsubseteq P^{\sharp'}$,
$\sigma' = \sigma \in \gamma (\sigma^\sharp) \subseteq \gamma (\sigma^{\sharp'})$.
Altogether, $t' \in \Get\CWCU\,[v,S\cup\{a\}]$ for all $t \in \Get^{i-1}\,[u,S]$.
We conclude that the return value of
$\sem{[u,S],\lock(a)}\CWCU\, \Get^{i-1}$ is subsumed by the value $\Get\CWCU\,[v,S\cup\{a\}]$
and since the constraint causes no side-effects, the claim holds.

\vspace{1em}
\noindent Next, for $\unlock(a)$ $a \in \M$:
\[
\begin{array}{lll}
  \sem{[u,S],\unlock(a)}\CWCU\,\Get\CWCU
  &=&	\Let\;T = \sem{e}(\Get\CWCU\,[u,S])\;\In	\\
  & & \Let\;\rho = \{[g,a,S\setminus \{ a \}, w ]\mapsto \{t\} \mid  t \in T, g \in \G, w \subseteq \M \\
  & & \qquad (\lTlW _g\,t = (\bar u, g =x, \bar u ') \land L_t[\bar u'] \subseteq w) \lor\\
  & & \qquad (\lTlW _g\,t = \bot) \}\;\In \\
  & &	(\rho, T)\\[1ex]

	\sem{[u,S],\unlock(a)}^\sharp_\ModWCL\Get_\ModWC
		&=&	\Let\;(W^\sharp,P^\sharp,\sigma^\sharp) = \Get_\ModWC\,[u,S]\;\In	\\
		& & \Let\;P^{\sharp''} = \{ g \mapsto P^\sharp\,g \sqcup \{S \setminus \{a\}\} \mid g \in \G \}\;\In\\
    & & \Let\;\rho^\sharp = \{ [g,a,S \setminus \{a\},w] \mapsto \sigma^\sharp\,g \mid g \in \G, w'\in W^\sharp\,g, w' \subseteq w \}\;\In\\
		& &	(\rho^\sharp, (W^\sharp,P^{\sharp''},\sigma^\sharp))
\end{array}
\]
Let $\Get_\ModWC\,[u,S] = (W^\sharp,P^\sharp,\sigma^\sharp)$ and
$\Get_\ModWC\,[v,S\setminus\{a\}] = (W^{\sharp'},P^{\sharp'},\sigma^{\sharp'})$ the
value provided by $\Get_\ModWC$ for the end point of the given control-flow edge and lockset.
Since $\Get_\ModWC$ is a post-solution of $\C_\ModWC$,
$W^{\sharp} \sqsubseteq W^{\sharp'}$, $P^{\sharp''} \sqsubseteq P^{\sharp'}$, and $\sigma^{\sharp} \sqsubseteq \sigma^{\sharp'}$
hold.
Then, by definition:
\[
\begin{array}{lll}
\Get\CWCU[v,S\setminus\{a\}] = \gamma_{v,S\setminus\{a\}}(W^{\sharp'},P^{\sharp'},\sigma^{\sharp'}) &=&
    \{ t \in \T_{S\setminus\{a\}} \mid \loc\,t=v, \beta\,t = (W,P,\sigma),\\
    && \quad W \sqsubseteq W^{\sharp'}, P \sqsubseteq P^{\sharp'}, \sigma \in \gamma\,\sigma^{\sharp'} \}
\end{array}
\]
For every trace $t \in \Get^{i-1}\,[u,S]$, let $\beta\,t = (W,P,\sigma)$. By induction hypothesis,
$W \sqsubseteq W^\sharp$, $P \sqsubseteq P^\sharp$, and
$\sigma \subseteq \gamma_\D\circ\sigma^\sharp$. Let $t'= \sem{e}\,\{t\}$,
then $\loc(t') = v$, $t' \in \T_{S\setminus\{a\}}$, and
\[
\begin{array}{lll}
  \beta\,t' &=& (W',P',\sigma') \text{ where: }\\[1ex]

  W' &=& \{ g \mapsto \{ L_{t'}[\bar u'] \} \mid g \in \G, (\_,g=x, \bar u ') = \lTlW_{g}\,t'\}\\
  && \cup \; \{ g \mapsto \emptyset \mid g \in \G, \bot  = \lTlW_{g}\,t'\ \}\\
  &=& W\\[1ex]

  P' &=&  \{ g \mapsto \minLSince\,t'\,\bar u' \mid g \in \G, (\_g=x, \bar u ') = \lTlW_{g}\,t'\}\\
  && \cup \; \{ g \mapsto \{\emptyset\} \mid g \in \G, \bot  = \lTlW_{g}\,t'\ \}\\
  &=& \{ g \mapsto ((\minLSince\,t\,\bar u') \sqcup \{S\setminus\{a\}\}) \mid g \in \G, (\_,g=x, \bar u ') = \lTlW_{g}\,t\}\\
  && \cup \; \{ g \mapsto \{\emptyset\} \mid g \in \G, \bot  = \lTlW_{g}\,t'\ \}\\
  &=& \{ g \mapsto (P\,g \sqcup \{S \setminus \{a\}\}) \mid g \in \G \}\\[1ex]

  \sigma' &=&  \{ x \mapsto t'(x) \mid x \in \X \} \cup \{ g' \mapsto \bot \mid g' \in \G, \bot = \lTlW_{g'}\,t'\} \\
  && \cup \; \{ g' \mapsto \sigma_{j-1}\,x  \mid g' \in \G, ((j-1,u_{j-1},\sigma_{j-1}),g'=x, \_) = \lTlW_{g'}\,t'\} \\
  &=& \sigma
\end{array}
\]
Thus,
\[
\begin{array}{lll}
  W' &=& W \sqsubseteq W^\sharp \sqsubseteq W^{\sharp'}\\
  P' &=& \{ g \mapsto (P\,g \sqcup \{S \setminus \{a\}\}) \mid g \in \G \} \\
  &\sqsubseteq& \{ g \mapsto (P^\sharp\,g \sqcup \{S \setminus \{a\}\}) \mid g \in \G \} = P^{\sharp''}\\
  &\sqsubseteq& P^{\sharp'}\\
  \sigma' &=& \sigma \in \gamma (\sigma^\sharp) \subseteq \gamma (\sigma^{\sharp'})
\end{array}
\]
Altogether, $t' \in \Get\CWCU\,[v,S\cup\{a\}]$ for all $t \in \Get^{i-1}\,[u,S]$.
We conclude that the return value of
$\sem{[u,S],\unlock(a)}\CWCU\,\Get^{i-1}$ is subsumed by the value $\Get\CWCU\,[v,S\setminus\{a\}]$.
Next, we consider the side-effects of the corresponding right-hand-side functions.
For each $g \in \G$, we distinguish two cases for $t'$:
\begin{itemize}
  \item $\lTlW _g\,t' = \bot$: Then side-effects $\{ [g,a,S\setminus \{a\},w] \mapsto \{ t' \} \mid w \subseteq \M \}$
  are caused.
  These are accounted for by construction of $\Get\CWCU$:
  \[
    t' \in \{ t \in \T_{S\setminus\{a\}} \mid \last\,t=\unlock(a), \lTlW_g\,t = \bot \} \subseteq \Get\CWCU[g,a,S\setminus \{a\},w]
  \]
  \item $\lTlW _g\,t' = ((j-1,u_{j-1},\sigma_{j-1}),g=x, \bar u ')$: Then the side-effects caused by $\C\CWCU$ and $\C_\ModWC$
  for $g$, respectively, are given by
  \[
  \begin{array}{lll}
    \rho' &=& \{ [g,a,S\setminus \{a\},w] \mapsto \{t'\} \mid L_t[\bar u'] \subseteq w \}\\
    \rho^{\sharp'} &=& \{ [g,a,S \setminus \{a\},w''] \mapsto \sigma^\sharp\,g \mid w'\in W^\sharp\,g, w' \subseteq w'' \}\\
  \end{array}
  \]
  We remark that $\sigma\,g= \{\sigma_{j-1}\,x\}\subseteq
  (\gamma_\D\circ\sigma^{\sharp})\,g$, and that since $W \sqsubseteq W^\sharp$,
  there is $w' \in W^\sharp\,g$ where $w' \subseteq L_t[\bar u']$.
\end{itemize}
Hence, all side-effects for $\unlock(a)$ of $\C\CWCU$ are accounted for in $\Get\CWCU$, and the claim holds.

\vspace{1em}
\noindent Next, for $x = \create(u_1)$:
\[
\begin{array}{lll}
  \sem{[u,S],x = \create(u_1)}\CWCU\,\Get\CWCU	&=&
		\Let\;T = \sem{e}(\Get\CWCU\,[u,S])\;\In	\\
  &&		(\{[u_1,\emptyset]\mapsto\new\,u_1\,(\Get\CWCU\,[u,S])\},T)\\[1ex]

  \sem{[u,S],x = \create(u_1)}^\sharp_\ModWCL\Get_\ModWC	&=&	\Let\;(W^\sharp,P^\sharp,\sigma^\sharp) = \Get_\ModWC\,[u,S]\;\In	\\
				& & \Let\; W^{\sharp'''} = \{ g \mapsto \emptyset \mid g \in \G \}\;\In\\
				& & \Let\; P^{\sharp'''} = \{ g \mapsto \{ \emptyset \} \mid g \in \G \} \;\In\\
				& & \Let\; i^\sharp = \nu^\sharp\,u\,(W,P,\sigma)\,u_1\;\In\\
				& & \Let\; \sigma^{\sharp'''} = \sigma \oplus (\{ \self \mapsto i^\sharp\} \cup \{ g \mapsto \bot \mid g \in \G\})\;\In\\
        & & \Let\; \sigma^{\sharp''} = \sigma \oplus \{ x \mapsto i^\sharp \}\;\In\\
        & & \Let\; \rho = \{ [u_1,\emptyset] \mapsto (W^{\sharp'''},P^{\sharp'''},\sigma^{\sharp'''}) \}\;\In\\
				& &	(\rho,(W^\sharp,P^\sharp,\sigma^{\sharp''}))
\end{array}
\]
Let $\Get_\ModWC\,[u,S] = (W^\sharp,P^\sharp,\sigma^\sharp)$ and
$\Get_\ModWC\,[v,S] = (W^{\sharp'},P^{\sharp'},\sigma^{\sharp'})$ the
value provided by $\Get_\ModWC$ for the end point of the given control-flow edge and lockset.
Since $\Get_\ModWC$ is a post-solution of $\C_\ModWC$,
$W^{\sharp} \sqsubseteq W^{\sharp'}$, $P^{\sharp} \sqsubseteq P^{\sharp'}$,
and $\sigma^{\sharp''} \sqsubseteq \sigma^{\sharp'}$
hold.
Then, by definition:
\[
\begin{array}{lll}
\Get\CWCU[v,S] = \gamma_{v,S}(W^{\sharp'},P^{\sharp'},\sigma^{\sharp'}) &=&
    \{ t \in \T_{S} \mid \loc\,t=v, \beta\,t = (W,P,\sigma),\\
    && \quad W \sqsubseteq W^{\sharp'}, P \sqsubseteq P^{\sharp'}, \sigma \in \gamma\,\sigma^{\sharp'} \}
\end{array}
\]
For every trace $t \in \Get^{i-1}\,[u,S]$, let $\beta\,t = (W,P,\sigma)$. By induction hypothesis,
$W \sqsubseteq W^\sharp$, $P \sqsubseteq P^\sharp$, and
$\sigma \subseteq \gamma_\D\circ\sigma^\sharp$. Let $t'= \sem{e}\,\{t\}$,
then $\loc(t') = v$, $t' \in \T_{S}$, and
\[
\begin{array}{lll}
  \beta\,t' &=& (W',P',\sigma') \text{ where: }\\[1ex]

  W' &=& \{ g \mapsto \{ L_{t'}[\bar u'] \} \mid g \in \G, (\_,g=x, \bar u ') = \lTlW_{g}\,t'\}\\
  && \cup \; \{ g \mapsto \emptyset \mid g \in \G, \bot  = \lTlW_{g}\,t'\ \}\\
  &=& W\\[1ex]

  P' &=&  \{ g \mapsto \minLSince\,t'\,\bar u' \mid g \in \G, (\_,g=x', \bar u ') = \lTlW_{g}\,t'\}\\
  && \cup \; \{ g \mapsto \{\emptyset\} \mid g \in \G, \bot  = \lTlW_{g}\,t'\ \}\\
  &=& P\\[1ex]

  \sigma' &=&  \{ x' \mapsto \{t'(x')\} \mid x' \in \X \} \cup \{ g \mapsto \emptyset \mid g \in \G, \bot = \lTlW_{g}\,t'\} \\
  && \cup \; \{ g \mapsto \{\sigma_{j-1}\,x'\}  \mid g \in \G, ((j-1,u_{j-1},\sigma_{j-1}),g=x', \_) = \lTlW_{g}\,t'\} \\
  &=& \sigma \oplus \{ x \mapsto \{\nu\,t\} \}
\end{array}
\]
Since by definition $\nu\,t \in\gamma_D (\nu^\sharp\,u\,(V^\sharp,L^\sharp,\sigma^\sharp)\,u_1) $, thus
\[
\begin{array}{lll}
  W' &=& W \sqsubseteq W^{\sharp} \sqsubseteq W^{\sharp'}\\
  P' &=& P \sqsubseteq P^{\sharp} \sqsubseteq P^{\sharp'}\\
  \sigma' &=& \sigma \oplus \{ x \mapsto \{\nu\,t\} \} \subseteq
    \gamma_\D\circ(\sigma^{\sharp} \oplus \{ x \mapsto \nu^\sharp\,u\,(W^\sharp,P^\sharp,\sigma^\sharp)\,u_1 \})
  = \gamma_\D\circ\sigma^{\sharp''}\\
  &\subseteq& \gamma_\D\circ\sigma^{\sharp'}
\end{array}
\]
Altogether, $t' \in \Get\CWCU\,[v,S]$ for all $t \in \Get^{i-1}\,[u,S]$.
We conclude that the return value of
$\sem{[u,S],x=\create(u_1)}\CWCU\,\Get^{i-1}$ is subsumed by the value $\Get\CWCU\,[v,S]$.
Next, we consider the side-effects of the corresponding right-hand-side functions for $t \in \Get^{i-1}\,[u,S]$:
\[
  \begin{array}{lll}
    \rho' &=& [u_1,\emptyset] \mapsto\new\,u_1\,\{ t \} \\
    \rho^{\sharp'} &=& [u_1,\emptyset] \mapsto (W^{\sharp'''},P^{\sharp'''},\sigma^{\sharp'''})\\
  \end{array}
\]
Let $t'' = \new\,u_1\,\{ t \}$. Then,
\[
\begin{array}{lll}
  \beta\,t'' &=& (W''',P''',\sigma''') \text{ where: }\\[1ex]

  W''' &=& \{ g \mapsto \{ L_{t''}[\bar u'] \} \mid g \in \G, (\_,g=x', \bar u ') = \lTlW_{g}\,t''\}\\
  && \quad \cup \; \{ g \mapsto \emptyset \mid g \in \G, \bot  = \lTlW_{g}\,t''\ \}\\
  &=& \{ g \mapsto \emptyset \mid g \in \G \}\\[1ex]

  P''' &=&  \{ g \mapsto \minLSince\,t''\,\bar u' \mid g \in \G, (\_,g=x', \bar u ') = \lTlW_{g}\,t''\}\\
  && \quad \cup \; \{ g \mapsto \{\emptyset\} \mid g \in \G, \bot  = \lTlW_{g}\,t''\ \}\\[1ex]
  &=& \{ g \mapsto \{\emptyset\} \mid g \in \G \}\\[1ex]

  \sigma''' &=&  \{ x' \mapsto \{t''(x')\} \mid x' \in \X \} \cup \{ g \mapsto \emptyset \mid g \in \G, \bot = \lTlW_{g}\,t''\} \\
  && \quad \cup \; \{ g \mapsto \{\sigma_{j-1}\,x\}  \mid g \in \G, ((j-1,u_{j-1},\sigma_{j-1}),g=x, \_) = \lTlW_{g}\,t''\} \\
  &=& \sigma \oplus (\{ \self \mapsto \{\nu\,t\} \} \cup \{ g \mapsto \emptyset \mid g \in \G \} )
\end{array}
\]
Since by definition $\nu\,t \in\gamma_D (\nu^\sharp\,u\,(W^\sharp,P^\sharp,\sigma^\sharp)\,u_1) $, thus
\[
\begin{array}{lll}
  W''' &=& W^{\sharp'''}\\
  P''' &=& P^{\sharp'''}\\
  \sigma''' &=& \sigma \oplus (\{ \self \mapsto \{\nu\,t\} \} \cup \{ g \mapsto \emptyset \} ) \\
  &\subseteq& \gamma_\D\circ(\sigma^{\sharp} \oplus
    (\{ \self \mapsto \nu^\sharp\,u\,(V^\sharp,L^\sharp,\sigma^\sharp)\,u_1 \}  \cup \{ g \mapsto \bot \}))\\
  &=& \gamma_\D\circ\sigma^{\sharp'''}
\end{array}
\]
We remark that $(W^{\sharp'''},P^{\sharp'''},\sigma^{\sharp'''}) \sqsubseteq  \Get_\ModWC\,[u_1,\emptyset]$ holds as
$\Get_\ModWC$ is a post-solution of $C_\ModWC$.
Thus,
\[
t'' \in \Get\CWCU[u_1,\emptyset] = \gamma_{u_1,\emptyset}(\Get_\ModWC\,[u_1,\emptyset])
\]
Hence, all side-effects for $x = \create(u_1)$ of $\C\CWCU$ are accounted for in $\Get\CWCU$.
This concludes the proof.
\qed
\end{proof}

%% file: correctness/protection-based.tex
\subsection{Protection-Based Reading}\label{s:soundness-protection-based}

To prove the \emph{Protection-Based Reading} analysis sound, we show that we can construct from its analysis result
a post-solution of the constraint system for the \emph{Write-Centered Reading} analysis described in \cref{s:write-centered}.

To simplify the proof, we make some minor adjustments to the constraint system $\C_\ProtectionBased$ and call this new constraint system $\C_\ModPB$.
In particular, we modify $\sem{[u,S],\unlock(a)}^\sharp$ and  $\sem{[u,S],\unlock(m_g)}^\sharp$ by introducing some additional
side-effects for unlocking.
\[
\begin{array}{lll}
\sem{[u,S],\unlock(a)}^\sharp\Get	&=&	\Let\;(P,\sigma) = \Get\,[u,S]\;\In	\\
	& &	\Let\;P' = \{g\in P\mid ((S\setminus\{a\})\cap\MM[g])\neq\emptyset\}\;\In	\\
	& &	\Let\;\rho = \{[g]\mapsto\sigma\,g\mid g \in G, (\MM[g] \setminus \{ m_g \}) \not\subseteq S \setminus \{a\}\} \\
        & &     \qquad \cup \{ [g]' \mapsto \sigma\,g \mid a\notin \MM[g] \}\;\In \\
	& &	(\rho,(P',\sigma))\\[1ex]

\sem{[u,S],\unlock(m_g)}^\sharp\Get	&=&	\Let\;(P,\sigma) = \Get\,[u,S]\;\In	\\
	& &	\Let\;P' = \{g'\in P\mid ((S\setminus\{m_g\})\cap\MM[g'])\neq\emptyset\}\;\In	\\
        & &	\Let\;\rho = \{[g']\mapsto\sigma\,g'\mid g' \in G, (\MM[g'] \setminus \{ m_{g'} \}) \not\subseteq S \setminus \{m_g\}\} \\
        & &	\qquad \cup \{ [g']' \mapsto \sigma\,g' \mid g' \in \G \}\;\In \\
	& &	(\rho,(P',\sigma))	\\[1ex]
\end{array}
\]
Compared to the original formulation, there are two changes:
\begin{enumerate}
        \item The local values of all globals $g$ that are not protected by $a$, are side-effected to $[g]'$.
        However, $[g]'$ already receives all written values immediately by means of a side-effect triggered at the $\unlock(m_g)$ directly following
        a write to $g$.
        \item The local values of globals $g$ that are no longer totally protected at the end point of the edge, i.e., where at least one of the
        mutexes $\MM[g] \setminus \{m_g\}$ is no longer held, are additionally side-effected to $[g]$. We distinguish two cases: For
        $a \in \MM[g]$, this side-effect also happens in the original formulation. Otherwise, there is a mutex $a' \in \MM[g] \setminus \{ m_g\}:
        a' \not\in S$. Since a write always happens with the full set of protecting mutexes $\MM[g]$, there must have been
        an $\unlock(a')$ since the last write at which point the value was already side-effected to $[g]$ in the original formulation as well.
\end{enumerate}
Therefore, these additional side-effects have no influence on solutions of the system and we obtain:
\begin{proposition}
The unique least solution of constraint system $\C_\ProtectionBased$ is also the unique least solution of $\C_\ModPB$.
\end{proposition}
Let us introduce a concretization mapping $\gamma$ defined by
\[
\begin{array}{lll}
\gamma\,(P,\sigma) &=& (W',P',\sigma') \qquad\text{where}	\\
        W' &=& \{g\mapsto\{\MM[g]\}\mid g\in\G\}		\\
        P' &=& \{g\mapsto \Iif\,g\in P\,\Then\,
                \{\{a\} \mid a \in \MM[g] \}\,\Eelse\,\{\emptyset\}\mid g\in\G\}	\\
	\sigma' &=& \sigma
\end{array}
\]
Moreover, we introduce a \emph{description relation} $\R_\ModPB$ between the sets of unknowns of the constraint systems
$\C_\WriteCentered$ and $\C_\ModPB$, respectively by
\[
\begin{array}{rlll}
\relax[u,S]	&\R_\ModPB& [u,S]	&\text{for all}\;u\in\N, S\subseteq\M	\\
\relax[g,a,S,w]	&\R_\ModPB& 	[g] &\text{for all}\;a\in \MM[g]\setminus\{m_g\}\\
\relax[g,a,S,w]	&\R_\ModPB& 	[g] &\text{for all}\;a\in\{m_g\}\cup(\M\setminus\MM[g]),( \MM[g] \setminus\{m_g\})\not\subseteq S\\
\relax[g,a,S,w]	&\R_\ModPB& 	[g]'&\text{for all}\;a\in\{m_g\}\cup(\M\setminus \MM[g]),(\MM[g] \setminus\{m_g\})\subseteq S\\
\end{array}
\]
\noindent where $g\in\G$, $\MM[g] \subseteq w$.

Let $\Get_\ModPB$ be the unique least solution of the constraint system $\C_\ModPB$.
We construct a mapping $\Get_\WriteCentered$ for the constraint system $\C_\WriteCentered$ from $\Get_\ModPB$
by $\Get_\WriteCentered[u,S] = \gamma(\Get_\ModPB[u,S])$
for all program points $u \in \N$ and locksets $S\subseteq\M$.
Moreover, we set
\[
\begin{array}{lll}
\Get_\WriteCentered\,[g, a, S, w] &=&  \left\{ \begin{array}{ll}
\Get_\ModPB\,[g]'	&\text{if}\; a\in\{m_g\}\cup(\M\setminus\MM[g])\land(\MM[g]\setminus\{m_g\})\subseteq S \land \MM[g] \subseteq w\\
\bot            &\text{if}\; \MM[g] \not\subseteq w \\
\Get_\ModPB\,[g] 	&\text{otherwise}\;  
\end{array}
        \right.
\end{array}
\]
\begin{theorem}\label{t:protection-based}
Then we have:
\begin{enumerate}
\item	$\Get_\ModPB\,[g]\sqsubseteq\Get_\ModPB\,[g]'$ holds for all $g\in\G$;
\item	$\Get_\WriteCentered$ is a post-solution of $\C_\WriteCentered$;
\item	Whenever $[g,a,S,w]\,\R_\ModPB\,[g]'$, then $\Get_\WriteCentered[g,a,S,w]\sqsubseteq\Get_\ModPB\,[g]'$;
\item	Whenever $[g,a,S,w]\,\R_\ModPB\,[g]$, then $\Get_\WriteCentered[g,a,S,w]\sqsubseteq\Get_\ModPB\,[g]$.
\end{enumerate}
\end{theorem}
Given that, by \cref{t:write-centered}, each post-solution of constraint system $\C_\WriteCentered$ is sound w.r.t.\ the concrete trace
semantics, this proposition implies that also the unique least solution of $\C_\ModPB$ is sound, and by extension the unique least solution of the
constraint system $\C_\ProtectionBased$ for \emph{Protection-Based Reading} described in \cref{s:protection-based}.

\begin{proof}
The proof is by verifying for each edge $(u,A,v)$ of the control-flow graph, each possible lockset $S$,
and $\Get_\WriteCentered$ constructed above that $\sem{[u,S],A}_\WriteCentered^\sharp\Get_\WriteCentered \sqsubseteq(\Get_\WriteCentered,\Get_\WriteCentered\,[v,S'])$ holds.
We exemplify this for the $\unlock(a)$ and $x=g$ operations.

We distinguish $\unlock(m_g)$ where $m_g \in \{ m_{g'} \mid g' \in \G \}$ and
$\unlock(a)$, $a\in \M \setminus \{m_g \mid g \in \G \}$. For $\unlock(m_g)$
\[
\begin{array}{lll}
\sem{[u,S],\unlock(m_g)}^\sharp_\WriteCentered\,\eta_\WriteCentered
		&=&	\Let\;(W_\WriteCentered,P_\WriteCentered,\sigma_\WriteCentered) = \Get_\WriteCentered\,[u,S]\;\In	\\
		& & \Let\;P'_\WriteCentered = \{ g' \mapsto P_\WriteCentered\,g' \sqcup \{S \setminus \{m_g\}\} \mid g' \in \G \}\;\In\\
		& & \Let\;\rho_\WriteCentered = \{ [g,m_g,S \setminus \{m_g\},w] \mapsto \sigma_\WriteCentered\,g \mid w\in W_\WriteCentered\,g \}\;\In\\
		& &	(\rho_\WriteCentered, (W_\WriteCentered,P_\WriteCentered',\sigma_\WriteCentered))
\end{array}
\]
Let $\Get_\ModPB\,[u,S] = (P_\ModPB,\sigma_\ModPB)$ and for $S'=S\setminus\{m_g\},$ $\Get_\ModPB\,[v,S'] = (P_\ModPB',\sigma_\ModPB')$ the
value provided by $\Get_\ModPB$ for the end point of the given control-flow edge.
In particular, $\gamma(P_\ModPB',\sigma_\ModPB') = (W_\WriteCentered'',P_\WriteCentered'',\sigma_\WriteCentered'')$,
where by definition,
\[
\begin{array}{lll}
        W_\WriteCentered'' &=& \{g'\mapsto\{\MM[g']\}\mid g'\in\G\}		\\
        P_\WriteCentered'' &=& \{g'\mapsto \Iif\,g'\in P_\ModPB'\,\Then\,
                \{\{a\} \mid a \in \MM[g'] \}\,\Eelse\,\{\emptyset\}\mid g'\in\G\}	\\
	\sigma_\WriteCentered'' &=& \sigma_\ModPB'
\end{array}
\]

Let $G$ denote the set of all $g'\in\G$ so that $S\cap\MM[g']=\{m_g\}$, i.e., no further protecting mutex of $g$ is held after the unlock operation.
In this case,
$P_\ModPB'\subseteq P_\ModPB\setminus G$ where in particular, for each $g'\in P_\ModPB'$, $S'\cap\MM[g']\neq\emptyset$.
By definition of $\Get_\WriteCentered$,
\[
	\begin{array}{lll}
        P_\WriteCentered &=& \{g'\mapsto \Iif\,g'\in P_\ModPB\,\Then\,
                \{\{a\} \mid a \in \MM[g'] \}\,\Eelse\,\{\emptyset\}\mid g'\in\G\}\quad\text{and therefore,}	\\
	P'_\WriteCentered &=& \{g'\mapsto (P_\WriteCentered\,g') \sqcup S'\mid g'\in\G\} \\
             &=& \{g'\mapsto \Iif\,g'\in P_\ModPB\,\Then\,
                \{\{a\} \mid a \in \MM[g'] \}\sqcup S'\,\Eelse\,\{\emptyset\}\sqcup S'\mid g'\in\G\}\\
             &=& \{g'\mapsto \Iif\,g'\in P_\ModPB\,\Then\,
                \{\{a\} \mid a \in \MM[g'] \}\sqcup S'\,\Eelse\,\{\emptyset\}\mid g'\in\G\}\quad\text{and therefore,}	\\
             &\sqsubseteq& \{g'\mapsto \Iif\,g'\in P_\ModPB'\,\Then\,
                \{\{a\} \mid a \in \MM[g'] \}\,\Eelse\,\{\emptyset\}\mid g'\in\G\}	\\
	     &=& P''_\WriteCentered
	\end{array}
\]
Moreover, $\sigma_\ModPB = \sigma_\WriteCentered \sqsubseteq \sigma_\WriteCentered'' = \sigma_\ModPB'$ and
$W_\WriteCentered = W_\WriteCentered''$. We conclude that the return value of
$\sem{[u,S],\unlock(m_g)}^\sharp_\WriteCentered\\,\eta_\WriteCentered$ is subsumed by the value $\Get_\WriteCentered[v,S']$.
It remains to check that also the side-effects produced by this constraint of $\C_\WriteCentered$ are
subsumed by the corresponding side-effects of the concretization applied to the side-effects of
$\C_\ModPB$ to corresponding unknowns (relative to $\R_\ModPB$).

The side-effects in $\C_\ModPB$ and $\C_\WriteCentered$ for $\unlock(m_g)$ and the assignments $\Get_\ModPB$ and $\Get_\WriteCentered$,
respectively, are given by
\[
\begin{array}{llll}
\rho_\ModPB &=&	\{[g']\mapsto\sigma_\ModPB\,g'\mid g' \in G, (\MM[g'] \setminus \{ m_{g'} \}) \not\subseteq S'\}  \cup \{ [g']' \mapsto \sigma_\ModPB\,g' \mid g' \in \G \} &\\[1ex]

\rho_\WriteCentered	&=& 	\{ [g',m_g,S',w]\mapsto \sigma_\WriteCentered\,g' \mid g' \in \G, w\in W_\WriteCentered\,g' \}
\end{array}
\] where $W_\WriteCentered\,g' = \{ \MM[g'] \}$ by definition.
For the side-effects of $\C_\WriteCentered$ we distinguish for each $g' \in G$ two separate cases:
\begin{itemize}
        \item $m_g \in \{m_{g'}\} \cup (\M \setminus \MM[g'])$ and $(\MM[g'] \setminus \{m_{g'}\}) \not\subseteq S'$: In this case,
        $[g',m_g,S',\MM[g']]\,\R_\ModPB\,[g']$ holds and the side-effect is accounted for by the corresponding side-effect to $\relax[g']$
        of $\C_\ModPB$.
        \item $m_g\in \{m_{g'}\} \cup (\M \setminus \MM[g'])$ and $(\MM[g'] \setminus \{ m_{g'} \}) \subseteq S'$: In this case
        $[g',m_g,S',\MM[g']]\,\R_\ModPB\,[g']'$ holds and the side-effect is also accounted for by the corresponding side-effect
        to $\relax[g']'$ of $\C_\ModPB$.
\end{itemize}
With the observations that $m_g \in \MM[g'] \setminus \{m_{g'}\}$ is a contradiction and $\sigma_\ModPB\,g =\sigma_\WriteCentered\,g$
all side-effects are accounted for, and the claim holds.

Now consider an $\unlock$ operation for some $a\in\M\setminus\{m_g\mid g\in\G\}$.
\[
\begin{array}{lll}
\sem{[u,S],\unlock(a)}^\sharp_\WriteCentered\,\eta_\WriteCentered
		&=&	\Let\;(W_\WriteCentered,P_\WriteCentered,\sigma) = \Get_\WriteCentered\,[u,S]\;\In	\\
		& & \Let\;P'_\WriteCentered = \{ g \mapsto P_\WriteCentered\,g \sqcup \{S \setminus \{a\}\} \mid g \in \G \}\;\In\\
		& & \Let\;\rho_\WriteCentered = \{ [g,a,S \setminus \{a\},w] \mapsto \sigma\,g \mid g \in \G, w\in W\,g \}\;\In\\
		& &	(\rho_\WriteCentered, (W_\WriteCentered,P_\WriteCentered',\sigma))
\end{array}
\]
Let $\Get_\ModPB\,[u,S] = (P_\ModPB,\sigma)$ and for $S'=S\setminus\{a\},$ $\Get_\ModPB\,[v,S'] = (P_\ModPB',\sigma')$ the
value provided by $\Get_\ModPB$ for the end point of the given control-flow edge.
In particular, $\gamma(P_\ModPB',\sigma') = (W_\WriteCentered'',P_\WriteCentered'',\sigma'')$,
where by definition,
\[
\begin{array}{lll}
        W_\WriteCentered'' &=& \{g\mapsto\{\MM[g]\}\mid g\in\G\}		\\
        P_\WriteCentered'' &=& \{g\mapsto \Iif\,g\in P_\ModPB'\,\Then\,
                \{\{a'\} \mid a' \in \MM[g] \}\,\Eelse\,\{\emptyset\}\mid g\in\G\}	\\
	\sigma'' &=& \sigma'
\end{array}
\]
The return value of
$\sem{[u,S],\unlock(a)}^\sharp_\WriteCentered\,\eta_\WriteCentered$ is subsumed by the value $\Get_\WriteCentered[v,S']$ by the same argument as for
$\unlock(m_g)$.
It remains to check that also the side-effects produced by this constraint of $\C_\WriteCentered$ are
subsumed by the corresponding side-effects of the concretization applied to the side-effects of
$\C_\ModPB$ to corresponding unknowns (relative to $\R_\ModPB$).

The side-effects in $\C_\ModPB$ and $\C_\WriteCentered$ for $\unlock(a)$ and the assignments $\Get_\ModPB$ and $\Get_\WriteCentered$,
respectively, are given by:
\[
\begin{array}{lllll}
\rho_\ModPB	&=& \{[g]\mapsto\sigma_\ModPB\,g\mid g \in G, (\MM[g] \setminus \{ m_g \}) \not\subseteq S' \}&\\
&&  \cup \{ [g]' \mapsto \sigma_\ModPB\,g \mid a\notin \MM[g] \} &\;\text{for}\quad \Get_\ModPB\,[u,S]=(P_\ModPB,\sigma_\ModPB)\\
\rho_\WriteCentered	&=& 	\{ [g,a,S',w]\mapsto \sigma_\WriteCentered\,g \mid g \in \G, w\in W_\WriteCentered\,g \}
\end{array}
\] where $W_\WriteCentered\,g = \{ \MM[g] \}$ by definition.
For the side-effects of $\C_\WriteCentered$ we distinguish for each $g \in G$ three separate case:
\begin{itemize}
        \item $a\in\MM[g]$: In this case, $[g,a,S',\MM[g]]\,\R_\ModPB\,[g]$ holds and the side-effect is accounted for by the corresponding
        side-effect to $\relax[g]$ of $\C_\ModPB$.
        \item $a \in (\M \setminus \MM[g])$ and $(\MM[g] \setminus \{m_g\}) \not\subseteq S'$: In this case
        $[g,a,S',\MM[g]]\,\R_\ModPB\,[g]$ holds and the side-effect is accounted for by the corresponding side-effect to $\relax[g]$
        of $\C_\ModPB$.
        \item $a\in (\M \setminus \MM[g])$ and $(\MM[g] \setminus \{m_g\}) \subseteq S'$: In this case
        $[g,a,S',\MM[g]]\,\R_\ModPB\,[g]'$ holds and the side-effect is also accounted for by the corresponding side-effect
        to $\relax[g]'$ of $\C_\ModPB$.
\end{itemize}
Additionally, $\sigma_\ModPB\,g =\sigma_\WriteCentered\,g$. Accordingly, the claim holds.

\medskip
\noindent Now consider a read from a global $x=g$.
\[
\begin{array}{lll}
        \sem{[u,S],x=g}^\sharp_\WriteCentered\Get_\WriteCentered	&=&	\Let\;(W_\WriteCentered,P_\WriteCentered,\sigma_\WriteCentered) = \Get_\WriteCentered\,[u,S]\;\In	\\
			& &\Let\;d = \sigma_\WriteCentered\, g \sqcup\bigsqcup\{\eta[g,a,S',w] \mid  a\in S, S\cap S' =\emptyset, \\
			& & \qquad \exists S'' \in P\,g: S'' \cap w = \emptyset, \\
			& & \qquad \exists S''' \in P\,g: a \notin S''' \}\;\In\\
                        & &\Let\; \sigma_\WriteCentered' = \sigma_\WriteCentered\oplus\{x\mapsto d\} \;\In\\
			& &	(\emptyset,(W_\WriteCentered,P_\WriteCentered,\sigma_\WriteCentered'))
\end{array}
\]

Let $\Get_\ModPB [u,S] = (P_\ModPB,\sigma_\ModPB)$ and $\Get_\ModPB [v,S] = (P_\ModPB',\sigma_\ModPB')$ the value provided by $\Get_\ModPB$ for the end point of the given control flow edge.
In particular, $\gamma(P_\ModPB',\sigma_\ModPB') = (W_\WriteCentered'',P_\WriteCentered'',\sigma_\WriteCentered'')$,
where by definition,
\[
\begin{array}{lll}
        W_\WriteCentered'' &=& \{g'\mapsto\{\MM[g']\}\mid g'\in\G\}		\\
        P_\WriteCentered'' &=& \{g'\mapsto \Iif\,g'\in P_\ModPB'\,\Then\,
                \{\{a\} \mid a \in \MM[g'] \}\,\Eelse\,\{\emptyset\}\mid g'\in\G\}	\\
	\sigma_\WriteCentered'' &=& \sigma_\ModPB' \\
\end{array}
\]
Since neither the constraints in $\C_\ModPB$ nor in $\C_\WriteCentered$ modify $P$,
$P_\ModPB \sqsubseteq P_\ModPB'$ and hence $P_\WriteCentered \sqsubseteq P_\WriteCentered''$. Also, $W_\WriteCentered = W_\WriteCentered'$.

\emph{Case 1.} $g \in P_\ModPB$, hence $\sigma_\ModPB \oplus \{ x \mapsto \sigma_\ModPB\,g \} \sqsubseteq \sigma_\ModPB'= \sigma_\WriteCentered''$,
$P_\WriteCentered\,g = \{\{a'\} \mid a' \in \MM[g] \}$,
\[
\begin{array}{lll}
\sigma_\WriteCentered' &=& \sigma_\WriteCentered\oplus\{x\mapsto \sigma_\WriteCentered\, g \sqcup\bigsqcup\{\Get_\WriteCentered[g,a,S',w] \mid  a\in S, S\cap S' =\emptyset,\\
        & & \qquad \exists S'' \in P_\WriteCentered\,g: S'' \cap w = \emptyset, \\
        & & \qquad \exists S''' \in P_\WriteCentered\,g: a \notin S''' \}\\
        &=&  \sigma_\WriteCentered\oplus\{x\mapsto \sigma_\WriteCentered\, g \sqcup\bigsqcup\{\Get_\WriteCentered[g,a,S',w] \mid  a\in S, S\cap S' =\emptyset,\\
        & & \qquad \exists S''  \in \{\{a'\} \mid a' \in \MM[g] \}: S'' \cap w = \emptyset, \\
        & & \qquad \exists S''' \in \{\{a'\} \mid a' \in \MM[g] \}: a \notin S''' \}  \quad \text{by construction of } \Get_\WriteCentered\\
        &=& \sigma_\WriteCentered\oplus\{x\mapsto \sigma_\WriteCentered\, g \sqcup \bot \}\\
        &\sqsubseteq& \sigma_\WriteCentered''
\end{array}
\]

\emph{Case 2.} $S \cap \MM[g] = \{m_g\} \land g \not\in P_\ModPB$, hence $\sigma_\ModPB \oplus \{ x \mapsto \sigma_\ModPB\,g \sqcup  \Get_\ModPB[g]' \}
\sqsubseteq \sigma_\ModPB'= \sigma_\WriteCentered''$, $P_\WriteCentered\,g = \{\emptyset\}$,
\[
\begin{array}{lll}
\sigma_\WriteCentered' &=& \sigma_\WriteCentered\oplus\{x\mapsto \sigma_\WriteCentered\, g \sqcup\bigsqcup\{\Get_\WriteCentered[g,a,S',w] \mid  a\in S, S\cap S' =\emptyset,\\
        & & \qquad \exists S'' \in P_\WriteCentered\,g: S'' \cap w = \emptyset, \\
        & & \qquad \exists S''' \in P_\WriteCentered\,g: a \notin S''' \}\\
        &=&  \sigma_\WriteCentered\oplus\{x\mapsto \sigma_\WriteCentered\, g \sqcup\bigsqcup\{\Get_\WriteCentered[g,a,S',w] \mid  a\in S, S\cap S' =\emptyset \}  \quad \text{by construction of } \Get_\WriteCentered\\
        &=& \sigma_\WriteCentered\oplus\{x \mapsto \sigma_\WriteCentered\, g \sqcup (\Get_\ModPB[g]' \sqcup \Get_\ModPB[g])\} \quad \text{by } \Get_\ModPB\,[g]\sqsubseteq\Get_\ModPB\,[g]'\\
        &=&  \sigma_\WriteCentered\oplus\{x \mapsto \sigma_\WriteCentered\, g \sqcup \Get_\ModPB[g]'\}\\
        &\sqsubseteq& \sigma_\WriteCentered''
\end{array}
\]

\emph{Case 3.} $S \cap \MM[g] \neq \{m_g\} \land g \not\in P_\ModPB$, hence $\sigma_\ModPB \oplus \{ x \mapsto \sigma_\ModPB\,g \sqcup  \Get_\ModPB[g] \}
\sqsubseteq \sigma_\ModPB'= \sigma_\WriteCentered''$, $P_\WriteCentered\,g = \{\emptyset\}$,
\[
\begin{array}{lll}
\sigma_\WriteCentered' &=& \sigma_\WriteCentered\oplus\{x\mapsto \sigma_\WriteCentered\, g \sqcup\bigsqcup\{\Get_\WriteCentered[g,a,S',w] \mid  a\in S, S\cap S' =\emptyset,\\
        & & \qquad \exists S'' \in P_\WriteCentered\,g: S'' \cap w = \emptyset, \\
        & & \qquad \exists S''' \in P_\WriteCentered\,g: a \notin S''' \}\\
        &=&  \sigma_\WriteCentered\oplus\{x\mapsto \sigma_\WriteCentered\, g \sqcup\bigsqcup\{\Get_\WriteCentered[g,a,S',w] \mid  a\in S, S\cap S' =\emptyset \}
          \quad \text{by construction of } \Get_\WriteCentered\\
        &=& \sigma_\WriteCentered\oplus\{x \mapsto \sigma_\WriteCentered\, g \sqcup \Get_\ModPB[g] \}\\
        &\sqsubseteq& \sigma_\WriteCentered''
\end{array}
\]
We conclude that the return value of
$\sem{[u,S],x=g}^\sharp_\WriteCentered\,\eta_\WriteCentered$ is subsumed by the value $\Get_\WriteCentered[v,S]$ and since the constraint causes no side-effects, the claim holds.
\end{proof}

%% file: analyses/evaluation.tex
\section{Experimental Evaluation}\label{s:experimental}
We have implemented the analyses described in the previous sections
as well as the side-effecting formulation of Miné's analysis (see~\cref{s:mine})
within the static analyzer framework \textsc{Goblint}, which analyzes C programs.
For \emph{Protection-Based Reading}, we implemented the variant that does not require prior information on the
locksets $\MM[g]$ protecting globals $g$, but instead discovers this information during the analysis.
The solvers in \textsc{Goblint} can handle the non-monotonicity in the side-effects this entails.

For experimental evaluation,
we use six multi-threaded \textsc{Posix} programs from
the \textsc{Goblint} benchmark suite\footnote{\url{https://github.com/goblint/bench}} and seven large \textsc{SV-Comp} benchmarks
in \texttt{c/ldv-linux-3.14-races/} from the \textsc{ConcurrencySafety-Main} category\footnote{\url{https://github.com/sosy-lab/sv-benchmarks}}.
The programs range from 1280 to 12778 physical LoC,
with logical LoC\footnote{Only lines with executable code, excluding struct and extern function declarations.} being between 600 and 3102.
The analyses are performed context-sensitively
with a standard points-to analysis for addresses and
inclusion/exclusion sets as the domain for integer values.
The evaluation was carried out on Ubuntu 20.04.1
and \textsc{OCaml}~4.11.1, running on a standard \textsc{Amd Epyc} processor.

We analyzed each of the programs with each of the analyses where the required analysis times
are presented in \cref{fig:benchmark-time}.
On smaller programs, \emph{Protection-Based Reading} is almost twice as fast as the others,
which have very similar running times.
On larger programs, the differences are much larger: \emph{Protection-Based Reading}
there is up to an order of magnitude faster, while the running times of the remaining
analyses grow with their sophistication.

Since the analyses use different local and global domains, their precision cannot be compared
directly via the constraint system solutions.
Instead, we record and compare the observable behavior in the form of abstract values of global variables read at program locations.
Our comparison reveals that, for 11 out of 13 programs,
all analyses are equally precise.
For the remaining two programs, \texttt{pfscan} and \texttt{ypbind},
all but Min\'e's analysis are equally precise, while Min\'e's was less precise at 6\% and 16\% of global reads, respectively.

Thus our experiments indicate that \emph{Protection-Based Reading} offers sufficient precision at a
significantly shorter analysis time, while the more involved \emph{Lock-} and
\emph{Write-Centered Reading} do not offer additional precision.
Moreover, the incomparability identified in the introduction can in fact be observed on
at least some real-world programs.
%
Still, more experimentation is required as the selection of
benchmarks may be biased towards programs using quite simple protection patterns.
Also, only one particular value domain for globals was considered.
\begin{figure}[tb]
	\centering
	\pgfplotslegendfromname{benchmark-time-legend} 
	\\[2ex]
	\begin{tikzpicture}
		\pgfplotstableset{
			create on use/benchmark-loc/.style={
				create col/assign/.code={
					\getthisrow{benchmark}\benchmark
					\getthisrow{loc}\loc
					\edef\entry{\benchmark\space(\loc)}
					\pgfkeyslet{/pgfplots/table/create col/next content}\entry
				}
			}
		}
		\pgfplotstableread{
			benchmark	loc	protection	mine-W	lock	write	write+lock
			ctrace	665	1.3	2.11	2.12	2.17	2.22
			pfscan	600	1.15	2.51	2.47	2.22	2.47
			knot	987	1.29	2.02	2.07	2.04	2.09
			aget	603	1.99	3.43	3.41	3.46	3.61
		}\timetablesmall
		\pgfplotstableread{
			benchmark	loc	protection	mine-W	lock	write	write+lock
			ypbind	1035	11.6	22.54	29.61	36.38	40.72
			smtprc	3102	14.81	30.44	31.34	31.75	32.83
			iowarrior	1358	2.72	20.37	22.14	16.7	23.59
			w83977af	1515	3.61	6.57	8.06	10.06	9.94
			adutux	1509	3.35	41.91	24.27	20.91	27.49
			tegra20	1560	3.29	8.73	7.9	11.07	11.24
			marvell1	2476	7.4	32.82	37.28	36.97	46.59
			marvell2	2476	7.83	34.11	43.06	41.48	52.54
			nsc	2394	6.78	16.02	14.96	16.81	17.2
		}\timetablebig

		\begin{groupplot}[
			group style={
				group size=2 by 1,
				horizontal sep=0.9cm,
			},
			height=5cm,
			ybar=0pt,
			enlarge y limits=upper,
			ymin=0,
			xtick=data,
			x tick label style={
				rotate=45,
				anchor=east,
			},
			xtick pos=bottom,
			ymajorgrids=true,
		]
			\nextgroupplot[
				width=0.395\textwidth,
				enlarge x limits=0.23,
				ybar legend,
				bar width=3pt,
				legend to name=benchmark-time-legend,
				legend columns=-1,
				symbolic x coords={pfscan (600),aget (603),ctrace (665),knot (987)},
				ytick distance=1,
				ylabel={Analysis time [s]},
				ylabel near ticks,
			]
				\addplot table [x=benchmark-loc,y=protection] {\timetablesmall};
				\addlegendentry{\hyperref[s:protection-based]{Protection-Based}};
				\addplot table [x=benchmark-loc,y=mine-W] {\timetablesmall};
				\addlegendentry{\hyperref[s:mine]{Min\'e}};
				\addplot table [x=benchmark-loc,y=lock] {\timetablesmall};
				\addlegendentry{\hyperref[s:lock-centered]{Lock-Centered}};
				\addplot table [x=benchmark-loc,y=write] {\timetablesmall};
				\addlegendentry{\hyperref[s:write-centered]{Write-Centered}};
				\addplot table [x=benchmark-loc,y=write+lock] {\timetablesmall};
				\addlegendentry{\hyperref[s:combined]{Combined}};
			\nextgroupplot[
				width=0.7\textwidth,
				enlarge x limits=0.09,
				bar width=3pt,
				symbolic x coords={ypbind (1035),iowarrior (1358),adutux (1509),w83977af (1515),tegra20 (1560),nsc (2394),marvell1 (2476),marvell2 (2476),smtprc (3102)},
				ytick distance=15,
			]
				\addplot table [x=benchmark-loc,y=protection] {\timetablebig};
				\addplot table [x=benchmark-loc,y=mine-W] {\timetablebig};
				\addplot table [x=benchmark-loc,y=lock] {\timetablebig};
				\addplot table [x=benchmark-loc,y=write] {\timetablebig};
				\addplot table [x=benchmark-loc,y=write+lock] {\timetablebig};
		\end{groupplot}
	\end{tikzpicture}
	\caption{Analysis times per benchmark program (logical LoC in parentheses).}
	\label{fig:benchmark-time}
\end{figure}
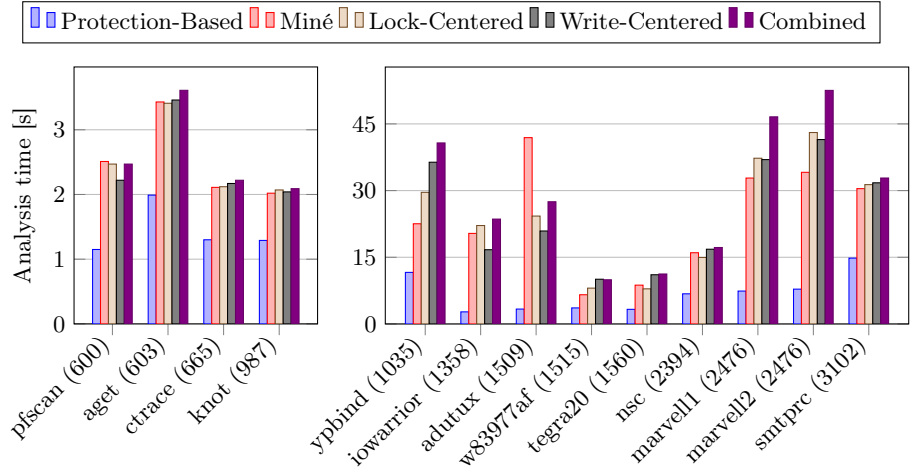

%% file: conclusion.tex
\section{Conclusion}\label{s:conclusion}

We have provided enhanced versions of the analyses by Min\'e \cite{Mine2012} as well as by
Vojdani \cite{Vojdani2010,Vojdani2016}. To Min\'e's original analysis, we added
lazy reading of globals and restricting local copies of globals to the
values written by the ego thread.
Vojdani's approach was purged of the assumption of common protecting mutexes,
while additionally, background locksets are taken into account to exclude certain written values
from being read.
For a better comparison, we relied on side-effecting constraint systems as a convenient
framework within which all analyses could be formalized.
That framework also enabled us to specify a concrete semantics of \emph{local traces} w.r.t.\
which all analyses could be proven correct.
We also provided an implementation of
all these analyses and practically compared them for precision and efficiency.
Interestingly, the simplest of all analyses still provided decent precision
while out-performing the others.

The given comparison and enhancements refer just to the first and most fundamental analysis
introduced by Min\'e. We would therefore like to address possible extensions
to \emph{relational} analyses in future work.
Also, we would like to explore how the framework can be extended
so that \emph{weak memory} effects can conveniently be taken into account.

\paragraph{Acknowledgements.} This work was supported by Deutsche Forschungsgemeinschaft (DFG)
– 378803395/2428 \textsc{ConVeY} and the Estonian Research Council grant PSG61.

%% file: analyses/mine-appendix.tex
\section{Side-Effecting Formulation of the Analysis by Miné}\label{s:mine}
We further detail the side-effecting formulation of Miné's original analysis from \cite{Mine2012}
adapted to our setting (non-relational, no unique thread \emph{ids}, no real-time features,
globals receive their initial values via an assignment).
On top of the mechanism already described in \cref{s:lock-centered} that handles synchronized
accesses to variables (\emph{synchronized interferences} in Miné's terminology), there also exist \emph{weak
interferences}, i.e., accesses not synchronized via some common mutex in his original setting.
Adapted to our setting such weak influences do not exist because the atomicity assumption introduced mutexes $m_g$
immediately surrounding each access to a global $g$. The weak interferences for a global $g$ thereby are stored at unknowns
$[g,m_g,S]$.
To be faithful to the analysis as proposed by Miné where such weak interferences
are only consulted at the read and not eagerly copied into the local state, locking and unlocking
some $m_g$ for $g\in\G$ does not affect the local state and the values stored at unknowns $[g,m_g,S]$
are instead taken into account when reading from or writing to a global.
We also track a set $W$ of written variables by which we restrict synchronized interferences, as is done with the help
of the weak interferences of a thread with a given thread \emph{id} in Miné's original setting.
The right-hand-side functions thus are defined as follows:
\[
\begin{array}{lll}
\init^\sharp\,\_
	&=& \Let\; \sigma = \{ x \mapsto \top \mid x \in \X \} \cup \{ g \mapsto \bot \mid g \in \G \} \;\In \\
	& & (\emptyset,(\emptyset,\sigma)) \\[1ex]
\sem{[u,S],x = \create(u_1)}^\sharp\Get
	&=&	\Let\;(W,\sigma) = \Get\,[u,S]\;\In	\\
	& & \Let\; i = \nu^\sharp\,u\,(W,\sigma)\,u_1\;\In\\
	& & \Let\; \sigma' = \sigma \oplus (\{ \self \mapsto i\} \cup \{ g \mapsto \bot \mid g \in \G\})\;\In\\
	& & \Let\; \rho = \{ [u_1,\emptyset] \mapsto (\emptyset,\sigma') \}\;\In\\
	& &	(\rho,(W,\sigma \oplus \{ x \mapsto i \}))	\\[1ex]
\sem{[u, S], g=x}^\sharp \Get
	&=&	\Let\;(W, \sigma) = \Get\,[u, S]\;\In \\
	& &	\Let\;\sigma' = \sigma \oplus \{ g \mapsto \sigma\,x \}\;\In \\
	& &	(\{ [g,m_g, S \setminus \{m_g\}] \mapsto \sigma'\,g \}, (W \cup \{g\}, \sigma')) \\[1ex]
\sem{[u, S], x=g}^\sharp \Get
	&=&	\Let\;(W, \sigma) = \Get\,[u, S]\;\In \\
	& &	\Let\;g' = \bigsqcup \{ \Get\,[g,m_g, S'] \mid S' \subseteq \M, S' \cap S = \emptyset \}\;\In \\ 
	& &	(\emptyset, (W, \sigma \oplus \{ x \mapsto \sigma\,g \sqcup g' \})) \\[1ex]
\sem{[u,S],\lock(a)}^\sharp\Get
	&=&	\Let\;(W,\sigma) = \Get\,[u,S]\;\In	\\
	& &	\Let\;\sigma' = \{ g \mapsto \bigsqcup \{ \Get[g,a, S'] \mid
	  S' \subseteq \M, S' \cap S = \emptyset \}\\
	&& \qquad \qquad \qquad \mid g \in \G\} \;\In	\\
	& &	(\emptyset, (W, \sigma \sqcup \sigma'))	\\[1ex]
\sem{[u,S],\unlock(a)}^\sharp\Get
	&=&	\Let\;(W,\sigma) = \Get\,[u,S]\;\In	\\
	& &	(\{ [g,a, S \setminus \{m\}] \mapsto \sigma\,g \mid g \in W \}, (W, \sigma))\\[1ex]
\sem{[u,S],\lock(m_g)}^\sharp\Get &=& \Get\,[u,S]\\[1ex]
\sem{[u,S],\unlock(m_g)}^\sharp\Get &=& \Get\,[u,S]
\end{array}
\] for $a \not\in \{ m_g \mid g \in \G \}$.
This is a complicated analysis; however, side-effecting constraint systems elegantly capture
the core idea in just a few lines. The weak interferences are associated with pseudo-lock $m_g$,
but a weak interference is only propagated from a write with lockset $S$ to a read with lockset $S'$
if these sets have an empty intersection. Similarly, synchronized interferences are only propagated from
an unlock to a lock if the ambient locksets permit it.

%% file: main.bbl
\begin{thebibliography}{22}
\providecommand{\natexlab}[1]{#1}
\providecommand{\url}[1]{\texttt{#1}}
\providecommand{\urlprefix}{URL }
\expandafter\ifx\csname urlstyle\endcsname\relax
  \providecommand{\doi}[1]{doi:\discretionary{}{}{}#1}\else
  \providecommand{\doi}{doi:\discretionary{}{}{}\begingroup
  \urlstyle{rm}\Url}\fi

\bibitem[{Alglave et~al.(2011)Alglave, Kroening, Lugton, Nimal, and
  Tautschnig}]{Alglave2011}
Alglave, J., Kroening, D., Lugton, J., Nimal, V., Tautschnig, M.: Soundness of
  data flow analyses for weak memory models. In: APLAS '11, vol. LNCS 7078, pp.
  272--288, Springer (2011), \doi{10.1007/978-3-642-25318-8_21}

\bibitem[{Apinis et~al.(2012)Apinis, Seidl, and Vojdani}]{apinis2012side}
Apinis, K., Seidl, H., Vojdani, V.: Side-effecting constraint systems: a swiss
  army knife for program analysis. In: APLAS '12, pp. 157--172, Springer
  (2012), \doi{10.1007/978-3-642-35182-2_12}

\bibitem[{Brookes(2007)}]{Brookes2007}
Brookes, S.: A semantics for concurrent separation logic. Theoretical Computer
  Science \textbf{375}(1-3), 227--270 (may 2007),
  \doi{10.1016/j.tcs.2006.12.034}

\bibitem[{De et~al.(2011)De, D'Souza, and Nasre}]{De2011}
De, A., D'Souza, D., Nasre, R.: Dataflow analysis for datarace-free programs.
  In: ESOP, vol. LNCS 6602, pp. 196--215, Springer (2011),
  \doi{10.1007/978-3-642-19718-5_11}

\bibitem[{Ferrara(2008)}]{Ferrara08}
Ferrara, P.: Static analysis via abstract interpretation of the happens-before
  memory model. In: TAP '08, vol. LNCS 4966, pp. 116--133, Springer (2008),
  \doi{10.1007/978-3-540-79124-9\_9}

\bibitem[{Gotsman et~al.(2007)Gotsman, Berdine, Cook, and Sagiv}]{Gotsman07}
Gotsman, A., Berdine, J., Cook, B., Sagiv, M.: Thread-modular shape analysis.
  In: PLDI '07, pp. 266--277, ACM (2007), \doi{10.1145/1250734.1250765}

\bibitem[{Kahlon et~al.(2005)Kahlon, Ivan{\v{c}}i{\'{c}}, and
  Gupta}]{Kahlon2005}
Kahlon, V., Ivan{\v{c}}i{\'{c}}, F., Gupta, A.: Reasoning about threads
  communicating via locks. In: CAV '05, vol. LNCS 3576, pp. 505--518, Springer
  (2005), \doi{10.1007/11513988_49}

\bibitem[{Kahlon et~al.(2007)Kahlon, Yang, Sankaranarayanan, and
  Gupta}]{Kahlon2007}
Kahlon, V., Yang, Y., Sankaranarayanan, S., Gupta, A.: Fast and accurate static
  data-race detection for concurrent programs. In: CAV '07, vol. LNCS 4590, pp.
  226--239, Springer (2007), \doi{10.1007/978-3-540-73368-3_26}

\bibitem[{Lamport(1978)}]{Lamport1978}
Lamport, L.: Time, clocks, and the ordering of events in a distributed system.
  Communications of the ACM \textbf{21}(7), 558–565 (1978)

\bibitem[{Min{\'{e}}(2012)}]{Mine2012}
Min{\'{e}}, A.: Static analysis of run-time errors in embedded real-time
  parallel {C} programs. Logical Methods in Computer Science \textbf{8}(1),
  1--63 (mar 2012), \doi{10.2168/LMCS-8(1:26)2012}

\bibitem[{Min{\'{e}}(2014)}]{Mine2014}
Min{\'{e}}, A.: Relational thread-modular static value analysis by abstract
  interpretation. In: VMCAI '14, vol. 8318 LNCS, pp. 39--58, Springer (2014),
  \doi{10.1007/978-3-642-54013-4_3}

\bibitem[{Monat and Min{\'{e}}(2017)}]{Mine2017}
Monat, R., Min{\'{e}}, A.: Precise thread-modular abstract interpretation of
  concurrent programs using relational interference abstractions. In: VMCAI
  '17, vol. 10145 LNCS, pp. 386--404, Springer (2017),
  \doi{10.1007/978-3-319-52234-0_21}

\bibitem[{Mukherjee et~al.(2017)Mukherjee, Padon, Shoham, D'Souza, and
  Rinetzky}]{Mukherjee2017}
Mukherjee, S., Padon, O., Shoham, S., D'Souza, D., Rinetzky, N.: Thread-local
  semantics and its efficient sequential abstractions for race-free programs.
  In: SAS '17, vol. LNCS 10422, pp. 253--276, Springer (2017),
  \doi{10.1007/978-3-319-66706-5_13}

\bibitem[{Nanevski et~al.(2019)Nanevski, Banerjee, Delbianco, and
  F{\'{a}}bregas}]{Nanevski2019}
Nanevski, A., Banerjee, A., Delbianco, G.A., F{\'{a}}bregas, I.: Specifying
  concurrent programs in separation logic: Morphisms and simulations. PACMPL
  \textbf{3}(OOPSLA), 1--30 (oct 2019), \doi{10.1145/3360587}

\bibitem[{Nanevski et~al.(2014)Nanevski, Ley-Wild, Sergey, and
  Delbianco}]{Nanevski2014}
Nanevski, A., Ley-Wild, R., Sergey, I., Delbianco, G.A.: Communicating state
  transition systems for fine-grained concurrent resources. In: ESOP '14, vol.
  LNCS 8410, pp. 290--310, Springer (2014), \doi{10.1007/978-3-642-54833-8_16}

\bibitem[{O'Hearn(2007)}]{OHearn07}
O'Hearn, P.W.: Resources, concurrency, and local reasoning. Theoretical
  Computer Science \textbf{375}(1), 271--307 (2007),
  \doi{10.1016/j.tcs.2006.12.035}

\bibitem[{Sergey et~al.(2015)Sergey, Nanevski, and Banerjee}]{Nanevski2015}
Sergey, I., Nanevski, A., Banerjee, A.: Mechanized verification of fine-grained
  concurrent programs. In: PLDI '15, pp. 77--87, ACM (jun 2015),
  \doi{10.1145/2737924.2737964}

\bibitem[{van Steen and Tanenbaum(2017)}]{DistrAlgos}
van Steen, M., Tanenbaum, A.S.: Distributed Systems. distributed-systems.net,
  3rd edn. (2017)

\bibitem[{Suzanne and Min{\'{e}}(2016)}]{Mine2016}
Suzanne, T., Min{\'{e}}, A.: From array domains to abstract interpretation
  under store-buffer-based memory models. In: SAS '16, vol. LNCS 9837, pp.
  469--488, Springer (2016), \doi{10.1007/978-3-662-53413-7_23}

\bibitem[{Suzanne and Min{\'{e}}(2018)}]{Mine2018}
Suzanne, T., Min{\'{e}}, A.: Relational thread-modular abstract interpretation
  under relaxed memory models. In: APLAS '18, vol. LNCS 11275, pp. 109--128,
  Springer (dec 2018), \doi{10.1007/978-3-030-02768-1_6}

\bibitem[{Vojdani(2010)}]{Vojdani2010}
Vojdani, V.: Static Data Race Analysis of Heap-Manipulating {C} Programs. Ph.D.
  thesis, University of Tartu. (December 2010)

\bibitem[{Vojdani et~al.(2016)Vojdani, Apinis, R\~{o}tov, Seidl, Vene, and
  Vogler}]{Vojdani2016}
Vojdani, V., Apinis, K., R\~{o}tov, V., Seidl, H., Vene, V., Vogler, R.: Static
  {Race} {Detection} for {Device} {Drivers}: {The} {Goblint} {Approach}. In:
  ASE '16, pp. 391--402, ACM (2016), \doi{10.1145/2970276.2970337}

\end{thebibliography}
